\documentclass{article}
\usepackage{microtype}
\usepackage{graphicx}
\usepackage{amsmath}
\usepackage{subcaption}
\usepackage{booktabs}
\usepackage{hyperref}
\usepackage{lipsum}
\usepackage{bbm}

\def\vn{{\mathbf{n}}}
\def\vx{{\mathbf{x}}}
\def\tx{{\tilde{\mathbf{x}}}}
\def\vy{{\mathbf{y}}}
\def\vz{{\mathbf{z}}}
\def\vI{{\mathbf{I}}}
\def\vR{{\mathbf{R}}}
\def\vN{{\mathbf{N}}}
\def\vmu{{\mathbf{\mu}}}
\def\veps{{\boldsymbol{\epsilon}}}

\def\mI{\mathbf{I}}
\def\mSig{{\mathbf{\Sigma}}}
\def\mA{{\mathbf{A}}}
\def\mS{{\mathbf{S}}}
\def\mH{{\mathbf{H}}}
\def\mF{{\mathbf{F}}}
\def\mLambda{{\mathbf{\Lambda}}}

\def\R{{\mathbb{R}}}
\def\N{{\mathbb{N}}}
\def\cK{\mathbf{K}}
\def\cL{\mathbf{L}}
\def\sS{\mathbf{S}}

\def\fR{\mathcal{R}}
\def\fF{\mathcal{F}}
\def\fJ{\mathcal{J}}
\def\fP{\mathcal{P}}

\def\eE{\mathbb{E}}
\def\nN{\mathcal{N}}

\def\Id{{{\operatorname{Id}}}}
\def\Com{{\operatorname{Com}}}
\DeclareMathOperator*{\argmax}{arg\,max}
\DeclareMathOperator*{\argmin}{arg\,min}
\newcommand{\Prb}[1]{\mathbb{P}\left(#1\right)}

\usepackage[accepted]{icml2024}

\usepackage{amsmath}
\usepackage{amssymb}
\usepackage{mathtools}
\usepackage{amsthm}

\usepackage[capitalize,noabbrev]{cleveref}

\theoremstyle{plain}
\newtheorem{theorem}{Theorem}[section]
\newtheorem{proposition}{Proposition}[section]

\theoremstyle{definition}
\newtheorem*{definition}{Definition}
\newtheorem{assumption}[theorem]{Assumption}
\theoremstyle{remark}

\usepackage[textsize=tiny]{todonotes}

\icmltitlerunning{Plug-and-Play image restoration with Stochastic deNOising REgularization}

\begin{document}

\twocolumn[
\icmltitle{{Plug-and-Play image restoration with Stochastic deNOising REgularization}}

\begin{icmlauthorlist}
\icmlauthor{Marien Renaud}{1}
\icmlauthor{Jean Prost}{1}
\icmlauthor{Arthur Leclaire}{2}
\icmlauthor{Nicolas Papadakis}{1}
\end{icmlauthorlist}

\icmlaffiliation{1}{Institut de Mathématiques de Bordeaux (IMB), France}
\icmlaffiliation{2}{Télécom Paris, IP Paris, Palaiseau, France}

\icmlcorrespondingauthor{Marien Renaud}{marien.renaud@math.u-bordeaux.fr}

\icmlkeywords{Machine Learning, ICML}

\vskip 0.3in
]
\printAffiliationsAndNotice{}

\begin{abstract}
Plug-and-Play (PnP) algorithms are a class of iterative algorithms that address image inverse problems by combining a physical model and a deep neural network for regularization. 
Even if they produce impressive image restoration results, these algorithms rely on a non-standard use of a denoiser on images that are less and less noisy along the iterations, which contrasts with recent algorithms based on Diffusion Models (DM), where the denoiser is applied only on re-noised images.
We propose a new PnP framework, called Stochastic deNOising REgularization (SNORE), 
which applies the denoiser only on images with noise of the adequate level.
It is based on an explicit stochastic regularization, which leads to a stochastic gradient descent algorithm to solve ill-posed inverse problems.
A convergence analysis of this algorithm and its annealing extension is provided. 
Experimentally, we prove that SNORE is competitive with respect to state-of-the-art methods on deblurring and inpainting tasks, both quantitatively and qualitatively.
\end{abstract}

\section{Introduction}

Many imaging problems can be formulated as \textit{inverse problems} seeking to recover high-quality images $\vx^*$ from their low-quality observations $\vy$ by solving a problem of the form
\begin{align}\label{eq:minimization_problem}
    \vx^* \in \argmin_{\vx \in \R^d}{\fF(\vx, \vy)} + \alpha \fR(\vx),
\end{align}
where $\fF$ measures the fidelity to the degraded observation~$\vy$ and $\fR$ is a regularization term weighted by a parameter $\alpha > 0$.

The regularization is crucial to complete the missing information of the observation by bringing prior knowledge on the high-quality image. 
Choosing a relevant prior $\fR$ in Problem~\eqref{eq:minimization_problem} is a long-standing~\citep{RUDIN1992259, mallat1999, zoran2011} and challenging task and recent approaches explore deep learning techniques that learn a prior from a database of clean images~\citep{zhang2017learning, lunz2018adversarial, laumont2022maximumaposteriori}. 

Problem~\eqref{eq:minimization_problem} can be addressed with proximal splitting algorithms~\citep{Boyd2011} which are first-order optimization algorithms based on the recursive application of gradient-descent and/or proximal operators of functions $\fF$ and $\fR$. 

The Plug-and-Play (PnP) framework~\cite{venkatakrishnan2013plug} consists in replacing, within a proximal splitting algorithm, the proximal step on the regularization $\fR$ by a denoising operation; and it allows to use implicit regularization priors encoded by pre-trained image denoisers.  
Similarly, following the Regularization by Denoising (RED) framework~\citep{romano2017little}, a gradient-descent step on the regularization can be substituted by a learned denoiser. It has been 
observed that plugging a pre-trained state-of-the-art deep denoiser is essential for achieving the best quality results in many imaging contexts~\citep{Metzler.etal2018, Ryu2019, hurault2022proximal, renaud2023plugandplay, Ulondu2023}.

In another line of works, inverse problems solvers based on denoising diffusion models (DDM)~\cite{ho2020denoising, song2021scorebased} have demonstrated their ability to provide high-quality restoration even for severely ill-posed problems where a large amount of information is missing~\citep{chung2023diffusion, song2023pseudoinverseguided}. 
DDM and PnP both rely on deep denoisers to implicitly model the prior distribution, and they decouple prior and  data-fidelity terms in order to provide flexible solvers. 
However, while the theoretical properties of PnP algorithms regarding convergence have been studied in depth~\cite{sreehari2016plug, sun2019online, gavaskar2019proof, Ryu2019, cohen2021regularization,hurault2022gradient}, the theoretical properties of restoration algorithms based on DDM remain poorly understood.
Indeed, DDM methods rely on heuristics to approximate the score of the intractable likelihood. To the best of our knowledge, the impact of the approximation error on the generated samples remains to be quantified.

In this work, we seek to develop an inverse problem solver that inherits the superior restoration quality of DDM methods, while satisfying the theoretical guarantees of convergence that are met for some PnP algorithms.
Our key observation is that, contrary to the PnP framework, in diffusion based methods the denoiser is applied on a \emph{noisy version of the image} at each iteration of the sampling process. We postulate that applying the denoiser to noisy images is a main ingredient to the impressive performance of diffusion-based samplers, as it avoids a domain shift relative to the data on which the denoiser was trained.

This motivates us to propose SNORE (Stochastic deNOising REgularization), a stochastic PnP algorithm, which differs from classical PnP schemes by injecting noise to the input of the denoiser at each iteration.
SNORE minimizes a classical variational objective, where the regularization term is defined as the average value of the smoothed log prior on noisy version of the image of interest, and can be viewed as a relaxed version of the usual negative log-prior.

\textbf{Contributions.} \textbf{(a)}~We propose a new explicit regularization leading to a novel PnP framework, named \textit{Stochastic deNOising Regularization (SNORE)}, in which 
the denoiser is applied on a noisy version of the image at each iteration. 
\textbf{(b)}~We show that SNORE can be optimized by a stochastic gradient-descent algorithm (Algorithm~\ref{alg:Average_PnP}). We prove that this algorithm converges with the exact MMSE denoiser (Proposition~\ref{prop:convergence_unbiased}) and we bound the error with an inexact MMSE denoiser (Proposition~\ref{prop:biais_convergence}). 
\textbf{(c)}~With a critical point analysis (Proposition~\ref{prop:critical_point_accumulate_values}), we motivate the practical use of an annealed algorithm (Algorithm~\ref{alg:Annealead_SNORE}). Finally, we demonstrate the efficiency of SNORE to solve inverse problems.

\section{Stochastic deNOising REgularization (SNORE)}

In this section, we propose a new stochastic regularization, SNORE (Equation~\ref{eq:new_regularization}), that can be used for PnP restoration in such a way that the denoising step now applies to a noisy version of the current image.

We first recall the RED regularization (Section~\ref{sec:background}) and the corresponding PnP algorithm (Algorithm~\ref{alg:RED}).
Then we introduce the SNORE regularization leading to Algorithm~\ref{alg:Average_PnP}. 
We also propose an Annealing SNORE algorithm (Algorithm~\ref{alg:Annealead_SNORE}) inspired by annealed importance sampling~\cite{neal1998annealed}. Finally, we discuss the positioning of our method in relation to existing related works (Section~\ref{sec:related_works}).

\subsection{Background}\label{sec:background}
\textbf{Bayesian inverse problem}
An inverse problem formulated as in Problem~\eqref{eq:minimization_problem} has a general Bayesian interpretation. From the observation $\vy \in \R^q$ (typically $q < d$), we can restore the image by computing the Maximum A Posteriori (MAP) estimator $\hat{\vx}$ defined by
\begin{align*}
    \hat{\vx} &= \argmax_{\vx \in \R^d}{p(\vx | \vy)}
    = \argmin_{\vx \in \R^d}{- \log p(\vx | \vy)} \\
    &= \argmin_{\vx \in \R^d}{\underbrace{- \log p(\vy | \vx)}_{=\fF(\vx, \vy)} \underbrace{- \log p(\vx)}_{=\fR(\vx)}}.
\end{align*}
Thus the data-fidelity $\fF$ is related to the image forward model, and the regularization $\fR$ reflects the adopted prior model on images (which can be improper).
Adding a weighting parameter $\alpha > 0$ (see Problem~\ref{eq:minimization_problem}) is equivalent to adding a temperature parameter on the prior $p$, which becomes $p^{\alpha}$.

\textbf{Data Fidelity} The forward model is supposed to have a known form
$$\vy = \mathcal{A}(\vx) + \vn,$$
with the degradation operator $\mathcal{A} : \R^d \mapsto \R^q$, and the noise $\vn \sim \nN(0, \sigma_{\vy}^2 \mI_q)$ where $\sigma_{\vy} >0$. Then $\fF(\vx, \vy) = \frac{1}{\sigma_{\vy}^2}\| \mathcal{A}(\vx) - \vy\|^2$.

\textbf{Deep Learning regularization}
With a Bayesian interpretation, the regularization $\fR = - \log p$ defines a model on the data. 
Recently, Deep Neural Networks (DNN) have proved their effectiveness in learning a realistic model from a database of observations.
The RED framework~\cite{romano2017little} uses the performance of DNN for image restoration. It consists in adopting a prior regularization $\fR$ whose gradient $\nabla \fR$ is given by a pre-trained denoiser.
This implicit relation relies on the regularization defined by
\begin{equation}\label{eq:prior_approx}
    \fR(\vx) \approx \fP_{\sigma}(\vx) := -\log p_{\sigma}(\vx),
\end{equation}
where $p_{\sigma}$ is the convolution $p \ast \nN_{\sigma}$ between $p$ and  $\nN_{\sigma} = \nN(0,\sigma^2 \mI_d)$.
The relation between denoising and regularization is  made explicit with Tweedie's formula~\citep{Efron2011}:
\begin{equation}\label{eq:tweedie}
    \nabla \fP_{\sigma}(\vx) = - \frac{1}{\sigma^2} \left( D_{\sigma}^\star(\vx) - \vx \right),
\end{equation}
where $D_{\sigma}^\star$ is the Minimum Mean Square Error (MMSE) denoiser is defined by
\begin{equation}
\label{eq:exact_mmse}
    D_{\sigma}^\star(\tx) := \eE [\vx | \tx ] = \int_{\R^d}{\vx p_{\vx|\tx}(\vx | \tx) d\vx},
\end{equation}
for $\tx = \vx + \veps \text{ with } \vx \sim p(\vx), \veps \sim \nN(0, \sigma^2 \mI_d)$.

In practice, we do not have access to the exact MMSE denoiser $D_{\sigma}^{\star}$, but only to a deep denoiser $D_{\sigma}$ that is trained to approximate the MMSE $D_{\sigma}^{\star}$. Then a gradient descent scheme (Algorithm~\ref{alg:RED}), as described by~\citet{Reehorst2019}, can be run to obtain an approximate solution of Problem~\eqref{eq:minimization_problem}.

\begin{algorithm}
\caption{RED~\citep{romano2017little}}\label{alg:RED}
\begin{algorithmic}[1]
\STATE \textbf{Param.:} $\vx_0 \in \R^d$, $\sigma > 0$, $\alpha > 0$, $\delta > 0$, $N \in \N$
\STATE \textbf{Input:} degraded image $\vy$
\STATE \textbf{Output:} restored image $\vx_{N}$
\FOR{$k = 0, 1, \dots, N-1$}
    \STATE $\vx_{k+1} \gets \vx_k - \delta \nabla \fF(\vx_k, \vy) - \frac{\alpha \delta}{\sigma^2} \left(\vx_k - D_{\sigma}(\vx_k) \right)$ 
\ENDFOR
\end{algorithmic}
\end{algorithm}

Algorithm~\ref{alg:RED} involves the computation of $D_{\sigma}(\vx_k)$, in which the denoiser is applied to an image iterate $\vx_k$ that is not necessarily noisy. As a denoiser is trained to denoise images with noise, the application of $D_\sigma$ to images  that are out of the training domain might be irrelevant. 
To bypass this issue, we propose a new regularization.

\subsection{SNORE regularization}\label{sec:average_pnp_intro}
We propose the SNORE regularization $\fR_{\sigma}$, whose gradient applies the MMSE denoiser on noisy images. This new regularization is defined by
\begin{align}\label{eq:new_regularization}
    \fR_{\sigma}(\vx) &= - \eE_{\tx \sim p_{\sigma}(\tx|\vx)}\left(  \log p_{\sigma}(\tx) \right) \\
    \nabla_{\vx} \fR_{\sigma}(\vx)
    &= - \frac{1}{\sigma^2} \left(\eE_{\tx \sim p_{\sigma}(\tx|\vx)}\left(  D_{\sigma}^{\star}(\tx) \right) - \vx \right).\label{eq:gradient_average_pnp_reg}
\end{align}

Minimizing $\fR_{\sigma}(\vx)$ is equivalent to maximizing 
$\eE_{\tx \sim p_{\sigma}(\tx|\vx)}\left(\log p_{\sigma}(\tx) \right)$. 
The last quantity is maximum if noisy versions of $\vx$ are highly probable in the noisy prior distribution, $p_{\sigma}(\vx)$. In other words: \textit{An image looks clean if its noisy versions look as noisy images}.

SNORE regularization can be seen as a relaxation of the classical PnP regularization $- \log (p \ast \nN_{\sigma})$, following the idea of~\citet{scarvelis2023closed}.
\begin{align}\label{eq:new_regularization_convolution}
    \fR_{\sigma}(\vx)= - \left(\log (p \ast \nN_{\sigma})\ast \nN_{\sigma}\right)(\vx).
\end{align}
    
In Appendix~\ref{sec:simple_cases_apnp}, we prove that $\fR_{\sigma}$ provides the same minimum than $\fP_{\sigma}$ if the prior is Gaussian. We also detail the case of Gaussian Mixture prior, with the convergence analysis of $\nabla \fP_{\sigma}$ and $\nabla \fR_{\sigma}$ to $-\nabla \log p$ when $\sigma \to 0$ and with a 1D illustration of the difference between $\nabla \fR_{\sigma}$ and $-\nabla \log p_{\sigma}$.

\paragraph{Interpretation of the SNORE regularization}
We first underline that $\fR_{\sigma}$ can be re-written as
\begin{align}
    &\fR_{\sigma}(\vx) = - \eE_{\tx \sim p_{\sigma}(\tx|\vx)}\left(  \log p_{\sigma}(\tx) \right)\nonumber \\
    &= \mathcal{KL}(p_{\sigma}(\tx|\vx)\|p_{\sigma}(\tx)) - \eE_{\tx \sim p_{\sigma}(\tx|\vx)}\left(  \log p_{\sigma}(\tx|\vx) \right)\nonumber\\
    &=\mathcal{KL}(p_{\sigma}(\tx|\vx)\|p_{\sigma}(\tx)) + C,\label{eq:regularization_as_kl}
\end{align}

where we introduced the Kullback-Leibler divergence 
$\mathcal{KL}(\mu \| \nu) := \int_{\R^d}{\log\left(\frac{d\mu}{d\nu}\right) d\mu}$ 
and  the constant 

\begin{align*}
 C= - \eE_{\tx \sim p_{\sigma}(\tx|\vx)}\left(  \log p_{\sigma}(\tx|\vx) \right) = \frac{d}{2} \left( 1 + \log \left( 2\pi\sigma^2\right)\right).
\end{align*}

Hence the potential $\fR_{\sigma}(\vx)$ has the same optimization profile than $\mathcal{KL}(p_{\sigma}(\tx|\vx)\|p_{\sigma}(\tx))$. This last quantity leads to another interpretation. Minimizing $\fR_{\sigma}(\vx)$ is equivalent to \textit{find the Gaussian mode $p_{\sigma}(\tx|\vx)$ of standard deviation~$\sigma$ that best approximates the noisy prior distribution $p_{\sigma}(\tx)$ in terms of KL divergence}.

\paragraph{Optimization algorithms}
With SNORE regularization, we solve the following optimization problem to restore an image
\begin{align}
\label{eq:opt_problem_sigma}
    \argmin_{\vx \in \R^d}{\fJ(\vx)  := \fF(\vx, \vy) + \alpha \fR_{\sigma}(\vx)}.
\end{align}

Due to the formulation of $\fR_{\sigma}$ as an expectation, we  implement\footnote{Note that a possible stochastic gradient of $\fR_{\sigma}$ can be $\frac{1}{\sigma^2} \left(\tx - D_{\sigma}(\tx) \right)$. We choose to only add noise in the denoiser to reduce the residual noise on the image.}  a stochastic gradient descent algorithm (Algorithm~\ref{alg:Average_PnP}) to solve Problem~\eqref{eq:opt_problem_sigma}.

\begin{algorithm}
\caption{SNORE}\label{alg:Average_PnP}
\begin{algorithmic}[1]
\STATE \textbf{Param.:} init. $\vx_0 \in \R^d$, $\sigma > 0$, $\alpha > 0$, $\delta > 0$, $N \in \N$
\STATE \textbf{Input:} degraded image $\vy$
\STATE \textbf{Output:} restored image $\vx_{N}$
\FOR{$k = 0, 1, \dots, N-1$}
    \STATE $\veps \gets \mathcal{N}(0, \mI_d)$
    \STATE $\tx_k \gets \vx_k + \sigma \veps$
    \STATE $\vx_{k+1} \gets \vx_k - \delta \nabla \fF(\vx_k, \vy) - \frac{\alpha \delta}{\sigma^2} \left(\vx_k - D_{\sigma}(\tx_k)\right)$ 
\ENDFOR
\end{algorithmic}
\end{algorithm}

\begin{algorithm}
\caption{Annealed SNORE}\label{alg:Annealead_SNORE}
\begin{algorithmic}[1]
\STATE \textbf{Param.:} init. $\vx_0 \in \R^d$, $\delta > 0$, annealing schedule $m \in \N$, $\sigma_0 > \dots > \sigma_{m-1} \approx 0$, $\alpha_0, \dots, \alpha_{m-1} > 0$, $N_0, \dots, N_{m-1} \in \N$
\STATE \textbf{Input:} degraded image $\vy$
\STATE \textbf{Output:} restored image $\vx_{N_{m-1}}$
\FOR{$i = 0, 1, \dots, m-1$}
    \FOR{$k = 0, 1, \dots, N_i-1$}
        \STATE $ \veps \gets \mathcal{N}(0, \mI_d)$
        \STATE $\tx_k \gets \vx_k + \sigma_i \veps$
        \STATE $\vx_{k+1} \gets \vx_k - \delta \nabla \fF(\vx_k, \vy) - \frac{\alpha_i \delta}{\sigma_i^2} \left(\vx_k - D_{\sigma_i}(\tx_k) \right)$ 
    \ENDFOR
\ENDFOR
\end{algorithmic}
\end{algorithm}

Inspired by annealed importance sampling~\cite{neal1998annealed} and the recent use of such a decreasing of $\sigma$ in diffusion model~\cite{sun2023provable}, we also develop an Annealed SNORE Algorithm (Algorithm~\ref{alg:Annealead_SNORE}). This algorithm, which proves more efficient in practice, is supported by a critical point analysis (Proposition~\ref{prop:critical_point_accumulate_values}).

\subsection{Related Works}\label{sec:related_works}
\paragraph{Other Stochastic Plug-and-Play algorithms}
In the existing literature, stochastic versions of Plug-and-Play have already been proposed. 
Most of these works intend to accelerate the computation by a stochastic mini-batch approximation on the data-fidelity~\citep{tang2020fast} or the regularization~\citep{sun2019block}. On the other hand, SNORE does not aim at accelerating PnP algorithms but it proposes a stochastic improvement of PnP by injecting noise inside the classical PnP regularization.

\citet{laumont2022maximumaposteriori} propose to run a stochastic gradient descent algorithm (PnP SGD) with the PnP regularization. Contrary to PnP SGD, SNORE injects the noise inside the denoiser only, and not in the data-fidelity term. Moreover, in SNORE, the standard deviation of the injected noise is fixed  (for fixed $\sigma$). 

Another line of works target image restoration by sampling the posterior law instead of solving Problem~\ref{eq:minimization_problem}.. This can be done within a PnP framework using the Tweedie formula (Equation~\ref{eq:tweedie}) to compute a Langevin dynamic~\citep{Laumont_2022_pnpula, renaud2023plugandplay} or a Gibbs sampling~\citep{coeurdoux2023plugandplay, Bouman2023}.

\paragraph{Link with diffusion based method}

Denoising diffusion models (DDM) are a class of generative models that can generate images by gradually transforming noise into data with deep denoising networks~\cite{ho2020denoising, song2021scorebased}.
A key feature of DDMs is that, for the adequate weighting schedule $w_t$, the log-likelihood of the generative model $p_{\theta}(\vx)$ is lower-bounded by a (negative) mixture of denoising losses at different noise levels~\citep{ho2020denoising, song2021maximum}: 
\begin{equation}
    \log p_{\theta}(\vx) 
    \geq 
    \underbrace{
    -\mathbb{E}_{t, \tx \sim p_{\sigma}(\tx|\vx)} \left[w_t \|\vx - D_{\sigma_t}(\tx)\|^2 \right]
    }_{\mathcal{L}\left( \vx\right)} 
    \label{eq:elbo_ddm}
\end{equation}
In order to use DDMs for regularizing inverse problems, several works have proposed to replace the intractable log-likelihood by its lower-bound~\eqref{eq:elbo_ddm}~\cite{poole2022dreamfusion, wang2023score, feng2023efficient, mardani2023variational}. In particular, the gradient of the lower-bound is:
\begin{equation}
    \nabla_{\vx}\mathcal{L}\left( \vx\right) = \mathbb{E}_{t, \tx \sim p_{\sigma}(\tx|\vx)} \left[\frac{w_t}{2} J_{D_{\sigma_t}}^ \top\left(\vx - D_{\sigma_t}(\tx)\right)^2 \right] \label{eq:elbo_ddm_grad}
\end{equation}
where $J_{D_{\sigma_t}}$ is the Jacobian matrix of the denoiser with respect to the input. Hence, the gradient of DDMs lower-bound~\eqref{eq:elbo_ddm_grad} has a similar formulation than the gradient of our regularization function~\eqref{eq:gradient_average_pnp_reg}, with the difference that it includes the Jacobian of the denoiser, and it is averaged over multiple noise levels $\sigma_t$. It has been found  that removing the Jacobian in~\eqref{eq:elbo_ddm_grad} yields better results in practice~\cite{poole2022dreamfusion}. Several works propose theoretical justifications to omit the Jacobian matrix~\citep{wang2023score, mardani2023variational}, by assuming that the denoising network provides the exact score of some prior function. Although we rely on a similar assumption, our theoretical analysis in Section~\ref{sec:biased_cvg} also covers the case of an imperfect denoiser.

A different approach for solving inverse problem with a DDM prior is to guide the generative process of an unconditional DDM to generate images consistent with an observation $\vy$~\cite{song2022solving, kawar2022denoising, chung2022improving, chung2023diffusion, song2023pseudoinverseguided, luther2023ddgm, Zhu_2023_CVPR}. Notice that~\citet{luther2023ddgm} proposed an algorithm similar to Algorithm~\ref{alg:Annealead_SNORE} but do not provide strong theoretical motivation or analysis.
Such reverse diffusion processes involve a gradual decrease of the strength of the denoising network, analogous to our annealing procedure. 
However, those approaches aim at sampling from the posterior distribution of the inverse problem, whereas we adopt a (stochastic) optimization perspective.
Despite their impressive practical results, DDM guided rely on heuristics to approximate the intractable likelihood model on noisy data. The impact of the approximation error on the distribution of generated samples remains to be quantified.

\section{Convergence Analysis}\label{sec:convergence_analysis}

In this section, we provide a theoretical analysis of our regularization SNORE and a convergence analysis of the associated algorithm.
Problem~\eqref{eq:opt_problem_sigma} is non-convex due to our  regularization. Hence, in the best-case scenario, one can only expect a convergence of the algorithm towards a critical point of the target functional $\fJ$, defined in relation~\eqref{eq:opt_problem_sigma}. Note that all global and local minima of $\fJ$  are critical points of~$\fJ$.
Based on the existing literature~\citep{tadic2017asymptotic, laumont2022maximumaposteriori}, we analyze our stochastic gradient descent in this challenging non-convex context.

We first show (Section~\ref{sec:critical_point}) that  our regularization $\fR_{\sigma}$ is a relevant approximation of the ideal regularization $-\log p$. Then we analyze  the asymptotic behavior of the critical points of Problem~\eqref{eq:opt_problem_sigma} which motivates the annealing Algorithm~\ref{alg:Annealead_SNORE}. Next we  prove, at fixed $\sigma$, the convergence of Algorithm~\ref{alg:Average_PnP} to a critical point of Problem~\eqref{eq:opt_problem_sigma} in the case of using the exact MMSE denoiser $D_{\sigma}^\star$ (Section~\ref{sec:unbiased_cvg}). 
In Section~\ref{sec:biased_cvg}, we quantify the error of Algorithm~\ref{alg:Average_PnP} with an inexact denoiser $D_{\sigma}$. Proofs can be found in Appendix~\ref{sec:proof_convergence_result}. 
A discussion on technical assumptions is given in Appendix~\ref{sec:discussion_assumptions}.

\subsection{Asymptotics of critical points when $\sigma \to 0$}\label{sec:critical_point}
Inspired by~\citet[Proposition 1]{laumont2022maximumaposteriori}, we first demonstrate that $\nabla \fR_{\sigma}$ converges uniformly to $-\nabla \log p$ on every compact when $\sigma \rightarrow 0$. This type of result requires technical assumptions, such as $\nabla \log p$ to be defined everywhere and smooth. 

\begin{assumption}\label{ass:prior_regularity}
    \textbf{(a)} The prior distribution $p \in \mathrm{C}^1(\R^d, ]0, +\infty[)$ with $\|p\|_{\infty} + \|\nabla p\|_{\infty} < + \infty$. \textbf{(b)} The prior score is sub-polynomial, there exist $A \in \R^+$ and $q \in \N$ such that $\forall \vx \in \R^d$, $\|\nabla \log p(\vx)\| \le A (1 + \|\vx\|^q)$.
\end{assumption}
Assumption~\ref{ass:prior_regularity}\textbf{(a)} ensures that the prior is smooth, non-degenerate and Lipschitz.

\begin{assumption}\label{ass:prior_score_approx}
    The noisy prior score is sub-polynomial, there exist $B\in \R^+$, $\beta \in \R$ and $r \in \N$ such that $\forall \vx \in \R^d$, $\|\nabla \log p_{\sigma}(\vx)\| \le B \sigma^{\beta} (1 + \|\vx\|^r)$.
\end{assumption}
Under the so-called manifold hypothesis~\cite{debortoli2023convergence} (see Assumption~\ref{ass:manifold}), Assumption~\ref{ass:prior_score_approx} is verified with $r = 1$ and $\beta = -2$. Assumption~\ref{ass:prior_score_approx} has also been proved with $r = 1$ and $\beta = 0$ in \citep{debortoli2023diffusion}, under the Assumption~\ref{ass:prior_regularity}\textbf{(b)} with $q = 1$ and the fact that there exist $m_0 > 0$ and $d_0 \ge 0$ such that $\forall \vx \in \R^d$, $\langle \nabla \log p(\vx), \vx  \rangle \le - m_0 \|\vx\|^2 + d_0 \|\vx\|$.

\begin{proposition}\label{prop:score_approx_apnp}
     Under Assumptions~\ref{ass:prior_regularity}-\ref{ass:prior_score_approx}, for $\cK$ a compact of $ \R^d$, $\nabla \fR_{\sigma}$ converges uniformly to $-\nabla \log p$ on $\cK$, 
    \begin{equation*}
        \lim_{\sigma \to 0} \sup_{\cK} \|\nabla \fR_{\sigma} + \nabla \log p\| = 0.
    \end{equation*}
\end{proposition}

Proposition~\ref{prop:score_approx_apnp} proves that our score $\nabla \fR_{\sigma}$ is close to the ideal score $-\nabla \log p$ when $\sigma \to 0$.
With this uniform approximation result, we are now able to study the behavior of the critical points of our optimization problem when $\sigma \to 0$. 
For $\cK$ a compact of $\R^d$, we define $\mathbf{S}_{\sigma, \cK} = \{\vx \in \cK | \nabla \fF(\vx, \vy) + \alpha \nabla \fR_{\sigma}(\vx) = 0\}$, the set of critical points of Problem~\eqref{eq:opt_problem_sigma} in $\cK$. In order to study the behavior of the critical points set $\sS_{\sigma, \cK}$ when $\sigma \to 0$, we define below the notion of limit for sets $\sS_{\sigma, \cK}$ when $\sigma \to 0$. To do so, we first introduce cluster points of sets.
\begin{definition}
    For a sequence of sets $(\sS_k)_{k \in \N} \in \left(\R^d\right)^{\N}$, $\vz$ is called a cluster point of these sets 
    if any neighborhood of $\vz$ is visited infinitely often by $(\sS_k)$, 
    i.e. $\forall \epsilon > 0$, $\forall k_0 \in \N$, there exist $k \ge k_0$, and $\vx_k \in \sS_k$ such that $\|\vx_k - \vz\| \le \epsilon$.
\end{definition}
We now apply this definition of cluster points of sets for a decreasing sequence of $\sigma > 0$.
For $\mathbf{E} = \{(\sigma_n)_{n\in \N} \in (\R)^{\N} | \forall n \in \N, \sigma_n > 0, \sigma_n \text{ decreases to }0\}$, and for $\boldsymbol{\sigma} = (\sigma_n)_{n\in \N} \in \mathbf{E}$, we define the cluster points of the sequence of set $(\mathbf{S}_{\sigma_n, \cK})_{n \in \N}$ by $\mathbf{S}_{\boldsymbol{\sigma}, \cK} = \{ \vx \in \cK | \forall \epsilon> 0, \forall m \in \N, \exists n \ge m, \vz_n \in \mathbf{S}_{\sigma_n, \cK}, \|\vx - \vz_n\| \le \epsilon\}$. 
Finally, we can define a limit of sets $\sS_{\sigma, \cK}$ for a continuous $\sigma \to 0$ with $\mathbf{S}_{\cK} = \cup_{\boldsymbol{\sigma} \in \mathbf{E}}{\mathbf{S}_{\boldsymbol{\sigma}, \cK}}$.
Our target points are the critical points of Problem~\eqref{eq:minimization_problem}, $\mathbf{S}^{\star}_{\cK} = \{\vx \in \cK | \nabla \fF(\vx) + \alpha \nabla \fR(\vx) = 0\}$ where $\fR = - \log p$. The following proposition finally establishes that the limit of set $\sS_{\sigma, \cK}$ (in the sense of cluster point) is included in the targeted points.

\begin{proposition}\label{prop:critical_point_accumulate_values}
    Under Assumptions~\ref{ass:prior_regularity}-\ref{ass:prior_score_approx}, for $\cK$ a compact subset of $\R^d$, we have
    \begin{equation*}
        \mathbf{S}_{\cK} \subseteq \mathbf{S}^{\star}_{\cK}.
    \end{equation*}
\end{proposition}
Proposition~\ref{prop:critical_point_accumulate_values} means that a sequence of computed critical points with $\sigma >0$ has all its cluster points in $\mathbf{S}^{\star}_{\cK}$, the set of critical points of the ideal optimization problem~\eqref{eq:minimization_problem}. This result suggests that the annealed algorithm may converge to a point of $\mathbf{S}^{\star}_{\cK}$. In fact annealing~\cite{neal1998annealed} consists in successively approximating critical points of $\mathbf{S}_{\sigma_i, \cK}$ for a decreasing sequence $\sigma_0 > \dots > \sigma_{m-1} \approx 0$.
Proposition~\ref{prop:critical_point_accumulate_values} thus motivates Algorithm~\ref{alg:Annealead_SNORE}, which will be proved efficient in practice.

\subsection{Unbiased algorithm analysis}\label{sec:unbiased_cvg}
In this section, we prove the convergence of the SNORE Algorithm~\ref{alg:Average_PnP} run with the exact MMSE denoiser~$D_{\sigma}^\star$. With this denoiser, an iteration of the algorithm is computed by
\begin{align}\label{eq:theoretical_process}
    \vx_{k+1} = \vx_k - \delta_k \nabla \fF(\vx_k, \vy) - \alpha \delta_k \nabla \log p_{\sigma}(\vx_k + \zeta_k),
\end{align}
with $\zeta_k \sim \mathcal{N}(0,\sigma^2\mI_d)$ and $(\delta_k)_{k \in \N} \in {\left(\R^+\right)}^{\N}$ the decreasing sequence of step-sizes. Algorithm~\ref{alg:Average_PnP} is a stochastic descent algorithm that solves Problem~\eqref{eq:opt_problem_sigma}.

One can note that the stochastic gradient estimation is unbiased. Indeed, by defining
\begin{align*}
    f(\vx, \zeta) = \nabla \fF(\vx, \vy) + \alpha \nabla \log p_{\sigma}(\vx + \zeta),
\end{align*}
we verify that
\begin{align*}
    \eE_{\zeta \sim \mathcal{N}(0, \sigma^2 \mI_d)}{\left(f(\vx, \zeta)\right)} = \nabla \fF(\vx, \vy) + \alpha \nabla \fR_{\sigma}(\vx).
\end{align*}

We make our convergence analysis based on previous studies on stochastic gradient algorithm~\citep[Corollary 6.7]{Benaim1999}, ~\citep[Section II-D]{Metivier1984} or~\cite{tadic2017asymptotic}. 
An assumption on the step-size decrease is required to ensure convergence.
\begin{assumption}
    \label{ass:step_size_decreas}
    \itshape The step-size decreases to zero but not too fast: $\sum_{k = 0}^{+\infty}{\delta_k} = + \infty$ and $\sum_{k = 0}^{+\infty}{\delta_k^2} < + \infty$.
\end{assumption}
Assumption~\ref{ass:step_size_decreas} guides the choice of the step-size rule to ensure convergence, for instance $\delta_k = \frac{\delta}{k^{a}}$ with $a \in ]\frac{1}{2},1]$.

\begin{assumption}
    \label{ass:data_fidelity_reg}
    \itshape The data-fidelity term $\mathcal{F}_y : \vx \in \R^d \mapsto \mathcal{F}(\vx, \vy)$ is $\mathcal{C}^{\infty}$. 
\end{assumption}
This assumption is typically verified for a data-fidelity term $\fF(\vx, \vy) = \frac{1}{2 \sigma_{\vy}^2} \|\vy - \mathcal{A} \vx \|^2$ associated to a linear inverse problem with additive white Gaussian noise.

We define the set of realizations where the sequence is bounded in the compact $\cK$ by
\begin{equation*}
    \Lambda_{\cK} = \bigcap_{k \in \N}{\{\vx_k \in \cK\}},
\end{equation*}
and the distance of a point to a set by $d(\vx, \mathbf{S}) = \inf_{\vy \in \mathbf{S}}{\|\vx - \vy \|}$, with $\vx \in \R^d$ and $\mathbf{S} \subset \R^d$.
The fact that we restrict to realizations in $\Lambda_{\cK}$ will be referred to as the "boundedness assumption".

\begin{proposition}\label{prop:convergence_unbiased}
    Under Assumptions~\ref{ass:prior_score_approx}-\ref{ass:data_fidelity_reg}, almost surely on $\Lambda_{\cK}$, we have
    \begin{align*}
        &\lim_{k \to + \infty}{d(\vx_k, \mathbf{S_{\sigma, \cK}})} = 0,\\
        &\lim_{k \to + \infty}{\|\nabla \fJ(\vx_k)\|} = 0,
    \end{align*}
    and $(\fJ(\vx_k))_{k \in \N}$ converges to a value of $\fJ(\mathbf{S_{\sigma, \cK}})$.
\end{proposition}

Proposition~\ref{prop:convergence_unbiased} proves that Algorithm~\ref{alg:Average_PnP} run with the exact MMSE denoiser~\eqref{eq:exact_mmse} converges to the set of critical points of Problem~\eqref{eq:opt_problem_sigma}. This is a weak convergence in the sense that this does not give \emph{a priori} that there is $\vx^\star \in \mathbf{S_{\sigma, \cK}}$ such that $\lim_{k \to + \infty}{\|\vx_k-\vx^\star\|} = 0$. 
Assuming that the sequence is bounded in $\Lambda_{\cK}$ is standard in the stochastic gradient descent analysis~\citep{benaim2006dynamics, castera2021inertial}. We discuss this assumption in Appendix~\ref{sec:boundness_sequence}.

In the previous result, we do not assume that the prior $p$ is smooth but we make Assumption~\ref{ass:prior_score_approx} of a subpolynomial noisy score. This assumption is difficult to verify for a general prior distribution but can be verified in the case of the so-called manifold hypothesis~\cite{debortoli2023convergence}.

\begin{assumption}[Manifold hypothesis]\label{ass:manifold}
    The prior $p$ is supported on a compact set $\mathcal{M}$.
\end{assumption}

Assumption~\ref{ass:manifold} is typically true for an image distribution with bounded pixel values.

\begin{proposition}\label{prop:convergence_manifold_hypothesis}
    Under Assumptions~\ref{ass:step_size_decreas}-\ref{ass:manifold}, almost surely on $\Lambda_{\cK}$, we have
    \begin{align*}
        &\lim_{k \to + \infty}{d(\vx_k, \mathbf{S_{\sigma, \cK}})} = 0,\\
        &\lim_{k \to + \infty}{\|\nabla \fJ(\vx_k)\|} = 0,
    \end{align*}
    and $(\fJ(\vx_k))_{k \in \N}$ converges to a value of $\fJ(\mathbf{S_{\sigma, \cK}})$.
\end{proposition}

Proposition~\ref{prop:convergence_manifold_hypothesis} establishes the convergence of the SNORE algorithm to the critical points of Problem~\eqref{eq:opt_problem_sigma}, under the three mild Assumptions~\ref{ass:step_size_decreas}-\ref{ass:manifold}.

\subsection{Biased algorithm analysis}\label{sec:biased_cvg}
We now quantify the error of the SNORE algorithm (Algorithm~\ref{alg:Average_PnP}) run with an inexact MMSE denoiser $D_{\sigma} \approx D_{\sigma}^\star$. 
Such a study is crucial as the algorithm is run in practice with a learned denoiser $D_{\sigma}$ which is not exact.
With this denoiser, an iteration of Algorithm~\ref{alg:Average_PnP} is computed by
\begin{equation}\label{eq:practical_process}
    \vx_{k+1} = \vx_k - \delta_k \nabla \fF(\vx_k, \vy) - \frac{\alpha \delta_k}{\sigma^2} \left(\vx_k - D_{\sigma}(\vx_k + \zeta_k) \right).
\end{equation}
It can be rewritten as
\begin{equation}
    \vx_{k+1} = \vx_k - \delta_k (\nabla \fJ(x_k) + \xi_k),
\end{equation}
where, by using Equation~\eqref{eq:gradient_average_pnp_reg}, the gradient perturbation writes as
\begin{align*}
    \xi_k &= \nabla \fF(\vx_k, \vy) + \frac{\alpha}{\sigma^2} \left( \vx_k - D_{\sigma}(\vx_k + \zeta_k) \right) - \nabla \fJ(x_k) \\
    &= \frac{\alpha}{\sigma^2}\left(D_{\sigma}(\vx_k + \zeta_k) - D_\sigma^*(x_k) \right).
\end{align*}
This stochastic shift $\xi_k$ is in general biased, i.e. ${\eE(\xi_k) \neq 0}$.

\begin{assumption}\label{ass:denoiser_approx}
    The learned denoiser $D_{\sigma}$ is $\mathcal{C}^{\infty}$ and is a bounded approximation of the exact MMSE denoiser $D_{\sigma}^\star$,  $\forall R > 0$, there exists $M(R) > 0$, such that $\forall \vx \in \mathcal{B}(0,R)$, $\|D_{\sigma}(\vx) - D_{\sigma}^{\star}(\vx)\| \le M(R)$.
\end{assumption}

Assumption~\ref{ass:denoiser_approx} is also made in~\cite{laumont2022maximumaposteriori}, and it can be ensured if the denoiser is learned with a specific loss \citep{Laumont_2022_pnpula}. Moreover, if the activation functions of the denoiser are $\mathcal{C}^{\infty}$ (in our case ELU), then the denoiser is $\mathcal{C}^{\infty}$.

\begin{assumption}\label{ass:denoiser_sub_linear}
    The exact MMSE denoiser and the learned denoiser are sublinear, there exists $C \ge 0$ such that $\forall \vx \in \R^d$, $\|D^{\star}_{\sigma}(\vx)\| \le \|\vx\| + C \sigma$ and $\|D_{\sigma}(\vx)\| \le \|\vx\| + C \sigma$.
\end{assumption}
Assumption~\ref{ass:denoiser_sub_linear} is the stable condition on the denoiser $D_{\sigma}$ in the sense that it is bounded in norm.
As a practical example, a bounded denoiser~\citep[Definition 1]{chan2016plugandplay} verifies Assumption~\ref{ass:denoiser_sub_linear}.

\begin{proposition}\label{prop:biais_convergence}
    Under Assumptions~\ref{ass:step_size_decreas},~\ref{ass:data_fidelity_reg},~\ref{ass:denoiser_approx}, \ref{ass:denoiser_sub_linear}, for $R>0$ and $\cK \subseteq \mathcal{B}(0,R)$ compact, almost surely on $\Lambda_{\cK}$, there exists $M_{\cK} \in ]0,+\infty[$ such that
    \begin{align*}
        \limsup_{k \to +\infty} \|\nabla \fJ(\vx_k)\| &\le M_{\cK} \eta^{\frac{1}{2}}, \\
        \limsup_{k \to +\infty} \fJ(\vx_k) - \liminf_{k \to +\infty} \fJ(\vx_k) &\le M_{\cK} \eta,
    \end{align*}
    with the bias $\eta = \underset{k \to + \infty}{\limsup}{\|\eE(\xi_k)\|} \underset{\sigma \to 0}{\le} \frac{\alpha}{\sigma^2} M(R) + o(\sigma)$.
\end{proposition}

The denoiser bias $\eta$ has a similar bound than in \citep[Proposition 3]{laumont2022maximumaposteriori}.
If the denoiser is well trained, $M(R) \approx 0$, i.e. the denoiser bias is small. 
Proposition~\ref{prop:biais_convergence} proves that the smaller is the denoiser bias, the closer the sequence $(x_k)_{k \in \mathbf{N}}$ is to critical points of Problem~\eqref{eq:opt_problem_sigma} (in terms of gradient norm).
This statement generalizes the convergence result of Proposition~\ref{prop:convergence_unbiased} to the biased case.

\begin{figure*}[!ht]
    \centering
    \begin{subfigure}{0.095\textwidth}
        \includegraphics[width=\textwidth]{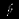}
    \end{subfigure}
    \begin{subfigure}{0.095\textwidth}
        \includegraphics[width=\textwidth]{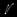}
    \end{subfigure}
    \begin{subfigure}{0.095\textwidth}
        \includegraphics[width=\textwidth]{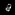}
    \end{subfigure}
    \begin{subfigure}{0.095\textwidth}
        \includegraphics[width=\textwidth]{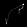}
    \end{subfigure}
    \begin{subfigure}{0.095\textwidth}
        \includegraphics[width=\textwidth]{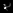}
    \end{subfigure}
    \begin{subfigure}{0.095\textwidth}
        \includegraphics[width=\textwidth]{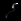}
    \end{subfigure}
    \begin{subfigure}{0.095\textwidth}
        \includegraphics[width=\textwidth]{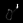}
    \end{subfigure}
    \begin{subfigure}{0.095\textwidth}
        \includegraphics[width=\textwidth]{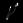}
    \end{subfigure}
    \begin{subfigure}{0.095\textwidth}
        \includegraphics[width=\textwidth]{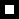}
    \end{subfigure}
    \begin{subfigure}{0.095\textwidth}
        \includegraphics[width=\textwidth]{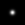}
    \end{subfigure}
    \caption{Kernels used for deblurring. As in~\citep{zhang2017learning, zhang2021plugandplay, pesquet2021learning, hurault2022gradient} we test the different methods on 8 real-world camera shake kernels proposed in~\citep{Levin2009} and on the uniform $9 \times 9$ kernel and the $25 \times 25$ Gaussian kernel with standard deviation 1.6 proposed in~\citep{romano2017little}.}
    \label{fig:kernels_of_blur}
\end{figure*}

\begin{figure*}[!ht]
    \centering
    \includegraphics[width=\textwidth]{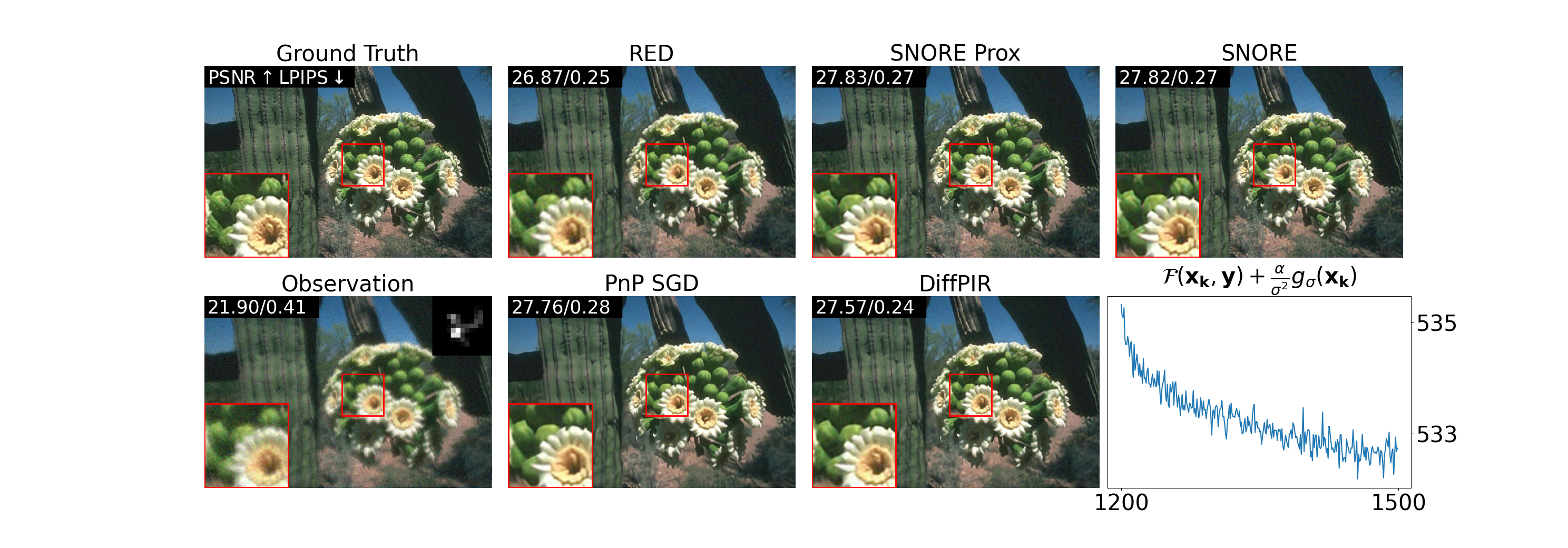}
    \caption{Deblurring with various methods of a motion blur kernel with input noise level $\sigma_{\vy} = 10 / 255$ with a GS-denoiser trained on natural images. Note that Ann-SNORE produces a better perceptual reconstruction (BRISQUE). 
    \textit{Bottom-Right:} Decrease of the optimized function~$\fJ$ (Equation~\eqref{eq:minimization_gs_denoiser}) along the stochastic gradient descent with the last parameters $(\alpha_{m-1}, \sigma_{m-1})$ of the annealing procedure. 
    }
    \label{fig:restauration}
\end{figure*}

\section{Experiments}
In this section, we show the performance of Annealing SNORE (Ann-SNORE) algorithm for image inverse problems, including deblurring and inpainting. Ann-SNORE is compared to several state-of-the-art image restoration methods. In Appendix~\ref{sec:more_exp} we give more details on our experiments and other results on various images and various inverse problems including deblurring, inpainting, super-resolution and despeckling. A study of Ann-SNORE sensitivity to its parameters and the randomness of the algorithm is also provided.
The code used in these experiments can be found in ~\hyperlink{https://github.com/Marien-RENAUD/SNORE}{https://github.com/Marien-RENAUD/SNORE}.

In our experiment, we use a Gradient-Step denoiser~\citep{hurault2022gradient} of the form
\begin{align}
    D_{\sigma} = \Id - \nabla g_{\sigma},
\end{align}
where $g_{\sigma}$ is a learned neural network. With this gradient-step denoiser, \citet{hurault2022gradient} demonstrated that RED converges to a critical point of an explicit objective function of the form $\fF(\vx, \vy) + \alpha g_{\sigma}\left(\vx\right)$. Furthermore, it is established that this objective function decreases throughout the algorithm.
Using the Gradient-Step denoiser in the SNORE Algorithm~\ref{alg:Average_PnP} yields a stochastic gradient descent that minimizes the objective
\begin{align}\label{eq:minimization_gs_denoiser}
    \hspace{-.3cm}\argmin_{\vx \in \R^d}{\mathcal{J}(\vx)}
   \hspace{-.1cm}:=\hspace{-.05cm}{\eE_{\tx \sim p_{\sigma}(\tx|\vx)}\hspace{-.05cm}\left(\fF(\vx, \vy) + \frac{\alpha}{\sigma^2} g_{\sigma}(\tx) \right)}. \hspace{-.1cm}
\end{align}

\subsection{Deblurring}
For image deblurring, the degradation operator is a convolution performed with circular boundary conditions. Therefore $\fF(\vx) = \frac{1}{2 \sigma_{\vy}^2} \|\vy - \mA \vx\|^2$, where $\mA = \mF \mLambda \mF^*$, $\mF$ is the orthogonal matrix of the discrete Fourier transform (and $\mF^*$ its inverse) and $\mLambda$ a diagonal matrix.

We make a gradual annealing, by keeping $\sigma$ and $\alpha$ fixed for some iterations, to efficiently minimize Problem~\eqref{eq:minimization_gs_denoiser}. 
We set $1500$ iterations and $m = 16$ annealing levels. To ensure  convergence, we run $300$ iterations with the last parameters $\sigma_{m-1}, \alpha_{m-1}$.
For all input noise levels $\sigma_{\vy}$, we set $\sigma_0 =  1.8 \sigma_{\vy}$, $\sigma_{m-1} =  0.5 \sigma_{\vy}$, $\alpha_0 = 0.1 \sigma_0^2\sigma_{\vy}^{-2}$ and $\alpha_{m-1} = \sigma_{m-1}^2\sigma_{\vy}^{-2}$. We initialize with the observation $\vx_0 = \vy$ and use a fixed step-size, 
$\delta_k = \delta = 0.1$, as we observe that a decreasing $\delta_k$ 
leads to a slower convergence, as notice by~\citet{laumont2022maximumaposteriori}.

\begin{figure*}[!ht]
  \noindent
  \makebox[\textwidth]{\includegraphics[width=1.01\textwidth]{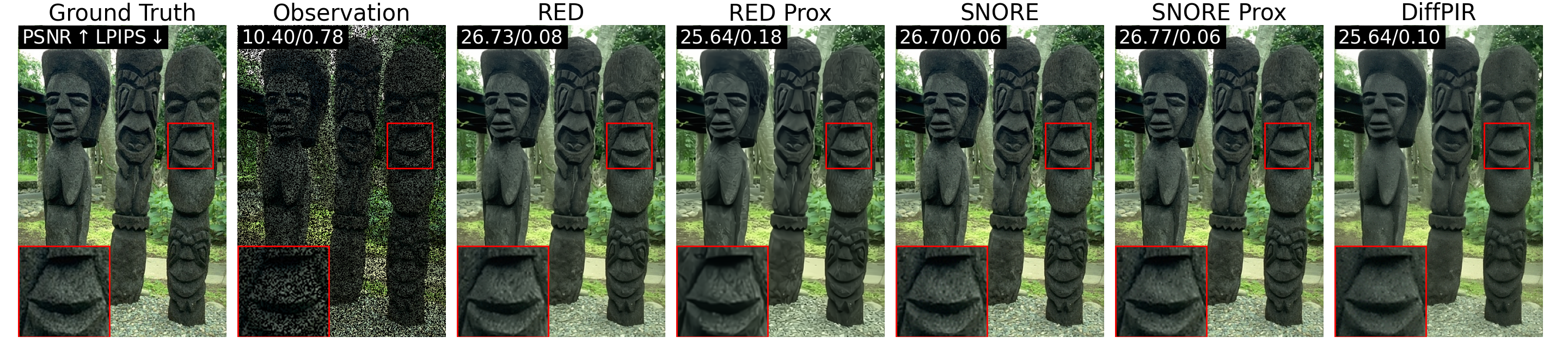}}
    \caption{Inpainting with various methods on a random mask (with a proportion $p=0.5$ of masked pixels) with a GS-denoiser trained on natural images. 
    One can observe here the ability of Ann-SNORE to recover both sharp structures and textural content.}
    \label{fig:restauration_inpainting}
\end{figure*}

We compare in Table~\ref{table:deblurring} our method to RED restoration algorithm of~\citep{romano2017little} (see Algorithm~\ref{alg:RED}) with a gradient-descent step on the data-fidelity, ``RED Prox" method~\citep{romano2017little} with a proximal-descent step on the data-fidelity (see Algorithm~\ref{alg:RED_Prox}), DiffPIR~\citep{Zhu_2023_CVPR} and PnP SGD~\citep{laumont2022maximumaposteriori}. 
As~\citep{hurault2022gradient}, we evaluate each method on $10$ images from CBSD68~\citep{Martin2001} and $10$ blur kernels presented in Figure~\ref{fig:kernels_of_blur}.
The same denoiser trained by~\citep{hurault2022gradient} on natural images is used for all methods to ensure a fair comparison. We compare two variants of the Ann-SNORE algorithm, with gradient step on the data-fidelity (Algorithm~\ref{alg:Annealead_SNORE}) or with a proximal step (Algorithm~\ref{alg:Annelead_SNORE_Prox} in Appendix~\ref{sec:more_exp}). DiffPIR is used with $t_{\text{start}} < t_{\text{trained}}$ as suggested by the authors~\citep{Zhu_2023_CVPR}. However, this diffusion method does not outperform in our experiments because our denoiser is not trained to tackle highly-noised images. An extra-parameter $\beta > 0$ is added to PnP SGD algorithm to increase the algorithm performances. On Table~\ref{table:deblurring}, results are presented with distorsion metrics (PSNR, SSIM) and perceptual metrics with reference (LPIPS) and without reference (BRISQUE).

On Table~\ref{table:deblurring}, we observe that Ann-SNORE has similar performance than other state-of-the-art methods. 
If we compare Ann-SNORE to RED Prox, which performs the best in terms of distortion, we observe that Ann-SNORE performs favorably in terms of perceptual metrics.
However, Ann-SNORE remains slower than other methods for deblurring as it requires a sufficient number of annealing levels 
for the restored images to have a high visual quality (see discussion in Appendix~\ref{sec:more_exp}).

On Figure~\ref{fig:restauration}, we provide a qualitative comparison. Note that the Ann-SNORE algorithm provides a more realistic result than RED Prox with equivalent quantitative score (PSNR). The global decreasing behavior of the function $\fJ$~\eqref{eq:minimization_gs_denoiser} empirically confirms that our algorithm minimizes this function.

\begin{table}
\centering
\resizebox{\linewidth}{!}{%
\begin{tabular}{ c c c c c c }
 Noise level & Method & PSNR$\uparrow$ & SSIM$\uparrow$&LPIPS$\downarrow$&BRISQUE$\downarrow$\\
 \hline
 & RED & 29.82 & 0.84 & 0.17 & \underline{21.09} \\
 & RED Prox & \textbf{30.64} & \textbf{0.87} & \textbf{0.15} & 45.27 \\ 
5/255 & Ann-SNORE & \underline{29.92} & \underline{0.85} & 0.17 & 28.24 \\ 
& Ann-SNORE Prox & \underline{29.92} & \underline{0.85} & 0.17 & 28.25 \\
& DiffPIR & 28.55 & 0.76 & \underline{0.16} & \textbf{17.87} \\

& PnP SGD & 29.37 & 0.83 & 0.20 & 26.44 \\

 \hline
 & RED & 27.18 & 0.72 & 0.25 & \textbf{20.16} \\
& RED Prox & \textbf{28.50} & \textbf{0.80} & \textbf{0.23} & 51.41 \\ 

10/255 & Ann-SNORE & \underline{27.91} & \underline{0.78} & \underline{0.24} & 27.89 \\ 
& Ann-SNORE Prox & \underline{27.91} & \underline{0.78} & \textbf{0.23} & 27.97 \\
& DiffPIR & 27.47 & 0.74 & \underline{0.24} & \underline{21.12} \\
& PnP SGD & 27.61 & 0.75 & 0.26 & 25.74 \\

 \hline
& RED & 24.03 & 0.54 & 0.43 & \textbf{20.75} \\
& RED Prox & \textbf{26.31} & \textbf{0.71} & \textbf{0.31} & 55.48 \\ 
20/255 & Ann-SNORE & 25.61 & 0.66 & \underline{0.32} & 28.88 \\ 
& Ann-SNORE Prox & 25.61 & 0.66 & \underline{0.32} & 29.00 \\
& DiffPIR & \underline{26.05} & \underline{0.69} & 0.33 & 33.57 \\
& PnP SGD & 25.50 & 0.63 & 0.33 & \underline{22.67} \\
 \hline
\end{tabular}
}
\caption{Quantitative comparisons  of image deblurring methods on CBSD10 with $10$ different blur kernels (see Figure~\ref{fig:kernels_of_blur}) and three different level of noise ${\sigma_{\vy} \in \{ 5, 10, 20  \} / 255}$. 
Best and second best results are respectively displayed in bold and underlined.}
\label{table:deblurring}
\end{table}

\subsection{Inpainting}

For image inpainting, the degradation operator $\mA$ is a diagonal matrix with coefficient in $\{0,1\}$. No noise is added to the degraded observation, $\vy = \mA \vx$. The proximal operator of $\fF$ is the orthogonal projection which imposes observed pixels values. This strict condition is relaxed in RED (Algorithm~\ref{alg:RED}) and Ann-SNORE (Algorithm~\ref{alg:Annealead_SNORE}) by taking $\nabla \fF(\vx, \vy) = \mA (\vx - \vy)$. 
We focus on random mask inpainting with a proportion of masked pixels $p = 0.5$.

For Ann-SNORE, we keep the same annealing scheme than for deblurring with $500$ iterations. We set $\sigma_0 = 50/255$, $\sigma_{m-1} = 5/255$, $\alpha_0 = \alpha_{m-1} = 0.15$. The initialization is done with a modified version of the observation $\vy$, where the masked pixels are set to the $0.5$ value. We set a fixed step-size $\delta_k = \delta = 0.5$.

On Table~\ref{table:inpainting}, we compare Ann-SNORE to other methods for the inpainting task with $p=0.5$ on CBSD68. A qualitative comparison is given in Figure~\ref{fig:restauration_inpainting}. We observe that Ann-SNORE outperforms other approaches with better distortion and perceptual scores (PSNR, SSIM, LPIPS). Visually, we observe that Ann-SNORE succeeds to restore textures. Other qualitative results are given in Appendix~\ref{sec:more_exp}. Note that each method (except DiffPIR) is run with the same number of iterations. Hence, contrary to the deblurring task, inpainting with Ann-SNORE does not involve any additional computational load with respect to RED methods.

\begin{table}
\centering
\resizebox{\linewidth}{!}{%
\begin{tabular}{c c c c c }
 Method & PSNR$\uparrow$&SSIM$\uparrow$&LPIPS$\downarrow$&BRISQUE$\downarrow$\\
 \hline
 RED & 31.26 & 0.91 & \underline{0.07} & 17.13 \\
 RED Prox & 30.31 & 0.89 & 0.12 & 38.29 \\ 
 Ann-SNORE & \underline{31.65} & \underline{0.92} & \textbf{0.04} & \underline{7.10} \\
Ann-SNORE Prox & \textbf{31.94} & \textbf{0.93} & \textbf{0.04} & 8.35 \\
DiffPIR & 29.57 & 0.87 & \underline{0.07} & \textbf{4.17} \\
 \hline
\end{tabular}
}
\caption{Inpainting result for random missing pixel with probability $p = 0.5$ on CBSD68 dataset. Best and second best results are respectively displayed in bold and underlined.}
\label{table:inpainting}
\end{table}

\section{Conclusion}
In this work, we introduce a stochastic denoising regularization (SNORE) for solving imaging inverse problems within the PnP framework.
This regularization realizes the heuristic idea that \textit{an image looks clean if its noisy versions look as noisy images (with the same noise level)}. 
Solving inverse problems with this regularization can be addressed with a provably-convergent stochastic optimization algorithm.
The algorithm differs from standard PnP in the fact that the regularization step consists in denoising noised images, thus avoiding a distribution-shift from the data on which the denoiser is trained.
It also draws a connection with recent diffusion-based approaches, which also involve noising-denoising steps, but included in a different global scheme that relates to backward diffusion.
Experiments conducted on  ill-posed inverse problems (deblurring, inpainting) show that SNORE attains state-of-the-art image restoration performance (in terms of full-reference and no-reference quality measures), at the expense of a computational cost which is, for now, larger than competing methods for deblurring.
It would be interesting to determine whether the computational cost could be reduced by relying on other stochastic gradient-descent algorithms (e.g. ADAM~\citep{kingma2014adam}, INNA~\citep{castera2021inertial}), for which theoretical convergence has not been shown yet for functionals including regularizations such as RED or SNORE.

\section{Impact Statement}\label{sec:ethic_statement}
The work presented in this paper addresses the highly ill-posed problem of restoring missing information within an image. 
This sensitivity to errors is particularly pronounced in scenarios involving post-processing algorithms such as segmentation, detection, or classification applied to the reconstructed image, where errors in the reconstruction process may propagate into erroneous decision-making based on the image data. 
This concern is particularly critical in the context of medical images.
SNORE, functioning as a stochastic process reliant on a learned denoiser, inherently produces random fluctuations in its output.
In Section~\ref{sec:uncertainity_snore}, we provide an analysis of SNORE uncertainty to its random seed and its initialization. 
These experiments suggest a robustness of our method. 
Subsequent research effort should focus on quantifying the errors associated with SNORE in order to confirm its utility as a reliable reconstruction algorithm.

\section{Acknowledgements}
This study has been carried out with financial support from the French Direction G\'en\'erale de l’Armement and the French Research Agency through the PostProdLEAP project (ANR-19-CE23-0027-01). Experiments presented in this paper were carried out using the PlaFRIM experimental testbed, 
supported by Inria, CNRS (LABRI and IMB), Universite de Bordeaux, Bordeaux INP and Conseil Regional d’Aquitaine (see https://www.plafrim.fr). We thank Samuel Hurault and Jean-Philippe Rolin for there time and discussions. We thank Gersende Fort that detect an error in the proofs of Proposition\ref{prop:convergence_unbiased}-\ref{prop:biais_convergence} and help us to correct it.

\bibliography{ref}
\bibliographystyle{icml2024}

\newpage
\appendix
\onecolumn

\section*{Supplementary Material}
Our analysis of SNORE is based on stochastic gradient-descent theory and our implementation on the Gradient-Step PnP code~\citep{hurault2022gradient} as well as the Python library DeepInverse for comparisons. In Supplement~\ref{sec:reproducibility}, we provide a reproducibility statement for this work. In Supplement~\ref{sec:related_works_annexe}, we present in more details some related works. In Supplement~\ref{sec:simple_cases_apnp}, SNORE regularization is detailed in simple cases such as Gaussian or Gaussian Mixture priors. In Supplement~\ref{sec:proof_convergence_result}, we provide the proofs of the different propositions presented in Section~\ref{sec:convergence_analysis}. In Supplement~\ref{sec:discussion_assumptions}, we give more explanations about our technical assumptions. In Supplement~\ref{sec:more_exp}, we provide more details on our experimental setting and present additional numerical results. In Supplement~\ref{sec:boundness_sequence}, we discuss the boundedness hypothesis and perspectives of generalization.

\section{Reproducibility Statement}\label{sec:reproducibility}
Anonymous source code is given in supplementary material. It contains a README.md file that explains step by step how to run the algorithm and replicate the results of the paper. 
Moreover, the pseudocode of SNORE algorithm is given in Algorithm~\ref{alg:Annealead_SNORE} and every comparing methods pseudocodes are given in Appendix~\ref{sec:more_exp}.
All parameters setting is detailed in Appendix~\ref{sec:more_exp}. The used datasets and the denoiser weights~\citep{hurault2022gradient} are given in the supplementary materials.
Theoretical results presented in Section~\ref{sec:convergence_analysis} are proved in the appendices.

\section{Related works}\label{sec:related_works_annexe}

First of all, the convergence analysis of RED (Algorithm~\ref{alg:RED}) and RED Prox (Algorithm~\ref{alg:RED_Prox}) algorithms have been made in the literature~\citep{fermanian2023pnpreg, hurault2022gradient}. However, to our knowledge, no convergence analysis have be developed for DiffPIR~\citep{Zhu_2023_CVPR}.

In the existing literature, some stochastic versions of Plug-and-Play have already been proposed. 

\citet{tang2020fast} propose to accelerate the computation of PnP-ADMM algorithms especially when the data-fidelity is heavy to compute by using a mini-batch approximation of the data-fidelity. In the context of gradient descent, this would lead to a computed sequence defined by $\vx_{k+1} = \vx_k - \delta_k \tilde{\nabla} \fF(\vx_k, \vy) - \frac{\alpha \delta_k}{\sigma^2} \left(\vx_k - D_{\sigma}(\vx_k)\right)$ with $\tilde{\nabla} \fF$ a random batch approximation of $\nabla \fF$.

\citet{sun2019block} propose a similar idea based on a batch of random indices, which are optimized at each step. Take $U_i \in \R^{n \times n_i}$ such that $\sum_{i=1}^b n_i = n$ and $\sum_{i=1}^b U_i U_i^T = \mI_n$ then $x_{k+1} = x_k - \delta U_{i_k}^T \nabla \fJ(x_k)$, with $i_k$ a random index in $[1,b]$ and $\fJ = \fF + \alpha \fR$ the functional to minimize. On can remark that the two previous works propose an acceleration of PnP based on random batches but no noise is added inside the process. SNORE does not aim at accelerating PnP algorithms but it rather proposes a stochastic improvement of PnP by injecting noise inside the classical PnP regularization.

\citet{laumont2022maximumaposteriori} propose to run a stochastic gradient descent (SGD) algorithm with the PnP regularization. Thus the computed sequence, PnP SGD, is defined by $\vx_{k+1} = \vx_k - \delta_k \nabla \fF(\vx_k, \vy) - \frac{\alpha \delta_k}{\sigma^2} \left(\vx_k - D_{\sigma}(\vx_k)\right) + \delta_k z_{k+1}$ with $z_{k+1} \sim \nN(0, \mI_d)$. The step-size decreases with time so the additive noise standard deviation also reduces with time. Unlike the above mentioned methods, PnP SGD is not an acceleration of RED~\citep{romano2017little} but an another type of algorithm to minimize the same objective function, $\fF(\cdot, \vy) + \alpha \fP_{\sigma}$. SNORE has two main differences with this algorithm. First the noise is only injected inside the denoiser, and never in the data-fidelity term. Second, the standard deviation of the injected noise is fixed during the time (for fixed $\sigma$). 
\citet{luther2023ddgm} proposed an algorithm similar to Algorithm~\ref{alg:Annealead_SNORE} run with a blind denoiser but do not provide strong theoretical motivation or analysis.

The authors of \citet{kadkhodaie2021solving} propose a method to solve linear inverse problems based on a modification of the Monte Carlo Markov Chain (MCMC) of Langevin Algorithm~\citep{song2020generative}. Their coarse-to-fine stochastic ascent method is defined by (in our notation) $\mathbf{x}_{k+1} = \mathbf{x}_k + \delta_k \left( \left(\operatorname{Id} - A^T A \right) \left( D(\mathbf{x}_k) - \mathbf{x}_k \right) + A^T \left( \mathbf{y} - A \mathbf{x}_k \right) \right) + \gamma_k \mathbf{z}_k$, with $\delta_k = \frac{\delta_0 k}{1+\delta_0 (k-1)}$ the step-size, $D$ a blind denoiser, $\mathbf{z}_k \sim \mathcal{N}(0, \operatorname{Id})$, $\gamma_k = \left( \left(1 - \beta \delta_k \right)^2 - \left(1 - \delta_k \right)^2 \right) \sigma_k$, $\beta \in [0,1]$ and $\sigma_k$ an estimation of the noise level of $\mathbf{x}_k$. In this method, a blind denoiser $D$ is used instead of a denoiser $D_{\sigma}$ which takes the noise level as an input and a non-exact estimation of the noise level of $\mathbf{x}_k$ is made at each iteration. Moreover the parameter $\beta$ is chosen to make the added noise level $\gamma_k$ decrease through the iterations, compared to $\sigma$ that is adapted to the denoiser $D_{\sigma}$ in SNORE. This noise $\gamma_k \mathbf{z}_k$ is added to the entire image and not only to the regularization such as in SNORE. Finally, the authors of \citet{kadkhodaie2021solving} do not provide a theoretical analysis of the proposed algorithm, and depending on the noise schedule, the algorithm may perform posterior sampling or stochastic optimization.

In section~\ref{sec:related_works}, we separate stochastic version of PnP and Diffusion Model (DM) for clarity, but this separation is not strict. 
We recall that PnP is using a denoiser inside an optimization algorithm and DM generate data by simulating a reverse diffusion process. 
However, many recent works do not fit exactly inside these categories.
~\citet{jalal2021robust,sun2023provable} develop a restoration method by running a Langevin dynamics (instead of a reverse diffusion process) with a score-matching network. Other posterior sampling algorithms can be simulated with Langevin Dynamics run with a PnP approximation of the score~\citep{Laumont_2022_pnpula, renaud2023plugandplay}. Another line of work, it to run Gibbs sampling~\citep{coeurdoux2023plugandplay, Bouman2023} with a diffusion model approximation of the score. 
Our new regularization SNORE is part of this field between PnP and DM.

\section{SNORE in simple cases}\label{sec:simple_cases_apnp}

In this section, we detail the behavior of the SNORE regularization in simple cases (Gaussian prior and Gaussian Mixture prior) in order to develop our intuition on this regularization. We observe that this regularization is equivalent to the PnP regularization with a Gaussian prior. Then, in the Gaussian Mixture case, we demonstrate that the gradient of the SNORE regularization converges to the gradient of the ideal regularization $-\nabla \log p$ on every compact when $\sigma \to 0$ and we exhibit the speed of convergence. Finally, we simulate SNORE regularization for 1D distribution and compare it to the PnP regularization.

\subsection{Gaussian Prior}
In order to understand the behavior of Algorithm~\ref{alg:Average_PnP} compared to other algorithms (such as Algorithm~\ref{alg:RED}) we make the computation in a very simple case where the prior is a Gaussian distribution.

We suppose that $p(\vx) = \mathcal{N}(\vx; \vmu, \mSig)$. Then $p_{\sigma}(\tx) = \mathcal{N}(\tx; \vmu, \mSig + \sigma^2 \mI)$.

The score of these distributions has the following expression
\begin{align*}
    - \nabla \log p(\vx) &= \mSig^{-1} (\vx - \vmu)  \\
    - \nabla \log p_{\sigma}(\tx) &= (\mSig + \sigma^2 \mI_d)^{-1} (\tx - \vmu).  
\end{align*}

Moreover, the SNORE regularization defined in Equation~\eqref{eq:new_regularization} can be computed in closed form
\begin{align*}
    &\fR_{\sigma}(\vx) = - \eE_{\tx \sim p_{\sigma}(\tx|\vx)}\left(\log p_{\sigma}(\tx) \right) = \frac{1}{2} \eE_{\tx \sim p_{\sigma}(\tx|\vx)} \left((\tx - \vmu)^T (\mSig + \sigma^2 \mI_d)^{-1} (\tx - \vmu) \right)  \\
    &= \frac{1}{2} \int_{\R^d} (\tx - \vmu)^T (\mSig + \sigma^2 \mI_d)^{-1} (\tx - \vmu) \mathcal{N}(\tx;\vx, \sigma^2 \mI_d) d\tx \\
    &= \frac{1}{2} \int_{\R^d} (\tx - \vx)^T (\mSig + \sigma^2 \mI_d)^{-1} (\tx - \vx) \mathcal{N}(\tx;\vx, \sigma^2 \mI_d) d\tx + \frac{1}{2} \int_{\R^d} (\vx - \vmu)^T (\mSig + \sigma^2 \mI_d)^{-1} (\vx - \vmu) \mathcal{N}(\tx;\vx, \sigma^2 \mI_d) d\tx \\
    &+ \int_{\R^d} (\tx - \vx)^T (\mSig + \sigma^2 \mI_d)^{-1} (\vx - \vmu) \mathcal{N}(\tx;\vx, \sigma^2 \mI_d) d\tx\\
    &= \sigma^2 Tr\left((\mSig + \sigma^2 \mI_d)^{-1} \right) + \frac{1}{2} (\vx - \vmu)^T (\mSig + \sigma^2 \mI_d)^{-1} (\vx - \vmu).
\end{align*}

The gradient of the SNORE regularization is thus
\begin{align*}
    \nabla \fR_{\sigma}(\vx) = (\mSig + \sigma^2 \mI_d)^{-1} (\vx - \vmu), 
\end{align*}
and we have that with a Gaussian prior, $\nabla \fR_{\sigma}(\vx) = \nabla \fP_{\sigma}(\vx) = - \nabla \log p_{\sigma}(\vx)$, which makes SNORE equivalent to traditional PnP.

\subsection{Gaussian Mixture prior}
In this section, we study the behavior of the SNORE regularization in the case of a Gaussian Mixture prior. First, we remember the convergence of $-\nabla \log p_{\sigma}$ to $-\nabla \log p$ and we compute the speed of convergence of $\nabla \fR_{\sigma}$ to $-\nabla \log p$ when $\sigma \to 0$. Then, simulations are run in 1D to give more intuition.

Let us suppose that the prior $p$ is  a Gaussian Mixture Model,
$$p(\vx) = \sum_{i = 1}^p{\pi_i \mathcal{N}(\vx;\vmu_i, \mSig_i)},$$
with $\pi_i \ge 0$ and $\sum_{i=1}^p \pi_i = 1$.

\subsubsection{Theoretical considerations}
The scores of the prior distribution and the noisy prior distributions have the closed forms
\begin{align}\label{eq:score_gaussian_mixture_closed_form}
    \nabla \log p(\vx) &= - \frac{\sum_{i = 1}^p{\pi_i \mSig_i^{-1} (\vx - \vmu_i) \mathcal{N}(\vx;\vmu_i, \mSig_i)}}{\sum_{i = 1}^p{\pi_i \mathcal{N}(\vx;\vmu_i, \mSig_i)}} \nonumber \\
    \nabla \log p_{\sigma}(\vx) &= - \frac{\sum_{i = 1}^p{\pi_i (\mSig_i + \sigma^2 \mI_d)^{-1} (\vx - \vmu_i) \mathcal{N}(\vx;\vmu_i, \mSig_i + \sigma^2
    \mI_d)}}{\sum_{i = 1}^p{\pi_i \mathcal{N}(\vx;\vmu_i, \mSig_i+ \sigma^2 \mI_d)}}.
\end{align}

First we study the behavior of this noisy prior distribution score when $\sigma \to 0$.

\begin{proposition}\label{prop:conv_speed_PnP}
    If the prior $p$ is a Gaussian Mixture Model, then there exist $a_p \in \R^+$, $\alpha_0 > 0$, such that $\forall \vx \in \R^d$, $\forall \sigma \in ]0, \sigma_0]$:
    \begin{align}
        \| \nabla \log p_{\sigma}(\vx) - \nabla \log p(\vx) \| \le  \sigma^2 a_p (\|\vx\|^3 + 1).
    \end{align}  
\end{proposition}

\begin{proof}

We make a series expansion of $\nabla \log p_{\sigma}(\vx)$ when $\sigma$ goes to $0$. Recall that
\begin{align*}
    \mathcal{N}(\vx;\vmu_i, \mSig_i+ \sigma^2 \mI_d) = \frac{1}{\sqrt{\det(2\pi (\mSig_i+ \sigma^2 \mI_d))}} \exp{\left(  - \frac{1}{2} (\vx - \vmu_i)^T (\mSig_i+ \sigma^2 \mI_d)^{-1}(\vx - \vmu_i) \right)}.
\end{align*}
It is known that $\det{(\mSig_i+ \sigma^2 \mI_d)} = \det{(\mSig_i)} + \sigma^2 \text{Tr}{(\Com(\mSig_i)^T)} + \mathcal{O}(\sigma^4)$, with $\text{Tr}$ the trace of the matrix and $\Com(\mSig_i)$ the comatrix of the matrix $\mSig_i$. Then we have
\begin{align*}
    &\mathcal{N}(\vx;\vmu_i, \mSig_i+ \sigma^2 \mI_d) = \frac{1}{\sqrt{\det(2\pi (\mSig_i+ \sigma^2 \mI_d))}} \exp{\left(  - \frac{1}{2} (\vx - \vmu_i)^T (\mSig_i+ \sigma^2 \mI_d)^{-1}(\vx - \vmu_i) \right)} \\
    &= \frac{1}{\sqrt{\det{(2 \pi\mSig_i)}}} \left( 1 - \sigma^2 \frac{ \text{Tr}{(\Com(\mSig_i))}}{2 \det{(\mSig_i)}}  + \mathcal{O}(\sigma^4) \right) \\
    &\exp{\left(  - \frac{1}{2} (\vx - \vmu_i)^T (\mSig_i)^{-1}(\vx - \vmu_i) + \frac{\sigma^2}{2} (\vx - \vmu_i)^T (\mSig_i)^{-2}(\vx - \vmu_i) + \mathcal{O}(\sigma^4) \right)} \\
    &= \mathcal{N}(\vx;\vmu_i, \mSig_i) \left( 1 - \sigma^2 \frac{ \text{Tr}{(\Com(\mSig_i))}}{2 \det{(\mSig_i)}}  + \mathcal{O}(\sigma^4) \right) \left( 1 + \frac{\sigma^2}{2} (\vx - \vmu_i)^T (\mSig_i)^{-2}(\vx - \vmu_i) + \mathcal{O}(\sigma^4) \right) \\
    &= \mathcal{N}(\vx;\vmu_i, \mSig_i) \left( 1 + \frac{\sigma^2}{2} \left( (\vx - \vmu_i)^T (\mSig_i)^{-2}(\vx - \vmu_i) - \frac{ \text{Tr}{(\Com(\mSig_i))}}{\det{(\mSig_i)}} \right) + \mathcal{O}(\sigma^4) \right).
\end{align*}
We define $u_i(\vx) = \frac{1}{2} \left( (\vx - \vmu_i)^T (\mSig_i)^{-2}(\vx - \vmu_i) - \frac{ \text{Tr}{(\Com(\mSig_i))}}{\det{(\mSig_i)}} \right)$. For the score approximation, we get
\begin{align*}
    &\nabla \log p_{\sigma}(\vx) = - \frac{\sum_{i = 1}^p{\pi_i (\mSig_i + \sigma^2 \mI_d)^{-1} (\vx - \vmu_i) \mathcal{N}(\vx;\vmu_i, \mSig_i + \sigma^2 \mI_d)}}{\sum_{i = 1}^p{\pi_i \mathcal{N}(\vx;\vmu_i, \mSig_i+ \sigma^2 \mI_d)}} \\
    &= - \frac{\sum_{i = 1}^p{\pi_i (\mSig_i^{-1} + \sigma^2 \mSig_i^{-2} + \mathcal{O}(\sigma^4)) (\vx - \vmu_i) \mathcal{N}(\vx;\vmu_i, \mSig_i) \left( 1 + \sigma^2 u_i(\vx) + \mathcal{O}(\sigma^4) \right)}}{\sum_{i = 1}^p{\pi_i \mathcal{N}(\vx;\vmu_i, \mSig_i) \left( 1 + \sigma^2 u_i(\vx) + \mathcal{O}(\sigma^4) \right)}} \\
    &= - \sigma^2 \left( \frac{\sum_{i = 1}^p{\pi_i \left( \mSig_i^{-1} (\vx - \vmu_i) u_i(\vx) + \mSig_i^{-2} (\vx - \vmu_i) \right) \mathcal{N}(\vx;\vmu_i, \mSig_i)}}{\sum_{i = 1}^p{\pi_i \mathcal{N}(\vx;\vmu_i, \mSig_i)}} \right) \\ 
    & \qquad \qquad \qquad - \frac{\nabla p (\vx)}{p(\vx) \left(1 + \frac{\sigma^2}{p(\vx)} \sum_{i = 1}^p{\pi_i \mathcal{N}(\vx;\vmu_i, \mSig_i) u_i(\vx)} \right)}
    + \mathcal{O}(\sigma^4)\\
    &= \nabla \log p(\vx) - \sigma^2 c_p(\vx) + \mathcal{O}(\sigma^4),
\end{align*}
with
\begin{align*}
    c_p(\vx) = \frac{\sum_{i = 1}^p{\pi_i \left( \mSig_i^{-1} (\vx - \vmu_i) u_i(\vx) + \mSig_i^{-2} (\vx - \vmu_i) \right) \mathcal{N}(\vx;\vmu_i, \mSig_i)}}{p(\vx)} - \frac{\nabla p(\vx)}{p^2(\vx)} \sum_{i=1}^p{\pi_i \mathcal{N}(\vx;\vmu_i, \mSig_i) u_i(\vx)}.
\end{align*}

This demonstrates the following inequality
\begin{align}\label{eq:score_dl}
    \nabla \log p_{\sigma}(\vx) \underset{\sigma \to 0}{=} \nabla \log p(\vx) - \sigma^2 c_p(\vx) + \mathcal{O}(\sigma^4),
\end{align}

from which we can  deduce a pointwise convergence of $\nabla \log p_{\sigma}$ to $\nabla \log p$ with speed $\sigma^2$. Indeed, there exist ${c_p : \R^d \mapsto \R}$ and $\sigma_0 > 0$, such that for $\sigma \in ]0, \sigma_0]$,
\begin{align}\label{eq:score_dl2}
    \| \nabla \log p_{\sigma}(\vx) - \nabla \log p(\vx) \| \le  2 \sigma^2 \|c_p(\vx)\|,
\end{align}

Next, one can remark that $\forall j \in [1,p]$,
\begin{align*}
    \frac{\pi_j \mathcal{N}(\vx;\vmu_j, \mSig_j)}{p(\vx)} = \frac{\pi_j \mathcal{N}(\vx;\vmu_j, \mSig_j)}{\sum_{i=1}^p{\pi_i \mathcal{N}(\vx;\vmu_i, \mSig_i)}} = \frac{1}{1 + \frac{\sum_{i \neq j}{\pi_i \mathcal{N}(\vx;\vmu_i, \mSig_i)}}{\pi_j \mathcal{N}(\vx;\vmu_j, \mSig_j)}} \le 1.
\end{align*}

As a consequence, there exists $\alpha_p \ge 0$ such that
\begin{align}\label{eq:c_p_control}
    \|c_p(\vx)\| \le \alpha_p (\|\vx\|^3 + 1)
\end{align}

By combining \eqref{eq:score_dl2} and \eqref{eq:c_p_control}, we obtain that there exist $a_p$ and  $\alpha_0 > 0$, such that $\forall \vx \in \R^d$, $\forall \sigma \in ]0, \sigma_0]$,
\begin{align}
    \| \nabla \log p_{\sigma}(\vx) - \nabla \log p(\vx) \| \le  \sigma^2 a_p (\|\vx\|^3 + 1).
\end{align}
\end{proof}

A similar work can be done to evaluate the approximation done by our regularization $\fR_{\sigma}$.
\begin{proposition}\label{prop:conv_speed_APnP}
    If the prior $p$ is a Gaussian Mixture Model, then there exist $\sigma_0 > 0$ and $e_p \in \R^+$ such that $\forall \vx \in \R^d$ and $\sigma \in ]0, \sigma_0]$,
    \begin{align}
        \| \nabla \fR_{\sigma}(\vx) + \nabla \log p(\vx) \| \le \sigma e_p (\|\vx\|^2 + 1).
    \end{align}
\end{proposition}

\begin{proof}

To emphasize the dependence of $\fR$ in $\sigma$, we will denote it as $\fR_{\sigma}$ in the following computations. The behavior of our regularization gradient $\nabla \fR_{\sigma}(\vx) = - \eE_{\tx \sim p_{\sigma}(\tx|\vx)}\left(\nabla \log p_{\sigma}(\tx) \right)$ is
\begin{align*}
    \| \nabla \fR_{\sigma}(\vx) + \nabla \log p(\vx) \|  &= \|\eE_{\tx \sim p_{\sigma}(\tx|\vx)}\left(\nabla \log p_{\sigma}(\tx) \right) - \nabla \log p(\vx) \| \\
    &\le \|\eE_{\tx \sim p_{\sigma}(\tx|\vx)}\left(\nabla \log p_{\sigma}(\tx) - \nabla \log p(\tx) \right) \| + \| \eE_{\tx \sim p_{\sigma}(\tx|\vx)}\left(\nabla \log p(\tx) \right) -  \nabla \log p(\vx) \| \\
    &\le \sigma^2 \eE_{\tx \sim p_{\sigma}(\tx|\vx)}\left( c_p(\tx) \right) + \| \eE_{\tx \sim p_{\sigma}(\tx|\vx)}\left(\nabla \log p(\tx) \right) -  \nabla \log p(\vx) \| \\
    &\le \sigma^2 a_p \eE_{\tx \sim p_{\sigma}(\tx|\vx)}\left( \|\tx\|^3 + 1 \right) + \| \eE_{\tx \sim p_{\sigma}(\tx|\vx)}\left(\nabla \log p(\tx) \right) -  \nabla \log p(\vx) \|. \\
\end{align*}

By using that for $a, b \in \R^+$, $(a+b)^3 \le 4 (a^3 + b^3)$, we get the following inequalities $\eE_{\tx \sim p_{\sigma}(\tx|\vx)}\left( \|\tx\|^3 \right) \le \eE_{\tx \sim p_{\sigma}(\tx|\vx)}\left( \left(\|\tx - \vx\| + \|\vx\|\right)^3 \right) \le 4  \eE_{\tx \sim p_{\sigma}(\tx|\vx)}\left( \|\tx - \vx\|^3 + \|\vx\|^3\right)  = 8 d \sigma^{3} \sqrt{\frac{2}{\pi}} + 4 \|\vx\|^3$.
This leads to the following inequality
\begin{align}\label{eq:decomposition_score_GMM}
    \| \nabla \fR_{\sigma}(\vx) + \nabla \log p(\vx) \|    &\le \sigma^2 a_p \left(8 d \sigma^{3} \sqrt{\frac{2}{\pi}} + 4 \|\vx\|^3 + 1 \right) + \| \eE_{\tx \sim p_{\sigma}(\tx|\vx)}\left(\nabla \log p(\tx) \right) -  \nabla \log p(\vx) \|.
\end{align}

As $\nabla \log p$ is $\mathcal{C}^2$ in $\vx$, by the Taylor theorem, there exists $r > 0$ such that $\forall \vy, \vz \in \mathcal{B}(\vx,r)$
$$\| \nabla \log p(\vy) - \nabla \log p(\vz)\| \le M_{\vx} \|\vy - \vz\|,$$
with $M_{\vx} = 2 \|\nabla^2 \log p(\vx)\|$ and $\mathcal{B}(\vx,r) = \{\vy \in \R^d | \|\vy - \vx \| \le r\}$.

So we have
\begin{align}
    &\| \eE_{\tx \sim p_{\sigma}(\tx|\vx)}\left( \nabla \log p(\tx) \right) - \nabla \log p(\vx) \| = \| \int_{\R^d}{\nabla \log p(\tx) \mathcal{N}(\tx;\vx,\sigma^2 \mI_d) d\tx}  - \nabla \log p(\vx) \|\nonumber \\
    &\le \| \int_{\mathcal{B}(\vx,r)}{(\nabla \log p(\tx) - \nabla \log p(\vx)) \mathcal{N}(\tx;\vx,\sigma^2 \mI_d) d\tx} \| + \| \int_{\R^d \setminus \mathcal{B}(\vx,r)}{(\nabla \log p(\tx) - \nabla \log p(\vx)) \mathcal{N}(\tx;\vx,\sigma^2 \mI_d) d\tx} \| \nonumber\\
    &\le M_{\vx} \int_{\mathcal{B}(\vx,r)}{\| \tx - \vx \| \mathcal{N}(\tx;\vx,\sigma^2 \mI_d) d\tx} + \| \int_{\R^d \setminus \mathcal{B}(\vx,r)}{(\nabla \log p(\tx) - \nabla \log p(\vx)) \mathcal{N}(\tx;\vx,\sigma^2 \mI_d) d\tx} \| \nonumber\\
    &\le M_{\vx} \sigma  \int_{\R^d}{\| \tx \| \mathcal{N}(\tx;0,\mI_d) d\tx}  + \| \int_{\R^d \setminus \mathcal{B}(\vx,r)}{(\nabla \log p(\tx) - \nabla \log p(\vx)) \mathcal{N}(\tx;\vx,\sigma^2 \mI_d) d\tx} \| \label{ineq1}.
\end{align}
One can notice that $\nabla \log p$ is sub-linear, so there exists $b_p \ge 0$ such that $\forall \vx \in \R^d$, $\|\nabla \log p(\vx)\| \le b_p (\|\vx\| + 1)$. Moreover, $\int_{\R^d}{\| \tx \| \mathcal{N}(\tx;0,\mI_d) d\tx} \le \sqrt{\int_{\R^d}{\| \tx \|^2 \mathcal{N}(\tx;0,\mI_d) d\tx}} = \sqrt{d}$. Combining this property with relation~\eqref{ineq1}, we get
\begin{align*}
    \| \eE_{\tx \sim p_{\sigma}(\tx|\vx)}\left( \nabla \log p(\tx) \right) - \nabla \log p(\vx) \| 
    &\le M_{\vx} \sqrt{d} \sigma  + \int_{\R^d \setminus \mathcal{B}(\vx,r)}{(b_p \|\tx\| + b_p \|\vx\| +2) \mathcal{N}(\tx;\vx,\sigma^2 \mI_d) d\tx} \\
    &\le M_{\vx} \sqrt{d} \sigma  + \int_{\R^d \setminus \mathcal{B}(\vx,r)}{(b_p \|\tx\| + b_p \|\vx\| +2) \mathcal{N}(\tx;\vx,\sigma^2 \mI_d) d\tx} .
\end{align*}
Moreover on the set $\R^d \setminus \mathcal{B}(\vx,r)$, the inequality $\|\vx - \tx\|^2 \ge \frac{r^2}{2} + \frac{\|\vx - \tx\|^2}{2}$ holds. This leads to $ \mathcal{N}(\tx;\vx,\sigma^2 \mI_d) \le 2^{\frac{d}{2}} \exp{(-\frac{r^2}{2 \sigma^2})}  \mathcal{N}(\tx;\vx,2\sigma^2 \mI_d)$. Injecting this relation into the previous computation, we obtain
\begin{align*}
    \| \eE_{\tx \sim p_{\sigma}(\tx|\vx)}\left( \nabla \log p(\tx) \right) - \nabla \log p(\vx) \| 
    &\le M_{\vx} \sqrt{d} \sigma  + 2^{\frac{d}{2}} \exp{(-\frac{r^2}{2 \sigma^2})} \int_{\R^d \setminus \mathcal{B}(\vx,r)}{(b_p \|\tx\| + b_p \|\vx\| +2) \mathcal{N}(\tx;\vx,2\sigma^2 \mI_d) d\tx} \\ 
    &\le M_{\vx} \sqrt{d} \sigma  + 2^{\frac{d}{2}}  \exp{(-\frac{r^2}{2 \sigma^2})} \int_{\R^d}{(b_p\|\tx\| + b_p\|\vx\| +2) \mathcal{N}(\tx;\vx,2\sigma^2 \mI_d) d\tx} \\
    &\le M_{\vx} \sqrt{d} \sigma  + 2^{\frac{d}{2}}  \exp{(-\frac{r^2}{2 \sigma^2})} (b_p\sqrt{\|\vx\|^2 + b_p\sqrt{d} \sigma^2 } + b_p\|\vx\| + 2) \\
    &\le M_{\vx} \sqrt{d} \sigma  + \mathcal{O}(\sigma^4)\\
    &\le e_p (\|\vx\|^2 + 1) \sqrt{d} \sigma  + \mathcal{O}(\sigma^4).
\end{align*}
The last inequality holds because there exists $e_p \in \R^+$ such that $M_{\vx} = 2 \|\nabla^2 \log p(\vx)\| \le e_p (\|\vx\|^2 + 1)$.

Combining the above inequalities with Equation~\eqref{eq:decomposition_score_GMM}, we have that
\begin{align*}
    \| \nabla \fR_{\sigma}(\vx) + \nabla \log p(\vx) \| 
    &\le \sigma^2 a_p \left(8 d \sigma^{3} \sqrt{\frac{2}{\pi}} + 4 \|\vx\|^3 + 1 \right) + e_p (\|\vx\|^2 + 1) \sqrt{d} \sigma + \mathcal{O}(\sigma^4) \\
    &\le \sigma \tilde{e}_p (\|\vx\|^2 + 1)  + \sigma^2  a_p (4\|\vx\|^3 + 1) + \mathcal{O}(\sigma^4),
\end{align*}
with $\tilde{e}_p = e_p \sqrt{d}$.

By defining $d_p = 2 \tilde{e}_p$, we have demonstrated that there exist $\sigma_0 > 0$, $d_p \ge 0$ such that $\forall \vx \in \R^d$ and $\forall \sigma \in ]0, \sigma_0]$,
\begin{align}
    \| \nabla \fR_{\sigma}(\vx) + \nabla \log p(\vx) \| \le \sigma d_p (\|\vx\|^2 + 1).
\end{align}
\end{proof}

Proposition~\ref{prop:conv_speed_APnP} shows that, in a case of a GMM prior, our regularization approximates the prior score with a pointwise speed of $\sigma$. This speed of approximation is slower than for the traditional PnP (Proposition~\ref{prop:conv_speed_PnP}).

\subsubsection{Simulations}

In Figure~\ref{fig:1D_gaussian_mixture}, we display a non-trivial Gaussian Mixture distribution in 1D with three Gaussians. Using Equation~\eqref{eq:score_gaussian_mixture_closed_form}, we display $-\log p_{\sigma}$. In order to show $\fR_{\sigma}$, we use Equation~\eqref{eq:new_regularization} and the Euler method to approximate the integration on $\tx$. 

\begin{figure}[!ht]
    \centering
    \includegraphics[width=\textwidth]{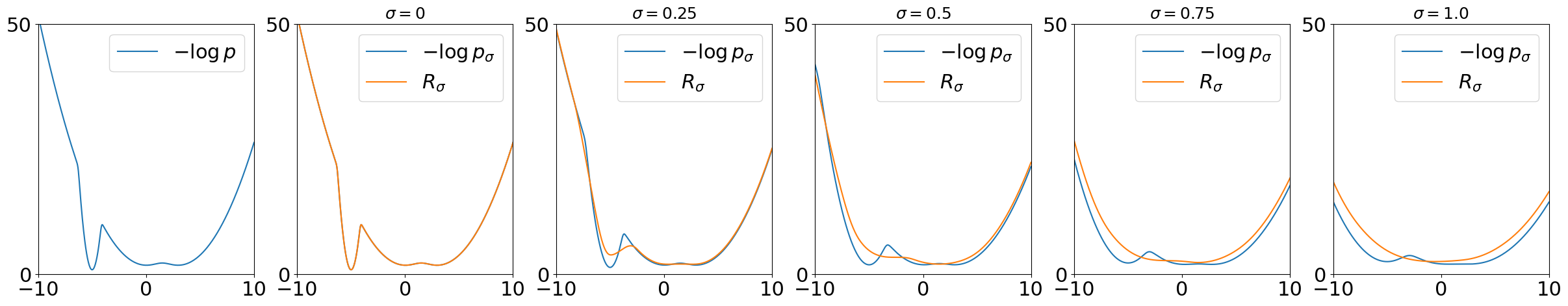}
    \caption{\textit{Leftmost:} Score of the prior $p$. \textit{Rightmost:} Values of $-\log p_{\sigma}$ and $\fR_{\sigma}$ for five values of $\sigma > 0$.}
    \label{fig:1D_gaussian_mixture}
\end{figure}

First, we observe that both $-\log p_{\sigma}$ and $\fR_{\sigma}$ converge to $-\log p$ when $\sigma \to 0$. We also see that $\fR_{\sigma}$ seems to converge to $-\log p$ more slowly that $-\log p_{\sigma}$ as suggested by Proposition~\ref{prop:conv_speed_PnP} and Proposition~\ref{prop:conv_speed_APnP}. Finally, $\fR_{\sigma}$ seems to be "more convex" so easier to minimize and it has a similar minimum that $-\log p_{\sigma}$. This simulation suggests that $\fR_{\sigma}$ is an easier potential to minimize. 

\section{Proofs of section~\ref{sec:convergence_analysis}}\label{sec:proof_convergence_result}
\subsection{Proof of critical Point Analysis}
First, we recall a part of the proof of~\citep[Proposition 1]{laumont2022maximumaposteriori} where it has been proved under some assumptions on the prior distribution that the posterior score approximation converges uniformly to the posterior score on every compact of $\R^d$.
\begin{proposition}\citep[Proposition 1]{laumont2022maximumaposteriori}\label{prop:score_approx_pnp}
    Under Assumption~\ref{ass:prior_regularity}, for $\cK$ a compact of $\R^d$, $\nabla \log p_{\sigma}(\cdot|\vy)$ converges uniformly to $\nabla \log p(\cdot|\vy)$ on $\cK$, 
    \begin{equation*}
        \lim_{\sigma \to 0} \sup_{\cK} \|\nabla \log p_{\sigma}(\cdot|\vy) - \nabla \log p(\cdot|\vy)\| = 0.
    \end{equation*}
\end{proposition}

\begin{proof}
By the Bayes' theorem, $\nabla \log p_{\sigma}(\cdot|\vy) = \nabla \log p(\vy|\cdot) + \nabla \log p_{\sigma}$. As a consequence, Proposition~\ref{prop:score_approx_pnp} is equivalent to show that the score converges uniformly on $\cK$.

For $f \in \mathrm{C}(\R^d,\R^p)$, with $p \in \N$ and $\|f\|_{\infty}<+\infty$, we define for $\vx \in \R^d$
\begin{align*}
    f_{\sigma}(\vx) = \left(f \ast \mathcal{N}_{\sigma}\right)(\vx) 
    = \int_{\R^d}{f(\tx- \tx) \mathcal{N}(\vx ; 0, \sigma^2 \mI_d) d\tx}.
\end{align*}
For $\epsilon > 0$, there exists $R \in \R^+$ such that 
\begin{align*}
    \int_{\R^d \setminus \mathcal{B}(0,R)}{\|f(\vx - \sigma \tx) - f(\vx)\| \mathcal{N}(\tx; 0, \mI_d)d\tx} \le 2 \|f\|_{\infty} \int_{\R^d \setminus \mathcal{B}(0,R)}{\mathcal{N}(\tx; 0, \mI_d)d\tx} \le \frac{\epsilon}{2}.
\end{align*}
Then $\tilde{\cK} = \cK + \mathcal{B}(0,R)$ (Minkowski sum) is compact, so $f$ is uniformly continuous on $\tilde{\cK}$. There exists $\xi > 0$ such that
\begin{align*}
  \forall \sigma \in [0,\xi], \forall \tx \in \mathcal{B}(0,R), \forall \vx \in \cK, \quad  \|f(\vx - \sigma \tx) - f(\vx)\| \le \frac{\epsilon}{2}.
\end{align*}

We can deduce for $\vx \in \cK, \sigma \in [0, \xi]$ that
\begin{align*}
    \|f_{\sigma}(\vx) - f(\vx)\| &\le \int_{\R^d}{\|f(\vx - \tx) - f(\vx)\| \mathcal{N}(\tx; 0, \sigma^2 \mI_d)d\tx} \\
    &\le \int_{\R^d}{\|f(\vx - \sigma \tx) - f(\vx)\| \mathcal{N}(\tx; 0, \mI_d)d\tx} \\
    &\le \int_{\mathcal{B}(0,R)}{\|f(\vx - \sigma \tx) - f(\vx)\| \mathcal{N}(\tx; 0, \mI_d)d\tx} + \int_{\R^d \setminus \mathcal{B}(0,R)}{\|f(\vx - \sigma \tx) - f(\vx)\| \mathcal{N}(\tx; 0, \mI_d)d\tx} \\
    &\le \epsilon.
\end{align*}

So we have the uniform convergence of $f_{\sigma}$ to $f$ on $\cK$: $\forall \epsilon >0$, there exists $\xi>0$ such that $\forall \sigma \in [0, \xi]$
\begin{align*}
    \sup_{\vx \in \cK}\|f_{\sigma}(\vx) - f(\vx)\| &\le \epsilon.
\end{align*}

Applying this result with $f = p$ and $f = \nabla p$ (because $\nabla p_{\sigma} = \nabla (p \ast \nN_{\sigma}) = (\nabla p) \ast \nN_{\sigma}$, we finally get
\begin{align*}
    \sup_{\cK} \|\nabla \log p_{\sigma} - \nabla \log p\| 
    &= \sup_{\cK} \left\|\frac{\nabla p_{\sigma}}{p_{\sigma}} - \frac{\nabla p}{p} \right\|
    = \sup_{\cK} \left\|\frac{(\nabla p_{\sigma}- \nabla p) p + \nabla p (p- p_{\sigma})}{p_{\sigma}p}\right\|.
\end{align*}
We define $m_p = \min\{\inf_{\vx \in \cK}{p(\vx)}\} > 0$, because $\sup_{\cK} \|p- p_{\sigma}\| \to 0$, there exists $\sigma_1 > 0$ such that for $0< \sigma \le \sigma_1$, $\forall \vx \in \cK$, $p_{\sigma}(x) > \frac{m_p}{2}$.
Thus for $0< \sigma \le \sigma_1$,

\begin{align*}
    \sup_{\cK} \|\nabla \log p_{\sigma} - \nabla \log p\| 
    &\le \frac{2(\sup_{\cK} \|\nabla p_{\sigma}- \nabla p\| ) \|p\|_{\infty} + \|\nabla p\|_{\infty} (\sup_{\cK} \|p- p_{\sigma}\|)}{m_p^2} \to 0
\end{align*}

Thus we have prove Propositon~\ref{prop:score_approx_pnp}.

\end{proof}

\subsubsection{Proof of Proposition~\ref{prop:score_approx_apnp}}

We define $\cL = \cK + \mathcal{B}(0,1)$. By Proposition~\ref{prop:score_approx_pnp}, $\nabla \log p_{\sigma}$ converges uniformly to $\nabla \log p$ on $\cL$.

For $\epsilon > 0, \vx \in \cK$,
\begin{align*}
    &\|-\fR_{\sigma}(\vx) - \nabla \log p(\vx)\| = \|\int_{\R^d}{\nabla \log p_{\sigma}(\vx + \zeta) \nN(\zeta; 0, \sigma^2 \mI_d)} d\zeta - \nabla \log p(\vx)\| \\
    &\le\int_{\R^d}{\|\nabla \log p_{\sigma}(\vx + \zeta) - \nabla \log p(\vx)\| \nN(\zeta; 0, \sigma^2 \mI_d)} d\zeta \\
    &\le\int_{\R^d}{\|\nabla \log p_{\sigma}(\vx + \zeta) - \nabla \log p(\vx + \zeta)\| \nN(\zeta; 0, \sigma^2 \mI_d)} d\zeta + \int_{\R^d}{\|\nabla \log p(\vx + \zeta) - \nabla \log p(\vx)\| \nN(\zeta; 0, \sigma^2 \mI_d)} d\zeta \\
    &\le\int_{\mathcal{B}(0,1)}{\|\nabla \log p_{\sigma} - \nabla \log p\|_{\infty, \cL} \nN(\zeta; 0, \sigma^2 \mI_d)} d\zeta + \int_{\R^d \setminus \mathcal{B}(0,1)}{\left(\|\nabla \log p_{\sigma}(\vx + \zeta)\| + \|\nabla \log p(\vx + \zeta)\|\right) \nN(\zeta; 0, \sigma^2 \mI_d)} d\zeta \\
    &+ \int_{\R^d}{\|\nabla \log p(\vx + \zeta) - \nabla \log p(\vx)\| \nN(\zeta; 0, \sigma^2 \mI_d)} d\zeta.
\end{align*}
Because $\cL$ is compact and $\nabla \log p$ is continuous on $\bar{\cK}$, $\nabla \log p$ is uniformly continuous on $\cL$. So there exists $1 \ge \mu > 0$ such that $\forall \vx, \vy \in \cL$, if $\|\vx - \vy \|\le \mu$,
\begin{align}\label{eq:uniform_continuous}
    \|\nabla \log p(\vx) - \nabla \log p(\vy)\| \le \epsilon.
\end{align}

Then by using Equation~\eqref{eq:uniform_continuous}, Assumptions~\ref{ass:prior_regularity}\textbf{(b)} and Assumptions~\ref{ass:prior_score_approx}, we have
\begin{align*}
    &\|\nabla \fR_{\sigma}(\vx) + \nabla \log p(\vx)\| \\
    &\le (\sup_{\cL} \|\nabla \log p_{\sigma} - \nabla \log p\| ) \int_{\mathcal{B}(0,1)}{\nN(\zeta; 0, \sigma^2 \mI_d)} d\zeta + \int_{\R^d \setminus \mathcal{B}(0,1)}{\left(B \sigma^{\beta}(1+\|\vx + \zeta\|^r) + A (1+\|\vx + \zeta\|^q)\right) \nN(\zeta; 0, \sigma^2 \mI_d)} d\zeta \\
    &+ \epsilon \int_{\mathcal{B}(0,\mu)}{\nN(\zeta; 0, \sigma^2 \mI_d)} d\zeta + \int_{\R^d \setminus \mathcal{B}(0,\mu)}{(A (1+\|\vx + \zeta\|^q) + A (1+\|\vx\|^q)) \nN(\zeta; 0, \sigma^2 \mI_d)} d\zeta \\
    &\le \sup_{\cL} (\|\nabla \log p_{\sigma} - \nabla \log p\|) + \epsilon + \int_{\R^d \setminus \mathcal{B}(0,\mu)}{\left(B \sigma^{\beta}(1+\|\vx +\zeta\|^r) + 2 A (1+\|\vx + \zeta\|^q) + A (1+\|\vx\|^q)\right) \nN(\zeta; 0, \sigma^2 \mI_d)} d\zeta,
\end{align*}
where in the last inequality we have used that $\mu \le 1$. Then $\|\vx + \zeta \| \le \|\vx\| + \|\zeta \| \le R + \|\zeta \|$, where $R = \sup_{\vx \in \cK}{\|\vx\|} < + \infty$ because $\cK$ is compact. Moreover on the set $\R^d \setminus \mathcal{B}(0,\mu)$, the inequality $\|\zeta\|^2 \ge \frac{\mu^2}{2} + \frac{\|\zeta\|^2}{2}$ holds. This leads to $ \mathcal{N}(\zeta;0,\sigma^2 \mI_d) \le 2^{\frac{d}{2}} \exp{(-\frac{\mu^2}{2 \sigma^2})} \mathcal{N}(\zeta;0,2\sigma^2 \mI_d)$; and then to the following inequality
\begin{align*}
    &\|\nabla \fR_{\sigma}(\vx) + \nabla \log p(\vx)\| \le \|\nabla \log p_{\sigma} - \nabla \log p\|_{\infty, \cL} + \epsilon \\
    &+ \int_{\R^d \setminus \mathcal{B}(0,\mu)}{\left(B \sigma^{\beta}(1+(R + \|\zeta\|)^r) + 2 A (1+(R + \|\zeta\|)^q) + A (1+R^q)\right) 2^{\frac{d}{2}} \exp{(-\frac{\mu^2}{2 \sigma^2})} \mathcal{N}(\zeta;0,2\sigma^2 \mI_d) d\zeta} \\
    &\le \sup_{\cL}(\|\nabla \log p_{\sigma} - \nabla \log p\|) + \epsilon + (\sigma^{\beta} C_{r, B, R}+ C_{q, A, R}) \exp{(-\frac{\mu^2}{2 \sigma^2})} ,
\end{align*}
with $C_{r, B, R}$ a constant depending on  $r$, $B$ and $R$; and $C_{q, A, R}$ a constant depending on $q$, $A$ and $R$. 
By the uniform convergence of $\nabla \log p_{\sigma}$ to $\nabla \log p$ on $\cL$, there exists $\sigma_0 >0$ such that $\forall \sigma \le \sigma_0$, $\sup_{\cL} \|\nabla \log p_{\sigma} - \nabla \log p\| \le \epsilon$. Then with the polynomial-exponential behavior, there exists, $\sigma_1 > 0$, such that $\forall \sigma \le \sigma_1$, $(\sigma^{\beta} C_{r, B, R} + C_{q, A, R} )\exp{(-\frac{\mu^2}{2 \sigma^2})} \le \epsilon$. Finally, for $\sigma \le \min(\sigma_0, \sigma_1)$,
\begin{align*}
    \|\nabla \fR_{\sigma}(\vx) + \nabla \log p(\vx)\| \le 3 \epsilon.
\end{align*}

\subsubsection{Proof of Proposition~\ref{prop:critical_point_accumulate_values}}

We will follow the same structure as the proof of \citep[Proposition 1]{laumont2022maximumaposteriori} and use our Proposition~\ref{prop:score_approx_apnp}. For $\vx \in \mathbf{S}_{\cK}$, by definition there exist $(\sigma_n)_{n\in \N} \in \mathbf{E}$, $\vx_n \in \mathbf{S}_{\sigma_n, \cK}$ and $\phi : \N \to \N$ strictly increasing such that $\vx_{\phi(n)} \to \vx$. By Proposition~\ref{prop:score_approx_apnp}, $\nabla \fR_{\sigma}$ converges uniformly to $-\nabla \log p$ on $\cK$. So $\nabla \fF(\cdot, \vy) + \alpha \nabla \fR_{\sigma}$ converges uniformly to $\nabla \fF(\cdot, \vy) - \alpha \nabla \log p = \nabla \fJ$ on $\cK$. Because $\forall n \in \N, \vx_{\phi(n)} \in \cK$, we have that $\nabla \fF(\vx_{\phi(n)}, \vy) + \alpha \nabla \fR_{\sigma_{\phi(n)}}(\vx_{\phi(n)}) \to \nabla \fJ(\vx)$. On the other hand, as $\forall n \in \N$, because $\vx_{\phi(n)} \in \sS_{\sigma_{\phi(n)}, \cK}$, $\nabla \fF(\vx_{\phi(n)}, \vy) + \alpha \nabla \fR_{\sigma_{\phi(n)}}(\vx_{\phi(n)}) = 0$, we have $\nabla \fJ(\vx) = 0$, which shows that $\vx \in \mathbf{S}^{\star}_{\cK}$.

\subsubsection{Proof of Proposition~\ref{prop:convergence_unbiased}}\label{sec:proof_unbiased}
We thank Gersende Fort that find an error in the previous proofs of Proposition~\ref{prop:convergence_unbiased}-\ref{prop:biais_convergence} and help us to correct it.

We will apply the result \citep[Theorem 2.1, (ii)]{tadic2017asymptotic}. The \citep[Assumption 2.1]{tadic2017asymptotic} is verified by Assumption~\ref{ass:step_size_decreas}.
First we define $\xi_k = \nabla \fF(\vx_k, \vy) + \alpha \nabla \log p_{\sigma}(\vx_k + \zeta_k) - \nabla \fJ_{\sigma}(x_k) = \alpha(\nabla \log p_{\sigma}(\vx_k + \zeta_k)-\nabla \log p(\vx_k))$, with $\alpha >0, \zeta_k \sim \nN(0,\sigma^2 \mI_d)$. By definition we have $\eE(\xi_k|\mathcal{F}_k) = 0$ and 
\begin{align*}
    \vx_{k+1} = \vx_k - \delta_k (\nabla \fJ_{\sigma}(x_k) + \xi_k).
\end{align*}

Also, $p_{\sigma}$ is $\mathcal{C}^{\infty}$ by convolution with a Gaussian and then $\log p_{\sigma}$ is also $\mathcal{C}^{\infty}$. 
Again, by convolution, $\fR_{\sigma}$ is clearly $\mathcal{C}^{\infty}$. By Assumption~\ref{ass:data_fidelity_reg}, $\fJ = \fF(\cdot, \vy) + \fR_{\sigma}$ is $\mathcal{C}^{\infty}$. So \citep[Assumption 2.3.b]{tadic2017asymptotic} is verified.

Now, we will verify that~\citep[Assumption 2.2]{tadic2017asymptotic} is verified in our case.
We define the martingale
\begin{equation*}
    M_n = \sum_{i = 1}^n \delta_i \xi_i.
\end{equation*}

We will show that $M_n$ converges almost surely on the set $\{\sup_{k\in \N} \|\vx_k\| < +\infty\}$ which implies~\citep[Assumption 2.2]{tadic2017asymptotic}. We have

\begin{align}
    \Prb{M_n \text{ converges}, \sup_{k \in \N} \|\vx_k\| < +\infty} &= \lim_{q \to +\infty} \Prb{M_n \text{ converges}, \sup_{k \in \N} \|\vx_k\| \le q} \nonumber\\
    &= \lim_{q \to +\infty} \Prb{\sum_{i = 1}^n \delta_i \xi_i \text{ converges}, \sup_{k \in \N} \|\vx_k\| \le q} \nonumber\\
    &= \lim_{q \to +\infty} \Prb{\sum_{i = 1}^n \delta_i \xi_i \mathbbm{1}_{\|\vx_{i}\|\le q} \text{ converges}, \sup_{k \in \N} \|\vx_k\| \le q}.\label{eq:martingale}
\end{align}

Therefore we define the martingale
\begin{equation*}
    \bar M_n = \sum_{i = 1}^n \delta_i \xi_i \mathbbm{1}_{\|\vx_{i}\|\le q}.
\end{equation*}

Next we observe that $\bar M_n$ is zero-mean. Indeed we have
\begin{align*}
    &\eE\left( \bar M_n\right) = \eE\left( \eE\left( \bar M_n | \mathcal{F}_{n}\right)\right) = \eE\left( \bar M_{n-1} + \eE\left( \delta_n \xi_{n} \mathbbm{1}_{\|\vx_{n}\| \le q} | \mathcal{F}_{n}\right)\right)= \eE\left( \bar M_{n-1} \right),
\end{align*}
with $\mathcal{F}_{n}$ the filtration adapted to the sequence $\vx_n$.
Since $\eE\left( \bar M_{1} \right) = 0$, we get that $\forall n \in \N$, $\eE\left( \bar M_{n} \right) = 0$.

Moreover, $\bar M_n$ is square-integrable:
\begin{align}
    &\eE\left(\|\bar M_n\|^2 \right) = \eE\left( \eE\left( \|\bar M_n\|^2 | \mathcal{F}_{n}\right)\right) \nonumber\\
    &=\eE\left(\|\bar M_{n-1}\|^2 + 2 \eE\left( \langle \delta_n \xi_n \mathbbm{1}_{\|\vx_{n}\| \le q}, \bar M_{n-1} \rangle | \mathcal{F}_{n}\right) \right) +  \eE\left(\eE\left( \| \delta_n \xi_n \mathbbm{1}_{\|\vx_{n}\| \le q} \|^2 | \mathcal{F}_{n}\right) \right) \nonumber\\
    &=\eE\left(\|\bar M_{n-1}\|^2 + \eE\left( \| \delta_n \xi_n \mathbbm{1}_{\|\vx_{n}\| \le q} \|^2 | \mathcal{F}_{n}\right) \right).\label{eq:computation_square_integrable_m_n_bar}
\end{align}

Then, by Assumption~\ref{ass:prior_score_approx}, we have
\begin{align*}
    \eE(\|\xi_n \mathbbm{1}_{\|\vx_{n}\| \le q}\|^2 | \mathcal{F}_{n}) &= \alpha^2 \eE(\| \nabla \log p_{\sigma}(\vx_n + \zeta_n) - \eE_{\zeta \sim \nN(0,\sigma^2 \mI_d)}\left( \nabla \log p_{\sigma}(\vx_n + \zeta) \right)\|^2 \mathbbm{1}_{\|\vx_{n}\| \le q} | \mathcal{F}_{n}) \\
    &= \alpha^2 \left( \eE(\| \nabla \log p_{\sigma}(\vx_n + \zeta) \|^2\mathbbm{1}_{\|\vx_{n}\| \le q} | \mathcal{F}_{n}) - \|\eE\left( \nabla \log p_{\sigma}(\vx_n + \zeta) \right)\|^2\mathbbm{1}_{\|\vx_{n}\| \le q} | \mathcal{F}_{n}\right) \\
    &\le \alpha^2 \eE(\| \nabla \log p_{\sigma}(\vx_n + \zeta) \|^2\mathbbm{1}_{\|\vx_{n}\| \le q} | \mathcal{F}_{n}) \\
    &\le \alpha^2 \eE(B^2 \sigma^{2 \beta} \left( 1 + \|\vx_n + \zeta \|^r \right)^2 \mathbbm{1}_{\|\vx_{n}\| \le q} | \mathcal{F}_{n}) \\
    &\le \alpha^2 B^2 \sigma^{2 \beta} \eE(\left( 1 + (q + \|\zeta \|)^r \right)^2 | \mathcal{F}_{n})\\
    &:= C_{r, \sigma, d, q},
\end{align*}
where $C_{r, \sigma, d, q}$ is a constant independent of $k$ and $x_k$.

By injecting the previous inequality into equation~\eqref{eq:computation_square_integrable_m_n_bar}, we get
\begin{align*}
\eE\left[\|\bar M_n\|^2 \right] \le \eE\left[\|\bar M_{n-1}\|^2 \right] + C_{r, \sigma, d, q} \delta_n^2.
\end{align*}
By induction, we obtain
\begin{align*}
    &\eE\left[\|\bar M_n\|^2 \right]  \le C_{r, \sigma, d, q} \sum_{i=1}^n \delta_i^2.
\end{align*}

Moreover, we have that
\begin{align*}
    \sum_{i=1}^{\infty} \eE\left[ \|\delta_n \xi_n \mathbbm{1}_{\|x_{n}\| \le q}\|^2 |\mathcal{F}_{n} \right] \le C_{r, \sigma, d, q} \sum_{i=1}^{\infty} \delta_n^2 < +\infty.
\end{align*}
Therefore, by~\citep[Theorem 2.15]{hall2014martingale}, almost surely the martingale $\bar M_n$ converges. 
Then, for any $q\in \N$,
\begin{align*}
    \Prb{\bar M_n \text{ converges}, \sup_{k \in \N}\|\vx_k\| \le q} = \Prb{\sup_{k \in \N}\|\vx_k\| \le q},
\end{align*}
and we get from relation~\eqref{eq:martingale}
\begin{align*}
    \Prb{M_n \text{ converges}, \sup_{k \in \N}\|\vx_k\| < + \infty} &= \lim_{q \to +\infty} \Prb{\sup_{k \in \N}\|\vx_k\| \le q} \\
    &= \Prb{\sup_{k \in \N}\|\vx_k\| < +\infty}.
\end{align*}
The previous equality implies that
\begin{align*}
    \Prb{M_n \text{ converges} \middle| \sup_{k \in \N}\|\vx_k\| < + \infty} = 1.
\end{align*}
So, almost surely on $\{\sup_{k \in \N}\|\vx_k\| < + \infty\}$, the martingale $M_n$ converges and~\citep[Assumption 2.2]{tadic2017asymptotic} is verified.

We conclude the proof by applying \citep[Theorem 2.1. (ii)]{tadic2017asymptotic} with $\eta = 0$.

\subsubsection{Proof of Proposition~\ref{prop:convergence_manifold_hypothesis}}\label{sec:proof_convergence_under_manifold}

We only need to notice that Assumption~\ref{ass:manifold} implies Assumption~\ref{ass:prior_score_approx} and apply Proposition~\ref{prop:convergence_unbiased}.
One can directly apply~\citep[Lemma C.1]{debortoli2023convergence} but we give a proof here for the sake of clarity. Observe that 
\begin{equation*}
    p_{\sigma}(\vx) = (p \star \nN_{\sigma})(\vx) = \int_{\R^d}{p(\vy) \nN_{\sigma}(\vx - \vy) d\vy} = \int_{\mathcal{B}(0, R)}{p(\vy) \nN_{\sigma}(\vx - \vy) d\vy} .
\end{equation*}
By taking $R > 0$ such that $\mathcal{M} \subset \mathcal{B}(0, R)$, we get
\begin{align*}
    \|\nabla \log p_{\sigma}(\vx)\| &= \left\|\frac{\nabla p_{\sigma}(\vx)}{p_{\sigma}(\vx)}\right\| 
    = \frac{1}{p_{\sigma}(\vx)} \left\|\int_{\mathcal{B}(0, R)}{p(\vy) \frac{\vx - \vy}{\sigma^2} \nN_{\sigma}(\vx - \vy) d\vy} \right\|\\
    &\le \frac{1}{p_{\sigma}(\vx)} \int_{\mathcal{B}(0, R)}{p(\vy) \frac{\|\vx - \vy\|}{\sigma^2} \nN_{\sigma}(\vx - \vy) d\vy} \\
    &\le \frac{\|\vx\| + R}{\sigma^2},
\end{align*}
so Assumption~\ref{ass:prior_score_approx} holds with $B = \max(1, R)$, $\beta = -2$ and $r = 1$.

\subsubsection{Proof of Proposition~\ref{prop:biais_convergence}}
We can define $\eta_k$ by $\eta_k = \eE(\xi_k)$ and $\gamma_k = \xi_k - \eE(\xi_k)$. So we have $\xi_k = \gamma_k + \eta_k$ and $\eE(\gamma_k) = 0$.
We will apply~\citep[Theorem 2.1. (ii)]{tadic2017asymptotic}. Assumption 2.1. and 2.3.b. of this paper are already verified because of Assumption~\ref{ass:step_size_decreas}, Assumption~\ref{ass:data_fidelity_reg}, and Assumption~\ref{ass:denoiser_approx}.

Now, we are going to prove that Assumption 2.2. is verified. We start by proving that $\limsup_{k \to +\infty} \|\eta_k\| <+\infty$ almost surely on $\{\sup_{k \in \N} \|\vx_k\| <+\infty\}$. In fact, for $R > 0$
\begin{align*}
    \|\eta_k\| &= \|\eE(\xi_k)\| = \|\eE( \frac{\alpha}{\sigma^2} \left( D_{\sigma}(\vx_k + \zeta_k) - \vx_k - \zeta_k \right) - \frac{\alpha}{\sigma^2} \eE_{\zeta}\left( D^{\star}_{\sigma}(\vx_k + \zeta) - \vx_k - \zeta \right)\| \\
    &=\|\frac{\alpha}{\sigma^2}\eE\left( D_{\sigma}(\vx_k + \zeta) - D^{\star}_{\sigma}(\vx_k + \zeta)\right)\| \\
    &=\frac{\alpha}{\sigma^2} \|\int_{\zeta \in \R^d} (D_{\sigma}(\vx_k + \zeta) - D^{\star}_{\sigma}(\vx_k + \zeta) )\nN(\zeta;0,\sigma^2\mI_d) d\zeta\| \\
    &\le \frac{\alpha}{\sigma^2} \Bigg(\int_{\|\vx_k +\zeta\| \le R} \|D_{\sigma}(\vx_k + \zeta) - D^{\star}_{\sigma}(\vx_k + \zeta)\| \nN(\zeta;0,\sigma^2\mI_d) d\zeta \\
    &~~+ \int_{\|\vx_k +\zeta\| > R} \|D_{\sigma}(\vx_k + \zeta) - D^{\star}_{\sigma}(\vx_k + \zeta)\| \nN(\zeta;0,\sigma^2\mI_d) d\zeta\Bigg)\\
    &\le \frac{\alpha}{\sigma^2} \left(M(R) + \int_{\|\vx_k +\zeta\| > R} \left(\|D_{\sigma}(\vx_k + \zeta)\| + \|D^{\star}_{\sigma}(\vx_k + \zeta)\| \right) \nN(\zeta;0,\sigma^2\mI_d) d\zeta\right)\\
    &\le \frac{\alpha}{\sigma^2} \left(M(R) + \int_{\|\vx_k +\zeta\| > R} \left(2\|\vx_k + \zeta\| +2 C \sigma \right) \nN(\zeta;0,\sigma^2\mI_d) d\zeta\right)\\
    &\le \frac{\alpha}{\sigma^2} \left(M(R) + \int_{\|\vx_k +\zeta\| > R} \left(2\|\vx_k\| + 2 \|\zeta\| +2 C \sigma \right) \nN(\zeta;0,\sigma^2\mI_d) d\zeta\right).\\
\end{align*}

We define $\cK_{\infty}$ the random bounded set of the iterates that exists on the set of realization $\{\sup_{k \in \N} \|\vx_k\| <+\infty\}$.
For $R = R_{\cK_{\infty}} + 1$, with $R_{\cK_{\infty}}$ the radius of $\cK_{\infty}$, if $\|\vx_k +\zeta\| > R$, then $\|\zeta\| > 1$ because $\|\vx_k\| \le R_{\cK_{\infty}} \le R$. This implies that 
\begin{align*}
    \|\eta_k\| &\le \frac{\alpha}{\sigma^2} \left(M(R) + \int_{\|\zeta\| > 1} \left(2R + 2 \|\zeta\| +2 C \sigma \right) \nN(\zeta;0,\sigma^2\mI_d) d\zeta\right) < +\infty.\\
\end{align*}
So, we get that the $\limsup_{k \to \infty} \|\eta_k\|$ is almost surely finite on the set of realizations $\{\sup_{k \in \N} \|\vx_k\| <+\infty\}$, which prove the first part of~\citep[Assumption 2.2]{tadic2017asymptotic}.

Then, we bounded $\eta_k$ on the set of realization $\Lambda_{\cK}$, with $R = 1+ R_{\cK}$ and $R_{\cK}$ the radius of $\cK$. If $\|\zeta\| > 1$, $\|\zeta\|^2 \le \frac{1}{2} + \frac{\|\zeta\|^2}{2}$. So $\nN(\zeta;0,\sigma^2\mI_d) \le 2^{-\frac{d}{2}} \nN(\zeta;0,2\sigma^2\mI_d) \exp\left(-\frac{1}{4\sigma^2} \right)$. So
\begin{align*}
    \|\eta_k\| &\le \frac{\alpha}{\sigma^2} \left(M(R) + 2^{-\frac{d}{2}} \exp\left(-\frac{1}{4\sigma^2}\right) \int_{\|\zeta\| > 1} \left(2\|\vx_k\| + 2 \|\zeta\| +2 C \sigma \right)  \nN(\zeta;0,2\sigma^2\mI_d)  d\zeta\right)\\
    &\le \frac{\alpha}{\sigma^2} \left(M(R) + 2^{-\frac{d}{2}} \exp\left(-\frac{1}{4\sigma^2}\right) \left(2R_{\cK} + 2 \sqrt{d} \sigma +2 C \sigma \right) \right).\\
\end{align*}
This shows that 
\begin{equation*}
    \eta = \limsup_{k \in \N}{\|\eta_k\|} \underset{\sigma \to 0}{\le} \frac{\alpha}{\sigma^2} M(R) + o(\sigma).\\
\end{equation*}

Now, we verify the part of~\citep[Assumption 2.2]{tadic2017asymptotic} on $\gamma_k$, with a similar proof than in Proposition~\ref{prop:convergence_unbiased}.
We define the martingale
\begin{equation*}
    M_n = \sum_{i = 1}^n \delta_i \gamma_i.
\end{equation*}

We will show that $M_n$ converges almost surely on the set $\{\sup_{k\in \N} \|\vx_k\| < +\infty\}$ which implies~\citep[Assumption 2.2]{tadic2017asymptotic}.

\begin{align*}
    &\Prb{M_n \text{ converges}, \sup_{k \in \N} \|\vx_k\| < +\infty} = \lim_{q \to +\infty} \Prb{M_n \text{ converges}, \sup_{k \in \N} \|\vx_k\| \le q} \\
    &= \lim_{q \to +\infty} \Prb{\sum_{i = 1}^n \delta_i \gamma_i \text{ converges}, \sup_{k \in \N} \|\vx_k\| \le q} \\
    &= \lim_{q \to +\infty} \Prb{\sum_{i = 1}^n \delta_i \gamma_i \mathbbm{1}_{\|\vx_{i}\|\le q} \text{ converges}, \sup_{k \in \N} \|\vx_k\| \le q}
\end{align*}

Therefore we define the martingale
\begin{equation*}
    \bar M_n = \sum_{i = 1}^n \delta_i \gamma_i \mathbbm{1}_{\|\vx_{i}\|\le q}.
\end{equation*}

$\bar M_n$ is zero-mean, in fact
\begin{align*}
    &\eE\left( \bar M_n\right) = \eE\left( \eE\left( \bar M_n | \mathcal{F}_{n}\right)\right) = \eE\left( \bar M_{n-1} + \eE\left( \delta_n \gamma_{n} \mathbbm{1}_{\|\vx_{n}\| \le q} | \mathcal{F}_{n}\right)\right) \\
    &= \eE\left( \bar M_{n-1} \right).
\end{align*}
Since $\eE\left( \bar M_{1} \right) = 0$, we get that $\forall n \in \N$, $\eE\left( \bar M_{n} \right) = 0$.

Moreover, $\bar M_n$ is square-integrable. In fact
\begin{align}
    &\eE\left(\|\bar M_n\|^2 \right) = \eE\left( \eE\left( \|\bar M_n\|^2 | \mathcal{F}_{n}\right)\right) \nonumber\\
    &=\eE\left(\|\bar M_{n-1}\|^2 + 2 \eE\left( \langle \delta_n \gamma_n \mathbbm{1}_{\|\vx_{n}\| \le q}, \bar M_{n-1} \rangle | \mathcal{F}_{n}\right) \right) +  \eE\left(\eE\left( \| \delta_n \gamma_n \mathbbm{1}_{\|\vx_{n}\| \le q} \|^2 | \mathcal{F}_{n}\right) \right) \nonumber\\
    &=\eE\left(\|\bar M_{n-1}\|^2 + \eE\left( \| \delta_n \gamma_n \mathbbm{1}_{\|\vx_{n}\| \le q} \|^2 | \mathcal{F}_{n}\right) \right),\label{eq:computation_square_integrable_m_n_bar_2}
\end{align}

However, using the definition of $\xi_k$ and Assumption~\ref{ass:denoiser_sub_linear}, we have
\begin{align*}
    \eE(\|\gamma_k\|^2 \mathbbm{1}_{\|\vx_{k}\| \le q}|\mathcal{F}_{k}) &= \eE(\|\xi_k - \eE(\xi_k)\|^2|\vx_k) \\
    &= \frac{\alpha^2}{\sigma^4} \eE\left( \| D_{\sigma}(\vx_k + \zeta_k) - \vx_k - \zeta_k  - \eE_{\zeta}\left( D_{\sigma}(\vx_k + \zeta) - \vx_k - \zeta \right) \|^2\mathbbm{1}_{\|\vx_{k}\| \le q}|\mathcal{F}_{k} \right) \\
    &= \frac{\alpha^2}{\sigma^4} \eE\left( \| D_{\sigma}(\vx_k + \zeta)-\zeta- \eE_{\zeta}\left( D_{\sigma}(\vx_k + \zeta) - \zeta \right)   \|^2\mathbbm{1}_{\|\vx_{k}\| \le q} |\mathcal{F}_{k}\right) \\
    &= \frac{\alpha^2}{\sigma^4} \left( \eE\left( \| D_{\sigma}(\vx_k + \zeta) - \zeta   \|^2\mathbbm{1}_{\|\vx_{k}\| \le q}|\mathcal{F}_{k}\right) - \|\eE_{\zeta}\left( D_{\sigma}(\vx_k + \zeta)-\zeta |\mathcal{F}_{k}\right)\|^2 \mathbbm{1}_{\|\vx_{k}\| \le q}  \right)\\
    &\le \frac{\alpha^2}{\sigma^4} \eE\left( \| D_{\sigma}(\vx_k + \zeta) - \zeta   \|^2\mathbbm{1}_{\|\vx_{k}\| \le q}|\mathcal{F}_{k}\right)\\
    &\le \frac{2\alpha^2}{\sigma^4} \left(\eE\left( \| D_{\sigma}(\vx_k + \zeta) \|^2\mathbbm{1}_{\|\vx_{k}\| \le q} |\mathcal{F}_{k}\right)  +\eE\left( \| \zeta   \|^2|\mathcal{F}_{k}\right) \right)\\
    &\le \frac{2\alpha^2}{\sigma^4} \left(\eE\left( \left(\|\vx_k + \zeta\| + C \sigma \right)^2 \mathbbm{1}_{\|\vx_{k}\| \le q}|\mathcal{F}_{k}\right)  +d \sigma^2 \right)\\
    &\le \frac{2\alpha^2}{\sigma^4} \left( 2 \eE\left( \left(q + \|\zeta\|\right)^2|\mathcal{F}_{k}\right) + 2 C^2 \sigma^2  +d \sigma^2 \right)\\
    &\le \frac{2\alpha^2}{\sigma^4} \left( 4 q^2 + 4 d \sigma^2 + 2 C^2 \sigma^2  +d \sigma^2 \right) < + \infty\\
    &:= C_{2}.
\end{align*}

By injecting the previous inequality in equation~\eqref{eq:computation_square_integrable_m_n_bar_2}, we get
\begin{align*}
 \eE\left(\|\bar M_n\|^2 \right) \le \eE\left(\|\bar M_{n-1}\|^2 \right) + C_2 \delta_n^2
\end{align*}
By induction, we get
\begin{align*}
    &\eE\left(\|\bar M_n\|^2 \right)  \le C_2 \sum_{i=1}^n \delta_i^2.
\end{align*}

Moreover, we have that,
\begin{align*}
    \sum_{n=1}^{\infty} \eE\left( \|\delta_n \gamma_n \mathbbm{1}_{\|x_{n}\| \le q}\|^2 |\mathcal{F}_{n} \right) \le C_2 \sum_{n=1}^{\infty} \delta_n^2 < +\infty.
\end{align*}
Therefore, by~\citep[Theorem 2.15]{hall2014martingale}, almost surely the martingale $\bar M_n$ converges.

Then, for any $q\in \N$,
\begin{align*}
    \Prb{\bar M_n \text{ converges}, \sup_{k \in \N}\|\vx_k\| \le q} = \Prb{\sup_{k \in \N}\|\vx_k\| \le q},
\end{align*}
And we get
\begin{align*}
    \Prb{M_n \text{ converges}, \sup_{k \in \N}\|\vx_k\| < + \infty} &= \lim_{q \to +\infty} \Prb{\sup_{k \in \N}\|\vx_k\| \le q} \\
    &= \Prb{\sup_{k \in \N}\|\vx_k\| < +\infty}.
\end{align*}
The previous equality implies that
\begin{align*}
    \Prb{M_n \text{ converges} \middle| \sup_{k \in \N}\|\vx_k\| < + \infty} = 1.
\end{align*}
So, almost surely on $\{\sup_{k \in \N}\|\vx_k\| < + \infty\}$, the martingale $M_n$ converges and Assumption 2.2. of \citep{tadic2017asymptotic} is verified.
Finally, we obtain Proposition~\ref{prop:biais_convergence} by applying Theorem~\citep[Theorem 2.1. (ii)]{tadic2017asymptotic}.

\begin{figure}[!ht]
    \centering
    \includegraphics[width=\textwidth]{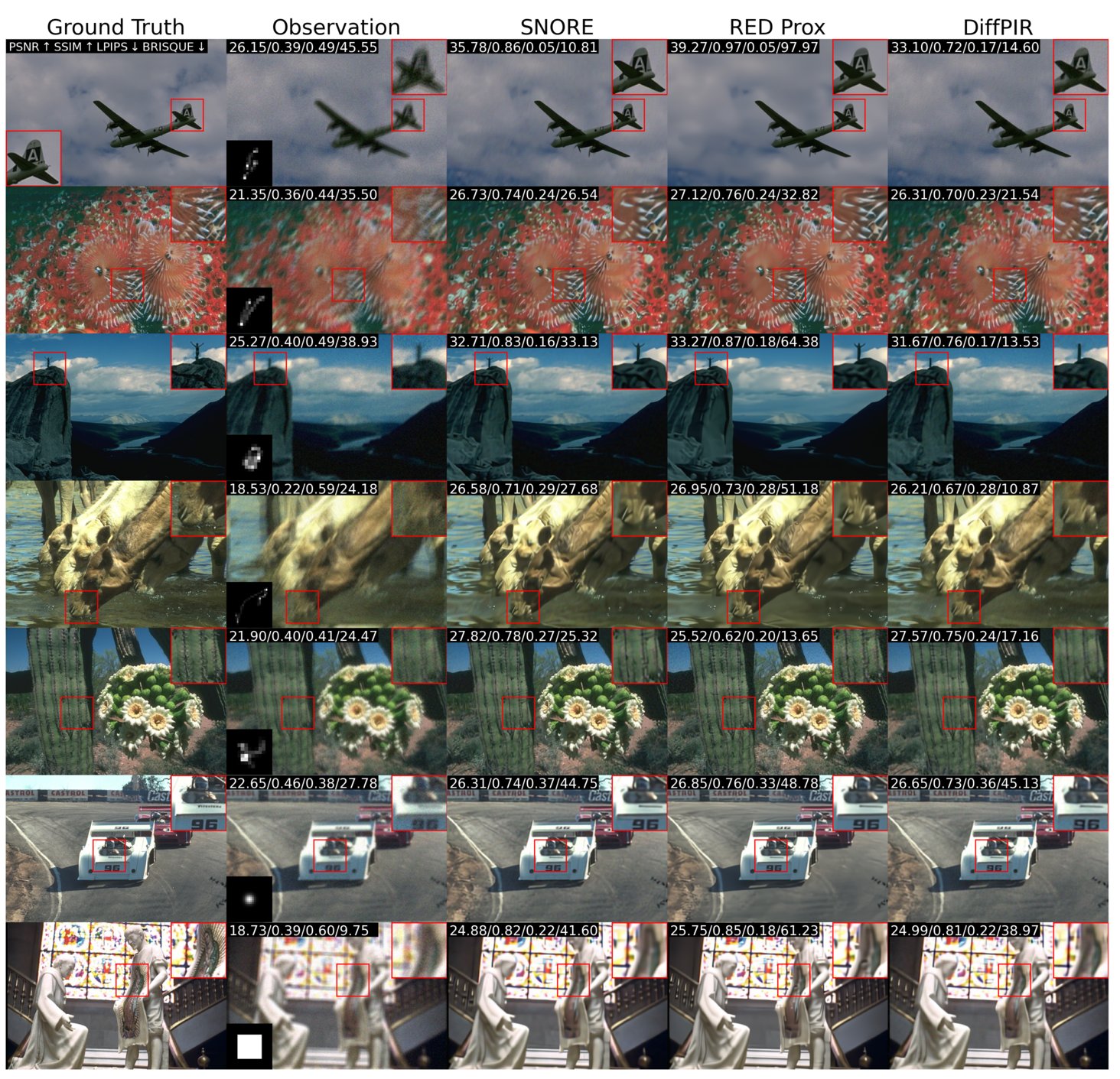}
    \caption{Results of SNORE, RED and DiffPIR algorithms on deblurring tasks with various blur kernels and various images from the dataset CBSD10. The input noise is $\sigma_{\vy} = 10/255$. Kernels have various sizes and have been resized for clarity of the figure. The last kernel, uniform of size $9 \times 9$ is plotted on a black background for visibility.}
    \label{fig:sets_of_result_deblurring}
\end{figure}

\section{Discussion on Assumptions}\label{sec:discussion_assumptions}
In this section, we explicitly list and comment all the assumptions used for results presented in Section~\ref{sec:convergence_analysis}. We detail in which case each assumption is verified, especially in the PnP context. A special discussion on the boundedness assumption is postponed in Appendix~\ref{sec:boundness_sequence}.

\begin{itemize}
    \item Assumption~\ref{ass:prior_regularity}~\textbf{(a)} implies that the score is non-localized with mass $p > 0$ everywhere, and the density is bounded and Lipschitz. Moreover the density is $C^1$, which is necessary to define the score $\nabla \log p$.  
    Note that image distributions are typically supported on a compact set, so their score $\nabla \log p$ is not defined everywhere. Thus, this hypothesis is not realistic in imaging context but it is necessary 
    to realize an analysis of the score $\nabla \log p$. However, if $\log p$ is taken among classical image models (e.g. related to Total Variation or Tychonov regularization), it is likely that Assumption~\ref{ass:prior_regularity}~\textbf{(a)} is verified.
    In particular, we conjecture that this hypothesis is true for the PnP regularization with explicit potential proposed in~\citep{hurault2022gradient}.
    \item Assumption~\ref{ass:prior_regularity}~\textbf{(b)} is verified for Gaussian, Gaussian Mixture or Cauchy distributions. This is a technical assumption to allow a critical point analysis. This hypothesis is non-restrictive because $q$ can be arbitrarily large.
    \item Assumption~\ref{ass:prior_score_approx} is implied by Assumption~\ref{ass:manifold} with $r = 1$ and $\beta = -2$~\citep{debortoli2023convergence}. We provided a simplified proof of this result in Section~\ref{sec:proof_convergence_under_manifold}. \citet{debortoli2023diffusion} have also proved Assumption~\ref{ass:prior_score_approx}  with $r = 1$ and $\beta = 0$ under the Assumption~\ref{ass:prior_regularity}\textbf{(b)} with $q = 1$ and an extra assumption difficult to verify: there exist $m_0 > 0$ and $d_0 \ge 0$ such that $\forall \vx \in \R^d$, $\langle \nabla \log p(\vx), \vx  \rangle \le - m_0 \|\vx\|^2 + d_0 \|\vx\|$. However, to our knowledge, there is no general study of the score approximation $\nabla \log p_{\sigma}$ for a general smooth distribution or a non-bounded and non-smooth distribution. Assumption~\ref{ass:prior_score_approx} seems non-restrictive because $r$ can be arbitrarily large and $\beta$ highly negative.
    \item Assumption~\ref{ass:step_size_decreas} is standard in stochastic gradient descent~\citep{Metivier1984, Benaim1999, tadic2017asymptotic}. A stochastic process imposes a decreasing step-size to ensure convergence. If the step-size is constant, the noise (with a non-decreasing variance) that is added at each step makes the process explore even after a large number of steps, so there is no convergence to expect. In practice, we take a constant step-size to converge faster. In fact,  we run the algorithm for only a small number of iterations (few hundreds).
    \item Assumption~\ref{ass:data_fidelity_reg} makes the data-fidelity $\fF$ smooth. It is verified for every linear inverse problem with a Gaussian noise. Assumption~\ref{ass:data_fidelity_reg} can also be verified in the case of a non-Gaussian noise, for instance with a Fischer-Tippett noise~\citep{Deledalle_2017}. However, some specific data-fidelity terms do not verified this assumption, such as the ones related to salt and pepper noise~\citep{chan2005salt, Nikolova2004} or Laplacian noise~\citep{Huang2017}, or other data-fidelity used for image segmentation~\citep{Chan2006}.
    \item Assumption~\ref{ass:manifold} has been used in the same form as presented by~\citet{debortoli2023convergence}. This assumption is validated by real image distributions, typically encoded within a finite range, such as $[0,1]$ or $[0,255]$. Traditionally, the manifold hypothesis asserts that the distribution of images is confined to a low-dimensional manifold.~\citet{fefferman2013testing} conducted tests to evaluate this hypothesis, and more recently,~\citet{brown2023verifying} have focused on analyzing image distributions.
    However, within the scope of our study, we refer to the \textit{manifold hypothesis} as defined in~\citep{debortoli2023convergence}, suggesting that images are supported within a compact set. This diverges from the conventional manifold hypothesis, constituting a relaxed version of the original assumption.
    \item Assumption~\ref{ass:denoiser_approx} quantifies the uniform distance between the exact and the inexact denoiser. If the activation function of the denoiser is $\mathcal{C}^{\infty}$, then the denoiser is $\mathcal{C}^{\infty}$. 
    \citet{laumont2022maximumaposteriori} make a similar assumption. In the literature, other types of assumptions have been made to control the error of the inexact MMSE  (see for instance~\citep[Assumption 5]{shoushtari2023prior}). We choose to take this form of assumption because \citet[Proposition 4]{Laumont_2022_pnpula} have proved that Assumption~\ref{ass:denoiser_approx} can be ensured if the denoiser is learned with the Noise2Noise loss~\citep{Lehtinen2018Noise2Noise}.
    \item Assumption~\ref{ass:denoiser_sub_linear} is natural for a well trained denoiser. It means that the denoiser $D_{\sigma}$ can only modify the image by a quantity bounded by the level of noise $\sigma$. As a practical example, a bounded denoiser~\citep[Definition~1]{chan2016plugandplay} verifies Assumption~\ref{ass:denoiser_sub_linear}. However, Assumption~\ref{ass:denoiser_sub_linear} is difficult to verify in practice (especially outside the training domain). To our knowledge there is no theoretical argument stronger than intuition to support this assumption. However, Assumption~\ref{ass:denoiser_sub_linear} is necessary to analyze the stability of SNORE with an inexact denoiser, presented in Proposition~\ref{prop:biais_convergence}.
\end{itemize}

\section{Additional Experiments}\label{sec:more_exp}

In this section, we provide more details about our experiments. First we present general technical details. Then, we present more experiment for image deblurring including parameters influence and uncertainty of SNORE. Moreover, we give all parameter setting for image deblurring and inpainting. Finally, a preliminary experiment on image super-resolution is shown.

\paragraph{Metrics} We use four metrics to evaluate our results. Structural SIMilarity (SSIM) and Peak Signal to Noise Ratio (PSNR) are two common distortions metrics. BRISQUE is a no-reference metric based on natural scene statistics that was proved to correlate well with human perception~\citep{Mittal2012}. This metric gives a score between $0$ (best) and $100$ (worst). We use the Python library ``brisque", with which we sometimes observe some incoherence with some outlier outputs (smaller than $0$ or larger than $100$). These outputs where rare so we kept this standard implementation for reproducibility purpose.
LPIPS~\citep{zhang2017learning} is another perceptual metric that compares the original image and the reconstructed one by measuring their differences in terms of deep features. As suggested in~\cite{ren2023multiscale}, looking at such perceptual metrics is relevant as they are correlated with human perception. We use the Python library "lpips" to compute this metric.

\paragraph{Denoiser}
We use the denoiser proposed by~\citep{hurault2022gradient} based on the DRUNet~\citep{zhang2021plugandplay} trained on a dataset of natural images composed of Berkeley segmentation dataset (CBSD)~\citep{Martin2001}, Waterloo Exploration dataset~\citep{ma2017waterloo}, DIV2K dataset~\citep{Agustsson2017} and Flick2K~\citep{lim2017enhanced}. We take the training weights of~\citep{hurault2022gradient}. In order to better analyze the advantages and drawbacks of other PnP and RED approaches, this denoiser is used for all our experiments and comparisons.

\paragraph{RED Prox} 
We name \textit{RED Prox} (Algorithm~\ref{alg:RED_Prox}) the splitting algorithm with a gradient descent step on the regularization and a proximal step on the data-fidelity. Just like RED (Algorithm~\ref{alg:RED}), RED Prox minimizes Problem~\eqref{eq:minimization_problem} with the regularization $\fP_{\sigma}$, defined in Equation~\eqref{eq:prior_approx}. That kind of splitting algorithm is commonly used in the PnP field~\cite{Ryu2019}. As proposed in~\citep{hurault2022gradient}, thanks to the special form of gradient-step denoiser, we implement a backtracking procedure for RED and RED Prox.

\begin{algorithm}
\caption{RED Prox}\label{alg:RED_Prox}
\begin{algorithmic}[1]
\STATE \textbf{input:} $\vx_0 \in \R^d$, $\sigma > 0$, $\alpha > 0$, $\delta > 0$, $N \in \N$
\FOR{$k = 0, 1, \dots, N-1$}
    \STATE $\vz_k \gets \vx_k - \frac{\alpha \delta}{\sigma^2} \left(\vx_k - D_{\sigma}(\vx_k) \right)$ 
    \STATE $\vx_{k+1} \gets \text{Prox}_{\delta \fF}(\vz_k)$ 
\ENDFOR
\end{algorithmic}
\end{algorithm}

\begin{figure}[!ht]
    \centering
    \begin{subfigure}{0.48\textwidth}
        \includegraphics[width=\textwidth]{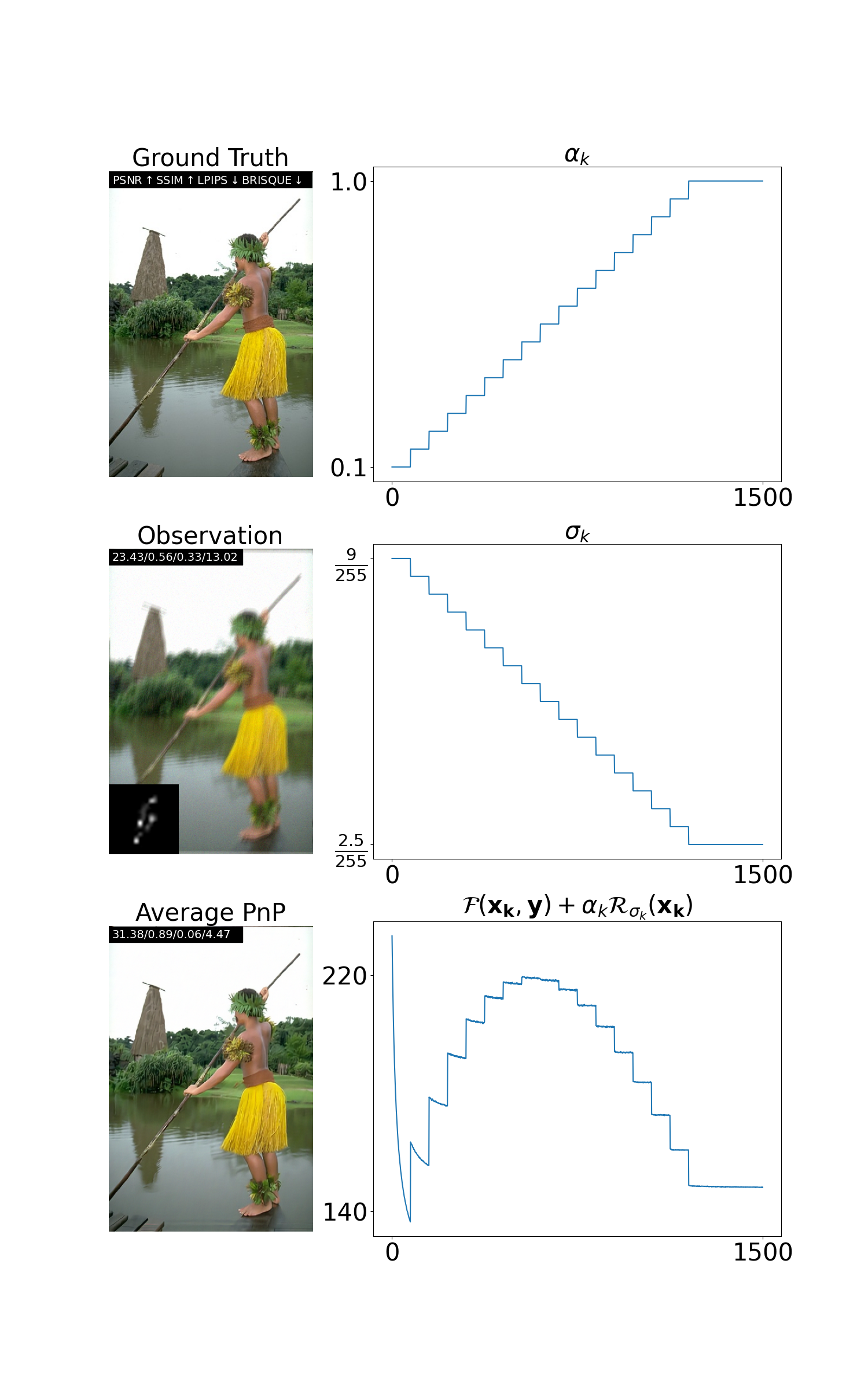}
    \end{subfigure}
    \hfill
    \begin{subfigure}{0.48\textwidth}
        \includegraphics[width=\textwidth]{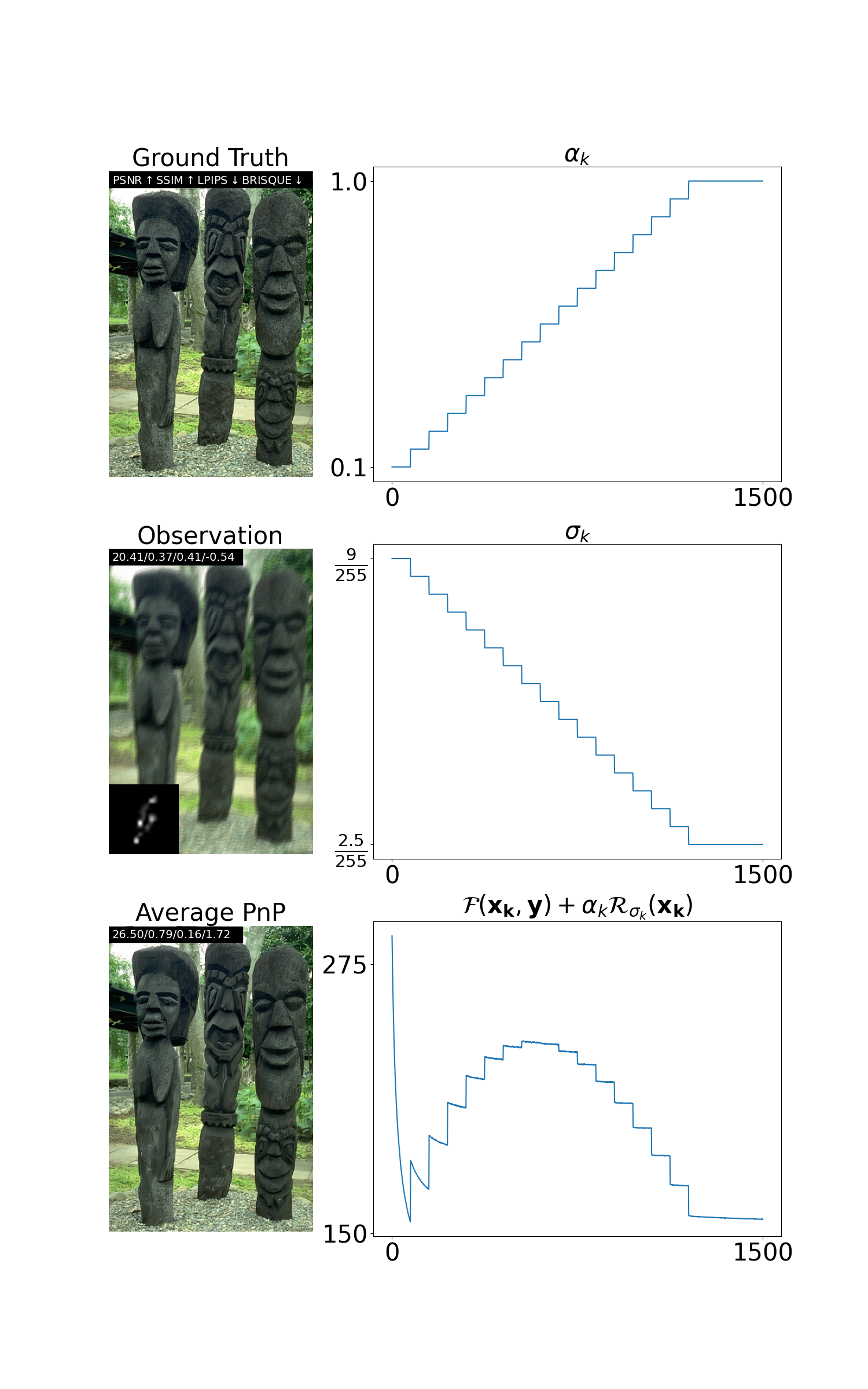}
    \end{subfigure}
    \caption{Image deblurring with an input noise level $\sigma_{\vy} = 5/255$, annealing parameters $\alpha_k$ and $\sigma_k$ and the potential $\fF + \alpha_k \fR_{\sigma_k}$ which is optimized by Annealed SNORE algorithm (Algorithm~\ref{alg:Annealead_SNORE}). There are $16$ annealing levels in the first $1200$ iterations and $300$ iterations with the final parameters. Note that the potential is minimized at each level for fixed $\alpha_k$ and $\sigma_k$. The optimization can be understood in two steps. In the first one, $\sigma_k$ is large and $\alpha_k$ is small to deblur the image. Then a refinement is realized with $\alpha_k$ large and $\sigma_k$ small, to generate a realistic image.}
    \label{fig:restauration_annealing_parameters}
\end{figure}

\paragraph{SNORE Prox} 
We name SNORE Prox the algorithm detailed in Algorithm~\ref{alg:Annelead_SNORE_Prox}. A Proximal descent step is computed on the data-fidelity instead of a gradient-descent step. Note that our convergence analysis of Section~\ref{sec:convergence_analysis} does not apply to this algorithm. We test this algorithm for a comparison with the original SNORE (Algorithm~\ref{alg:Annealead_SNORE}). Experimentally, we observe that both algorithms reach similar performance (see Table~\ref{table:deblurring} or Table~\ref{table:inpainting}).

\begin{algorithm}
\caption{Annealed SNORE Prox}\label{alg:Annelead_SNORE_Prox}
\begin{algorithmic}[1]
\STATE \textbf{input:} $\vx_0 \in \R^d$, $m \in \N$, $\delta > 0$, $\sigma_0 > \sigma_1 > \dots > \sigma_{m-1} \approx 0$, $\alpha_0, \alpha_1, \dots, \alpha_{m-1} > 0$, $N_0,  N_1, \dots, N_{m-1} \in \N$
\FOR{$i = 0, 1, \dots, m-1$}
    \FOR{$k = 0, 1, \dots, N_i-1$}
        \STATE $\veps \gets \mathcal{N}(0, \mI_d)$
        \STATE $\tx_k \gets \vx_k + \sigma_i \veps$
        \STATE $\vz_k \gets \vx_k  - \frac{\alpha_i \delta}{\sigma_i^2} \left(\vx_k - D_{\sigma_i}(\tx_k) \right)$
        \STATE $\vx_{k+1} \gets \text{Prox}_{\delta \fF}(\vz_k)$
    \ENDFOR
\ENDFOR
\end{algorithmic}
\end{algorithm}

\paragraph{DiffPIR}
We compare our algorithm to state-of-the-art methods, including a Diffusion Model. Among existing methods in the field of Diffusion Models, we choose to test DiffPIR~\citep{Zhu_2023_CVPR}. 
DiffPIR is a recent algorithm which makes a connection between PnP and Diffusion Models. 
We use the implementation of the Python library DeepInverse modified to add a time $t_{\text{start}} < T$ such as proposed by the authors~\citep[Section 4.4]{Zhu_2023_CVPR}.
We need to add this parameter $t_{\text{start}}$ as  the denoiser is not trained to generate relevant outputs for highly noisy images (compared to neural networks used for diffusion models), since the model has only been trained for noise with standard deviations in the set $[0, 50]/255$.

\begin{algorithm}
\caption{DiffPIR~\citep{Zhu_2023_CVPR}}\label{alg:DiffPIR}
\begin{algorithmic}[1]
\STATE \textbf{input:} $D_{\sigma}$, $T > 0$, $\vy \in \R^q$, $0 < t_{\text{start} < T}$, $\zeta > 0$,  $(\beta_t)_{0<t<T}$, $\lambda > 0$
\STATE \text{Initialize }$\veps_{t_{\text{start}}} \sim \nN(0, \mI_d)$, \text{pre-calculate }$(\bar{\alpha}_t)_{0<t<T}$, $(\sigma_t)_{0<t<T}$ \text{ and } $(\rho_t)_{0<t<T}$
\STATE $x_{t_{\text{start}}} = \sqrt{\bar{\alpha}_{\text{start}}} \vy + \sqrt{1 -\bar{\alpha}_{\text{start}} } \veps_{t_{\text{start}}}$
\FOR{$t = t_{\text{start}}, t_{\text{start}}-1, \dots, 1$}
    \STATE $x_0^t \gets D_{\sigma_t}(x_t)$
    \STATE $\hat{x_0}^t \gets \text{Prox}_{2 \fF(\cdot, \vy) / \rho_t}(x_0^t)$
    \STATE $\hat{\veps} \gets \left( x_t - \sqrt{\bar{\alpha}_t}  \hat{x_0}^t \right) / \sqrt{1 - \bar{\alpha}_t}$
    \STATE $\veps_t \gets \nN(0, \mI_d)$
    \STATE $x_{t_1} \gets \sqrt{\bar{\alpha}_t}  \hat{x_0}^t + \sqrt{1 - \bar{\alpha}_t} \left( \sqrt{1 - \zeta} \hat{\veps} + \sqrt{\zeta} \veps_t \right)$
\ENDFOR
\end{algorithmic}
\end{algorithm}
The parameter schemes $(\bar{\alpha}_t)_{0<t<T}$, $(\sigma_t)_{0<t<T}$  and $(\rho_t)_{0<t<T}$ are described in~\citep{Zhu_2023_CVPR}. For DiffPIR, we choose the  following parameters for all inverse problems: $\zeta = 0.8$, $T = 1000$, $t_{\text{start}} = 200$ and $\lambda = 0.13$.

\paragraph{PnP SGD}
We compare our algorithm to an other stochastic PnP method, PnP SGD~\citep{laumont2022maximumaposteriori}, which approximates the maximum of the Posterior Law of Problem~\ref{eq:minimization_problem}.
In the implementation of the method, we need to add an other parameter in PnP SGD, $\beta > 0$, to control the power of the additive noise at each iteration and optimize the performance of the algorithm. Adding the $\beta>0$ parameter does not change the analysis of the method. As noticed by~\citet{laumont2022maximumaposteriori}, keeping a fixed step-size allows us to obtain the best performance, so we decided to keep the step-size $\delta > 0$ fixed. We give the pseudo-code of PnP SGD in Algorithm~\ref{alg:PnP_SGD} and the used parameters for deblurring in Table~\ref{table:parameters_deblurring}. We do not succeed to make the method competitive for image inpainting so we do not include in this paper the output of PnP SGD for inpainting.

\begin{algorithm}
\caption{PnP SGD}\label{alg:PnP_SGD}
\begin{algorithmic}[1]
\STATE \textbf{input:} $\vx_0 \in \R^d$, $\sigma > 0$, $\alpha > 0$, $\beta > 0$, $\delta > 0$, $N \in \N$
\FOR{$k = 0, 1, \dots, N-1$}
    \STATE $\vz_k \sim \nN(0, \Id)$
    \STATE $\vx_{k+1} \gets \vx_k -  \delta \nabla \fF(\vx_k) - \frac{\alpha \delta}{\sigma^2} \left(\vx_k - D_{\sigma}(\vx_k) \right) + \beta \delta \vz_k$ 
\ENDFOR
\end{algorithmic}
\end{algorithm}

\paragraph{Computational time}\label{paragraph:computational_time}
On Table~\ref{table:computational_time}, we compare the computational time of various methods. 
For deblurring, SNORE and SNORE Prox are slow compare to other methods. This is due to the number of iterations require at each parameters level to converge, leading to a large number of iteration ($1500$). 
We observe on Figure~\ref{fig:annealing_influence} that with less annealing level, metrics performance are similar. 
However, we observe a qualitative impact of this parameter.
For inpainting, out method outperform RED and RED Prox with a fixed number of iteration. Only DiffPIR remains faster. For this inverse problem, we have observed that a smaller number of iteration ($500$) is sufficient.

Computing all the necessary experiments to generate Table~\ref{table:deblurring} requires 9 hours and 40 minutes on a GPU NVIDIA Quadro RTX 8000.
Similarly, generating Table~\ref{table:inpainting} requires 8 hours of computation on a GPU NVIDIA Quadro RTX 8000. The whole computational resources used for this paper are 17 hours and 40 minutes of computation on a GPU NVIDIA Quadro RTX 8000.

\begin{table}
\centering
\resizebox{0.65\linewidth}{!}{
\begin{tabular}{ |c || c| c|c|c|c| }
\hline
Inverse Problem & RED & RED Prox & SNORE & SNORE Prox & DiffPIR \\
\hline
Deblurring & 5 & 2 & 67 & 41 & 1\\
\hline
Inpainting & 39 & 40 & 40 & 40 & 1\\
\hline
\end{tabular}
}
\caption{Computational time in second (averaged on $4$ images) for various methods and inverse problems on a GPU NVIDIA Quadro RTX 8000.}
\label{table:computational_time}
\end{table}

\subsection{Deblurring}
In this part, we give more details on our experiments for image deblurring. We also discuss the influence of the parameters $m$, $\alpha_{m-1}$ and $\sigma_{m-1}$ on Annealing SNORE outputs (Algorithm~\ref{alg:Annealead_SNORE}).

\paragraph{Parameters setting} In Table~\ref{table:parameters_deblurring}, we give the values of the different parameters used in our experiments.

\begin{table}
\centering
\resizebox{0.65\linewidth}{!}{%
\begin{tabular}{ |c || c| c|c|c|c| }
\hline
Parameters & RED & RED Prox & SNORE & SNORE Prox & PnP SGD \\
\hline
$\alpha$ (motion blur, $\sigma_{\vy} < 20/255$) & 0.1 & 0.2 & & & 0.5 \\
\hline
$\alpha$ (motion blur, $\sigma_{\vy} \ge 20/255$) & 0.1 & 0.3 & & & 0.5 \\
\hline
$\alpha$ (fixed blur, $\sigma_{\vy} < 20/255$) & 0.075 & 0.2 & & & 0.5 \\
\hline
$\alpha$ (fixed blur, $\sigma_{\vy} \ge 20/255$) & 0.075 & 0.3 & & & 0.5 \\
\hline
$\alpha_{0}$ &  &  & 0.1 & 0.1 & \\
\hline
$\alpha_{m-1}$ &  &  & 1 & 1 & \\
\hline
$\sigma / \sigma_{\vy}$ ($\sigma_{\vy} < 20/255$) & 1.8 & 1.4 & & & 1.\\
\hline
$\sigma / \sigma_{\vy}$ ($\sigma_{\vy} \ge 20/255$) & 1.8 & 1.8 & & & 1.\\
\hline
$\sigma_{0} / \sigma_{\vy}$ & & & 1.8 & 1.8 &  \\
\hline
$\sigma_{m-1} / \sigma_{\vy}$ & & & 0.5 & 0.5 & \\
\hline
$\text{max}_{\text{itr}} $ & 100 & 100 & 1500 & 1500 & 1000\\
\hline
$\delta$ & $1/\alpha$ & $1/\alpha$ & 0.1 & 0.1 & 0.1 \\
\hline
$\beta$ & &  &  & & 0.01 \\
\hline
\end{tabular}
}
\caption{Parameters setting for image deblurring for the different implemented methods.}
\label{table:parameters_deblurring}
\end{table}

\paragraph{More deblurring results} For a better qualitative comparison between methods, we present in Figure~\ref{fig:sets_of_result_deblurring} several image deblurring results obtained with various blur kernels. In Figure~\ref{fig:restauration_annealing_parameters}, we present the  parameters evolution during the optimization process for the annealing SNORE algorithm (Algorithm~\ref{alg:Annealead_SNORE}). The annealing levels between $\alpha_0$, $\alpha_{m-1}$ and $\sigma_0$, $\sigma_{m-1}$ are chosen by a linear interpolation.

Next we focus on the choice of the final parameters $(\alpha_{m-1}, \sigma_{m-1})$, because these parameters define the optimization problem from which a critical point is finally computed. The preliminary resolution of optimization problems with parameters $(\alpha_i, \sigma_i)_{i<m-1}$ can be understood as a procedure to compute a relevant initialization for the last annealing level. 
We choose to optimize at different levels $(\alpha_i, \sigma_i)$ in order to compute a good approximation of a critical point at fixed $(\alpha_i, \sigma_i)$ and to fit with the setting of Proposition~\ref{prop:critical_point_accumulate_values}.

\paragraph{Influence of $\alpha_{m-1}$}
On Figure~\ref{fig:restauration_alpha_sensitivity}, we study the influence of the final weighting parameter $\alpha_{m-1}$ on the restoration provided by the SNORE algorithm. One can see that if $\alpha_{m-1}$ is too small, the problem is less regularized so there is a residual noise. On the other hand, if $\alpha_{m-1}$ is too large, the restored image is very flat.

\begin{figure}[!ht]
    \centering
    \includegraphics[width=\textwidth]{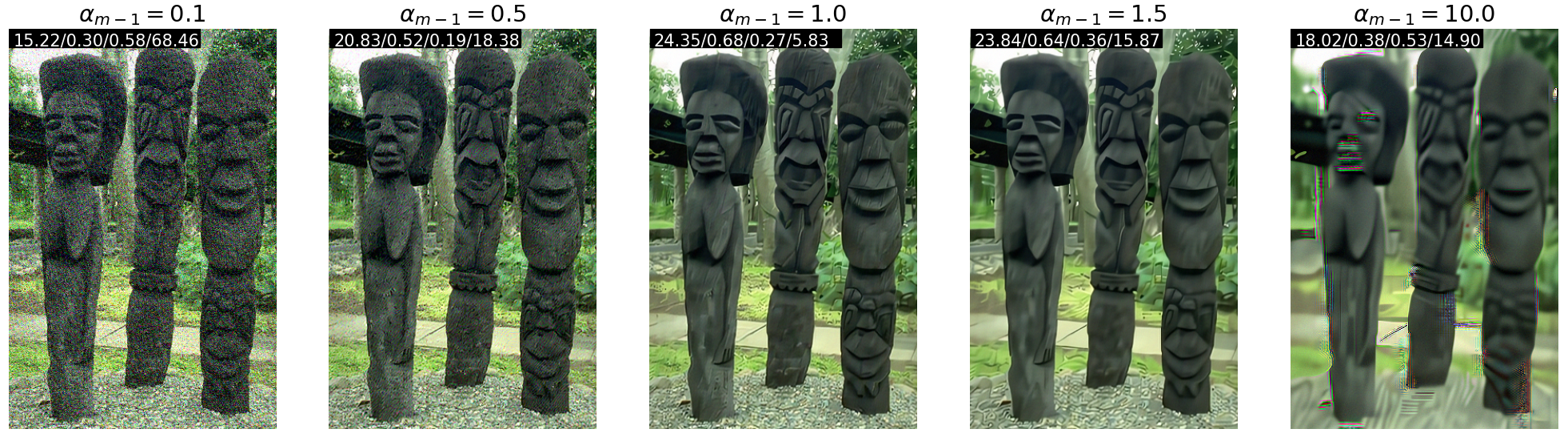}
    \caption{Influence of the parameter $\alpha_{m-1}$ in SNORE algorithm on deblurring with an input noise of level $\sigma_{\vy} = 10 /255$.}
    \label{fig:restauration_alpha_sensitivity}
\end{figure}

\paragraph{Influence of $\sigma_{m-1}$}
On Figure~\ref{fig:restauration_sigma_sensitivity}, we illustrate the influence of the final denoiser parameter $\sigma_{m-1}$ on the restoration obtained with the SNORE algorithm. One can see that if $\sigma_{m-1}$ is too small, the problem is less regularized  and a residual noise is present. On the other side, if $\sigma_{m-1}$ is too large, the restored image is too flat. 

The influence of parameters $\sigma_{m-1}$ and $\alpha_{m-1}$ is therefore similar. However, we observed experimentally that having these two free parameters allows to obtain better restoration results. 

\begin{figure}[!ht]
    \centering
    \includegraphics[width=\textwidth]{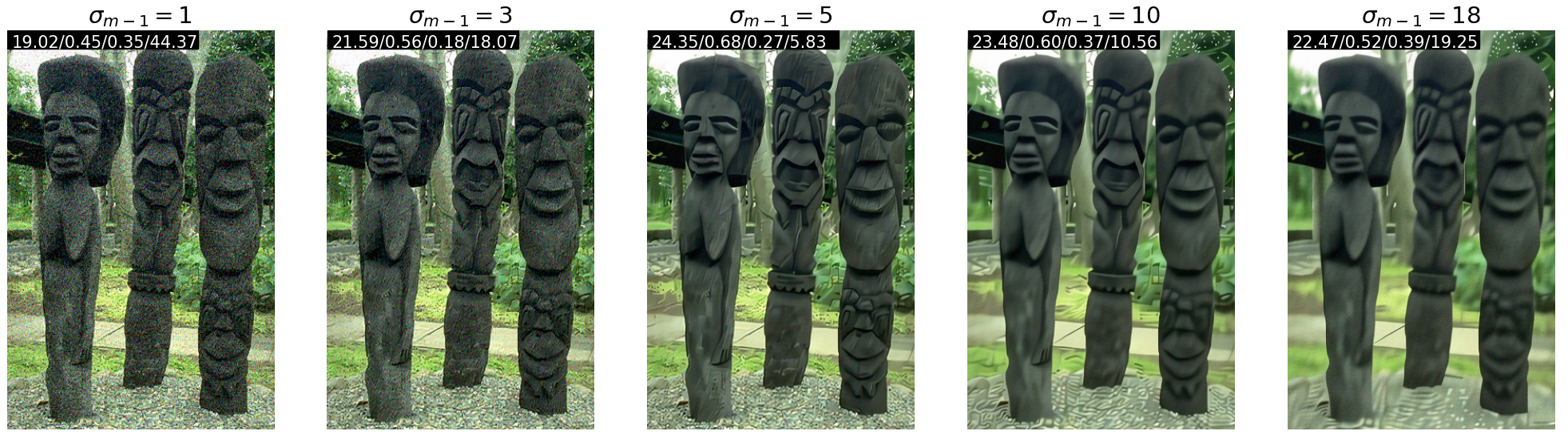}
    \caption{Influence of the parameter $\sigma_{m-1}$ in SNORE algorithm on deblurring with an input noise of level $\sigma_{\vy} = 10 / 255$.}
    \label{fig:restauration_sigma_sensitivity}
\end{figure}

\paragraph{Influence of $m$}
On Figure~\ref{fig:annealing_influence}, we observe the influence of the number of annealing levels $m$ on the quality of the reconstruction. Metrics are not sensitive to this parameter but we observe on images that some artifacts are reduced with additional annealing levels. Our experiments with the SNORE algorithm suggest that with more annealing levels, the algorithm performs better to inverse the degradation and less artifacts are visible. Images of Figure~\ref{fig:annealing_influence} support this claim. Note that local artifacts do not seem to have a significant influence on the metric values.

\paragraph{On the step-size $\delta$}
\citet{laumont2022maximumaposteriori} proposes to use a two-phase gradient-descent, a first one with $\delta_0 > 0$ fixed for finite number of iteration then a second phase with decreasing step-size $\delta_k = \delta_0 / k^{\gamma}$ (in their experiments $\gamma = 0.8$). They ensure that Assumption~\ref{ass:step_size_decreas} is verified. However, they observe that the second phase has no impact on the output of the algorithm. We try the same framework (with various $\gamma \in ]\frac{1}{2}, 1]$) and also observe that the second phase is useless. For efficiency, we choose to only compute the first phase with fixed $\delta > 0$.

\begin{figure}[!ht]
    \centering
    \begin{subfigure}{0.8\textwidth}
        \includegraphics[width=\textwidth]{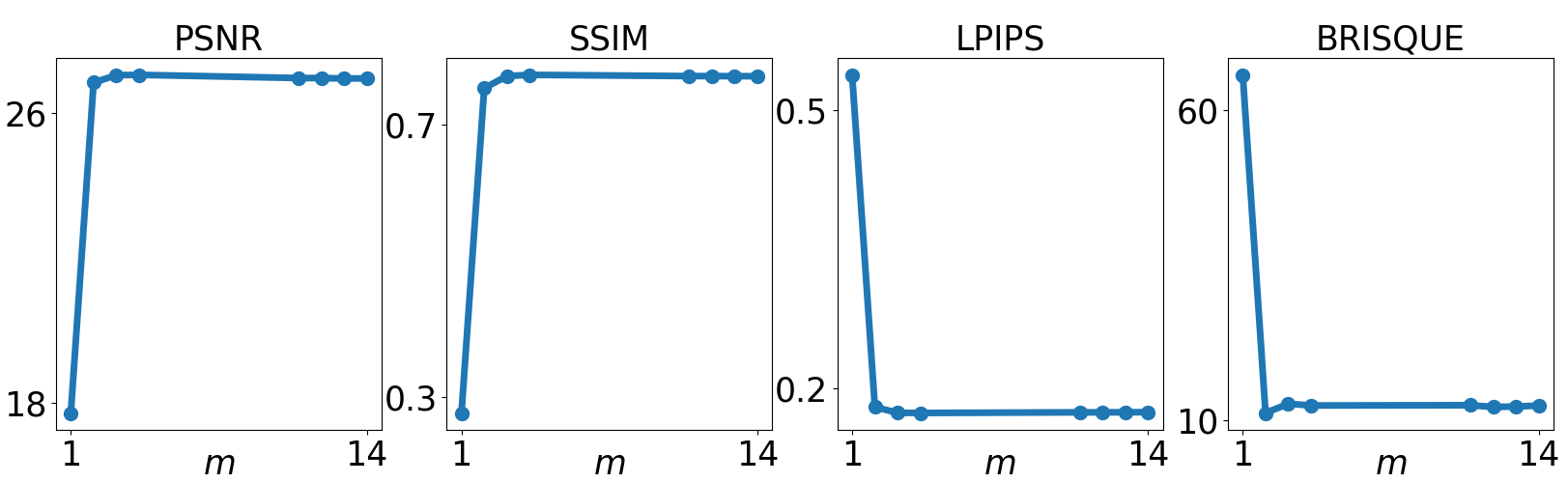}
        \caption{}
        \label{fig:annealing_influence_metrics}
    \end{subfigure}
    \begin{subfigure}{\textwidth}
        \includegraphics[width=\textwidth]{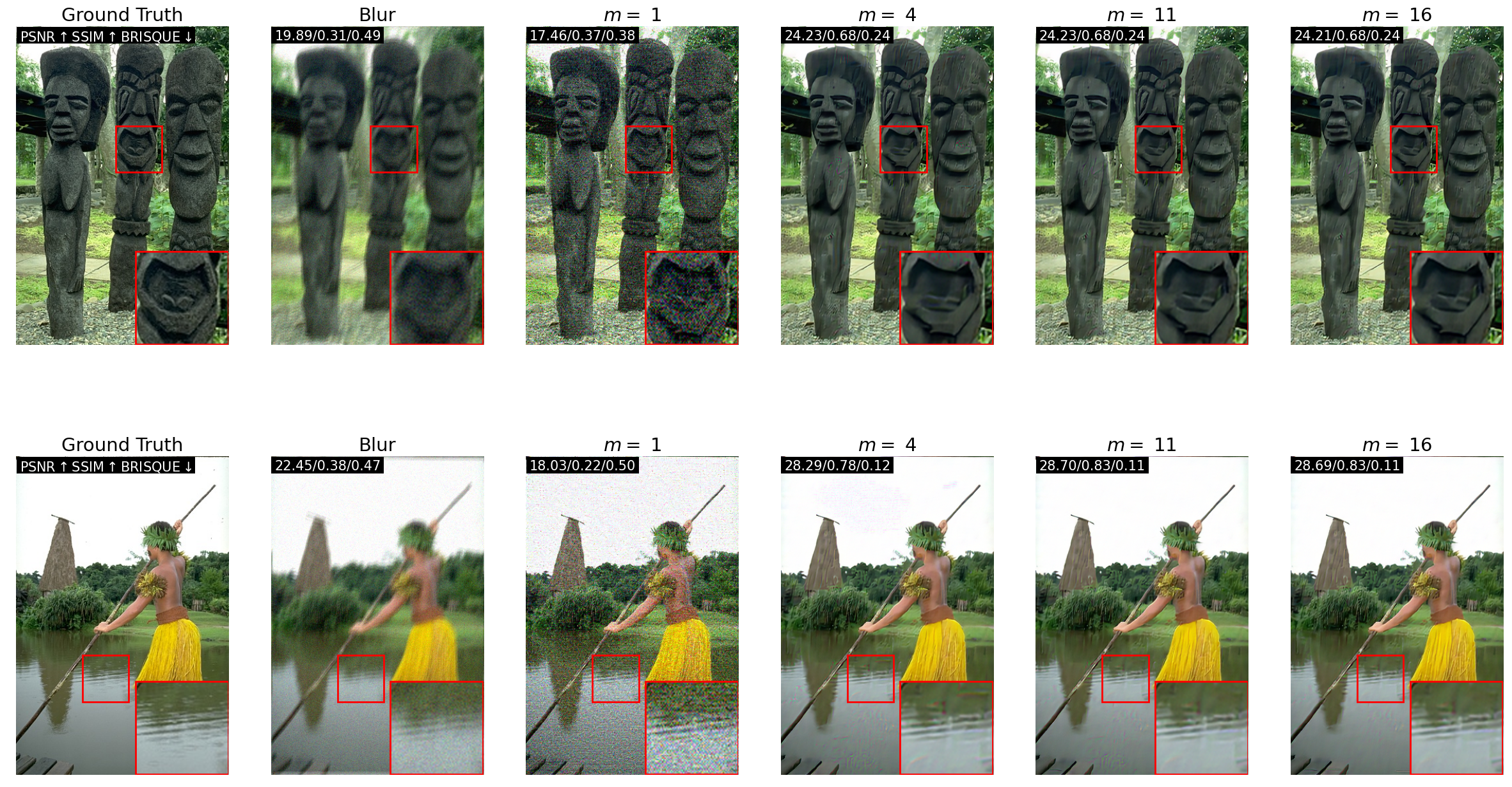}
        \caption{}
        \label{fig:annealing_influence_images}
    \end{subfigure}
    \caption{Influence of the number of annealing levels $m$ on the reconstruction with SNORE algorithm for a motion blur with a noise of standard deviation $\sigma_{\vy} = 10 /255$ and a fixed number of $1500$ iterations. \textit{Figure~\ref{fig:annealing_influence_metrics}:} Metrics evolution with different $m$. One can note that the number of annealing levels $m$ has a low influence on metrics values. \textit{Figure~\ref{fig:annealing_influence_images}:} Reconstructed images with the SNORE algorithm for different numbers of annealing levels $m$. One can note that, the larger $m$ is, the less artifacts are visible. However, compute a larger number of annealing levels impose to compute a larger number of iterations, to converge for each annealing parameters.}
    \label{fig:annealing_influence}
\end{figure}

\subsection{Uncertainty of SNORE}\label{sec:uncertainity_snore}

\paragraph{Seed sensitivity}
On Figure~\ref{fig:restauration_seed_sensitivity}, we illustrate the robustness of  of SNORE to stochasticity, by running the algorithm with different random seed and looking at the standard deviation of the corresponding reconstructions. We observe that our restoration has a low variability and thus a low uncertainty. This is a crucial behavior for the reliability of the algorithm. This experiment suggests that the restoration, and especially the reconstructed structures, is stable and thus reliable.

\begin{figure}[!ht]
    \centering
    \includegraphics[width=\textwidth]{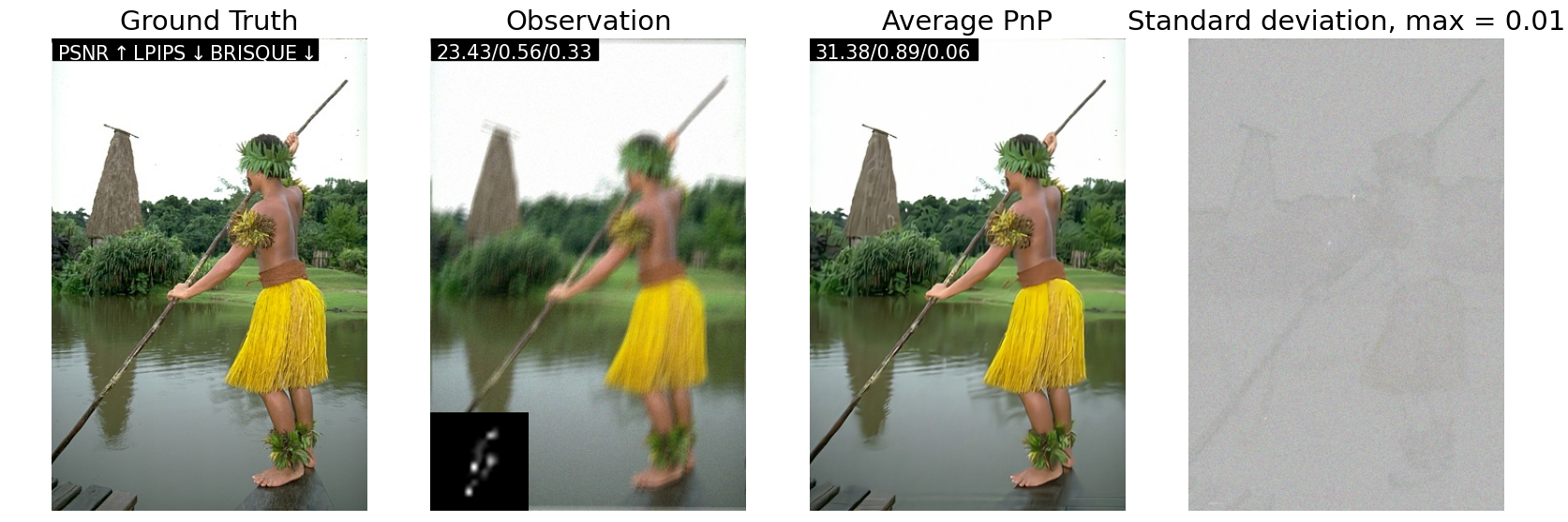}
    \caption{Uncertainty of SNORE to the algorithm randomness. \textit{Rightmost:} The algorithm has been run with $10$ different seeds and the standard deviation of restored images is shown. The blur kernel is shown on the observed images and the input noise level is $\sigma_{\vy} = 5/255$. Note how the algorithm is stable with a small standard deviation of maximum $0.01$ (pixel-values are in $[0,1]$). Especially, structures are stable with a particularly low standard deviation.}
    \label{fig:restauration_seed_sensitivity}
\end{figure}

\paragraph{Initialization sensitivity}
On Figure~\ref{fig:initialization_sensitivity}, three different initialization are shown for a deblurring task. We notice that the SNORE algorithm does not diverge, even with a random initialization. This is a remarkable property since we face a non-convex optimization for which the initialization of the algorithm is crucial. With a pure noise initialization, artifacts are nevertheless present on the restored images. We also observe that the gap between the algorithm run with an oracle initialization or the observation is tight. This observation suggests that the observation is a good initialization for the SNORE algorithm.

\begin{figure}[!ht]
    \centering
    \includegraphics[width=\textwidth]{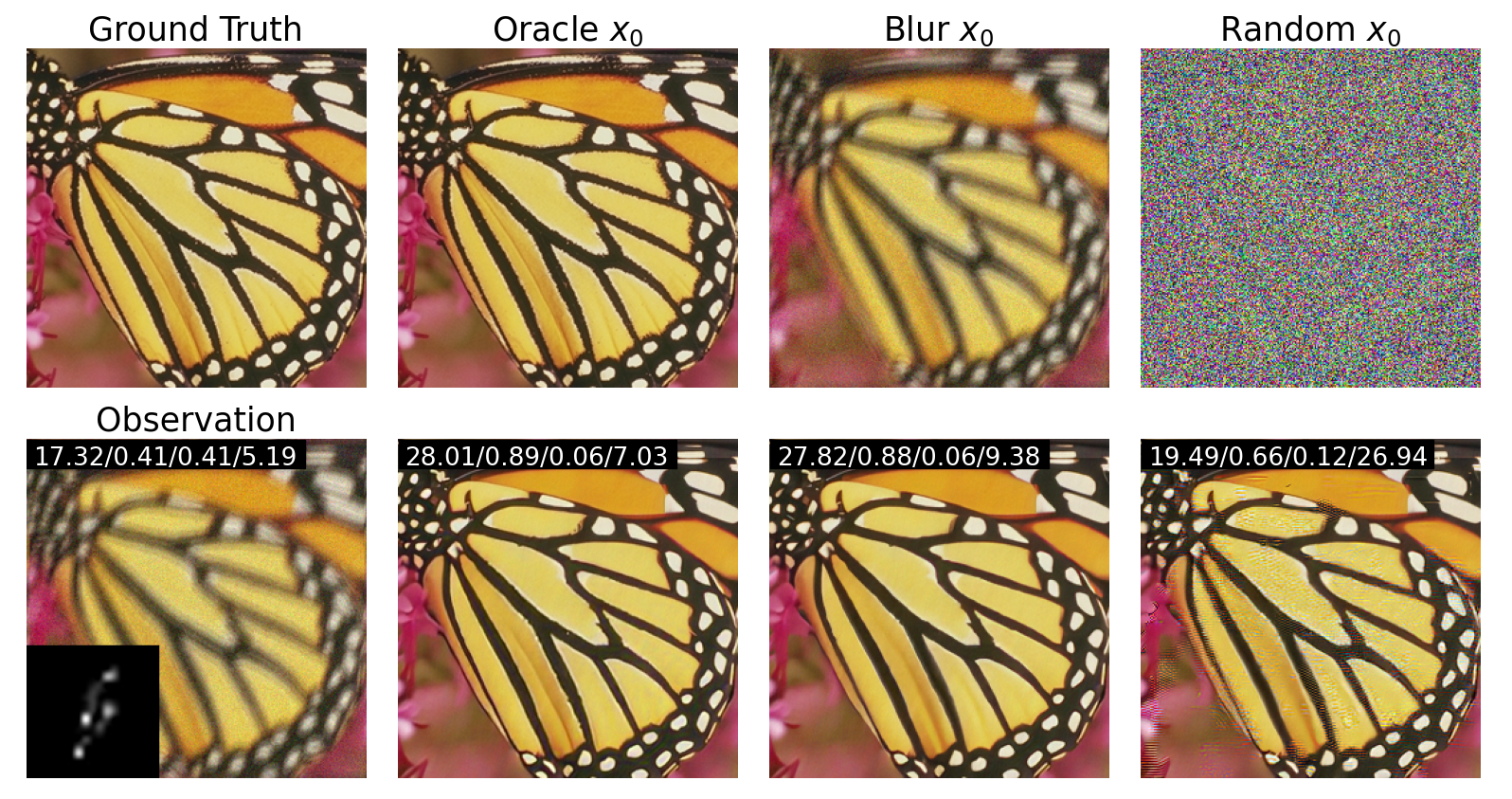}
    \caption{Sensitivity of SNORE to the algorithm initialization. A motion blur kernel and a noise of standard deviation $\sigma_{\vy} = 10/255$ are used to degrade the image. \textit{Top Three Leftmost:} Three initializations are used to start the SNORE algorithm: the ground-truth image (Oracle), the observation and a random image where each pixel is sampled uniformly in $[0,1]$. \textit{Bottom Three Rightmost:} Three corresponding reconstructions. Note how the algorithm succeeds to reconstruct a relevant image even with a random initialization.}
    \label{fig:initialization_sensitivity}
\end{figure}

\subsection{Inpainting}

The inpainting mask is created by sampling a Bernouilli law of success probability $p = 0.5$ for each pixel of the image.

\paragraph{Parameters setting} On Table~\ref{table:parameters_inpainting}, we detail the practical choice of parameters we made.
As suggested by~\citep{hurault2022gradient}, for image inpainting, we start by running the algorithm with a larger value $\sigma = 50/255$ for a number of iterations $n_{\text{init}}$. This allows the algorithm to tackle the ill-posedness of the inpainting task.

\begin{table}
\centering
\resizebox{0.65\linewidth}{!}{%
\begin{tabular}{ |c || c| c|c|c| }
\hline
Parameters & RED & RED Prox & SNORE & SNORE Prox \\
\hline
$\alpha$ & 0.15 & 0.15 & & \\
\hline
$\alpha_{0}$ &  &  & 0.15 & 0.15 \\
\hline
$\alpha_{m-1}$ &  &  & 0.4 & 0.15\\
\hline
$\sigma$ & 10/255 & 10/255 & & \\
\hline
$\sigma_{0}$ & & & 50/255 & 50/255 \\
\hline
$\sigma_{m-1}$ & & & 5/255 & 5/255\\
\hline
$n_{\text{init}} $ & 10 & 100 &  & \\
\hline
$\text{max}_{\text{itr}} $ & 500 & 500 & 500 & 500\\
\hline
$\delta$ (initialization) & $1/\alpha$ & 0.5 & 0.5 & 1\\
\hline
\end{tabular}
}
\caption{Parameters setting for image inpainting for the different implemented methods.}
\label{table:parameters_inpainting}
\end{table}

\begin{figure}[!ht]
    \centering
    \includegraphics[width=\textwidth]{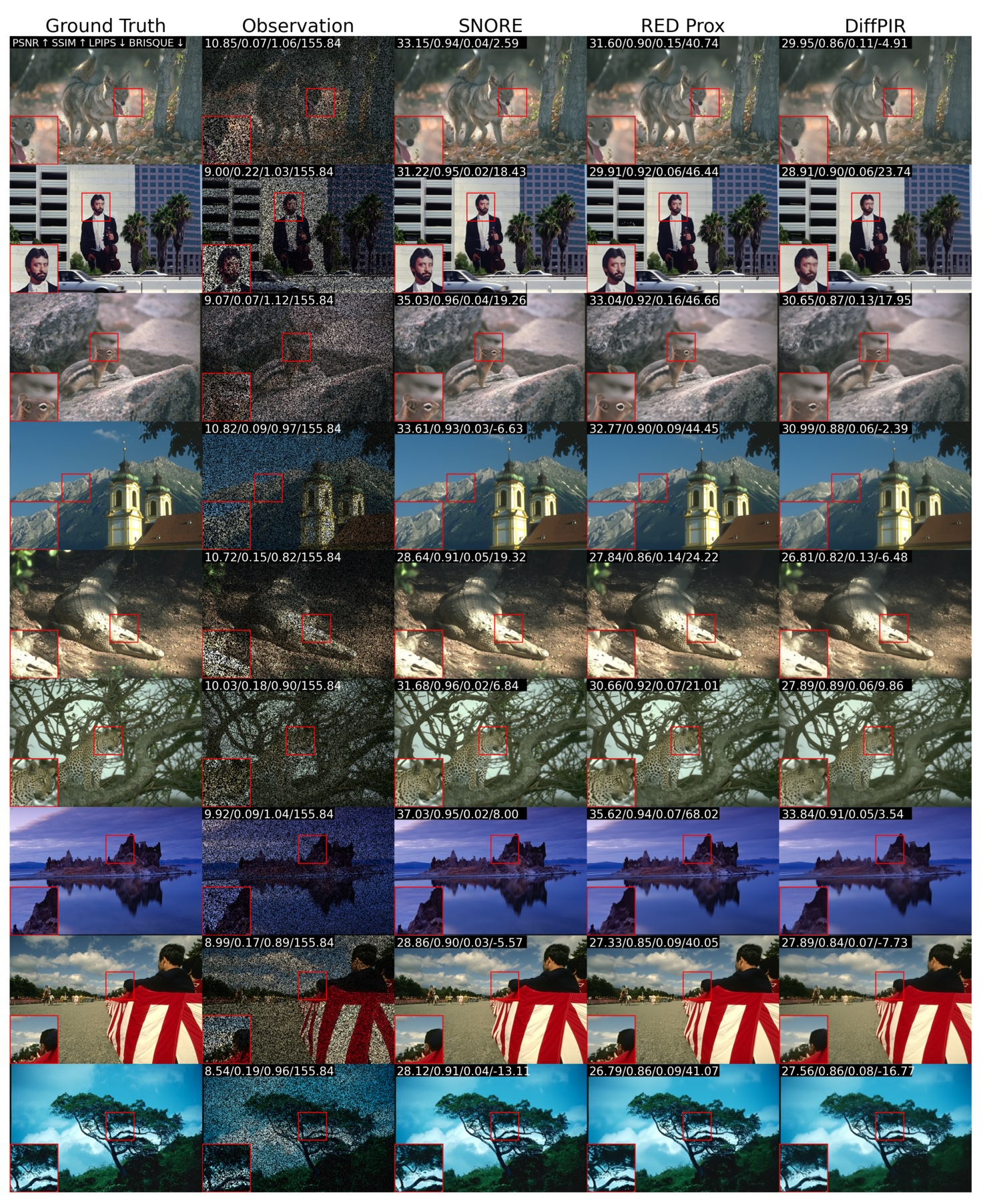}
    \caption{Restorations obtained with SNORE, RED and DiffPIR algorithms for various images from the dataset CBSD68 on the inpainting task with a random mask with a proportion $0.5$ of masked pixels.}
    \label{fig:sets_of_result_inpainting}
\end{figure}

\paragraph{More results} On Figure~\ref{fig:sets_of_result_inpainting}, we present various results of image inpainting for a better qualitative comparison between methods.

\begin{figure}[!ht]
    \centering
    \includegraphics[width=\textwidth]{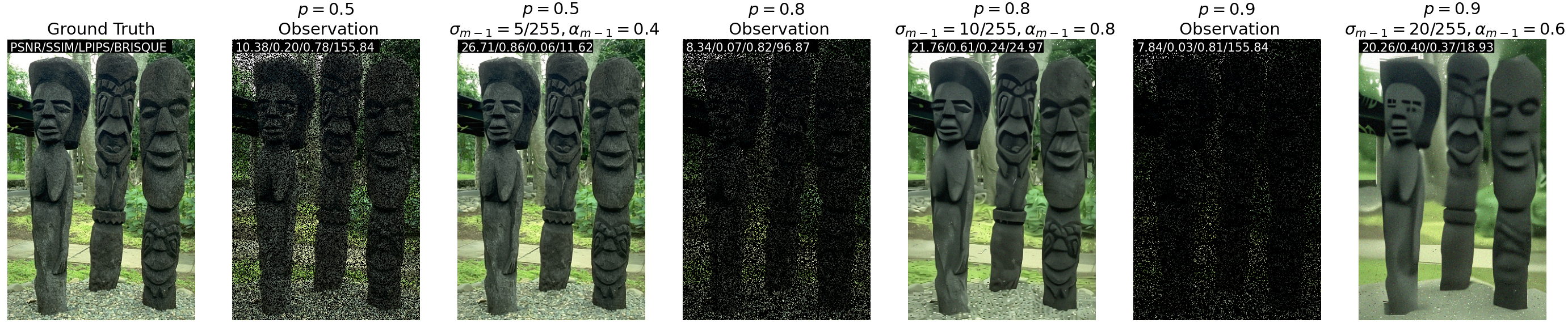}
    \caption{Restorations obtained with SNORE for inpainting, with a random mask of proportion $p$ of masked pixels, on one image from the dataset CBSD68. The last annealing parameters $(\alpha_{m-1}, \sigma_{m-1})$ are given for each restored image.}
    \label{fig:inpainting_various_masks}
\end{figure}

On Figure~\ref{fig:inpainting_various_masks}, we provide more results of SNORE algorithm on various inpainting problems (with of proportion $p \ge 0.5$ of masked pixels). As expected, we observe that the quality of the restoration decrease with the proportion of missing pixels.

\begin{figure}[!ht]
    \centering
    \includegraphics[width=\textwidth]{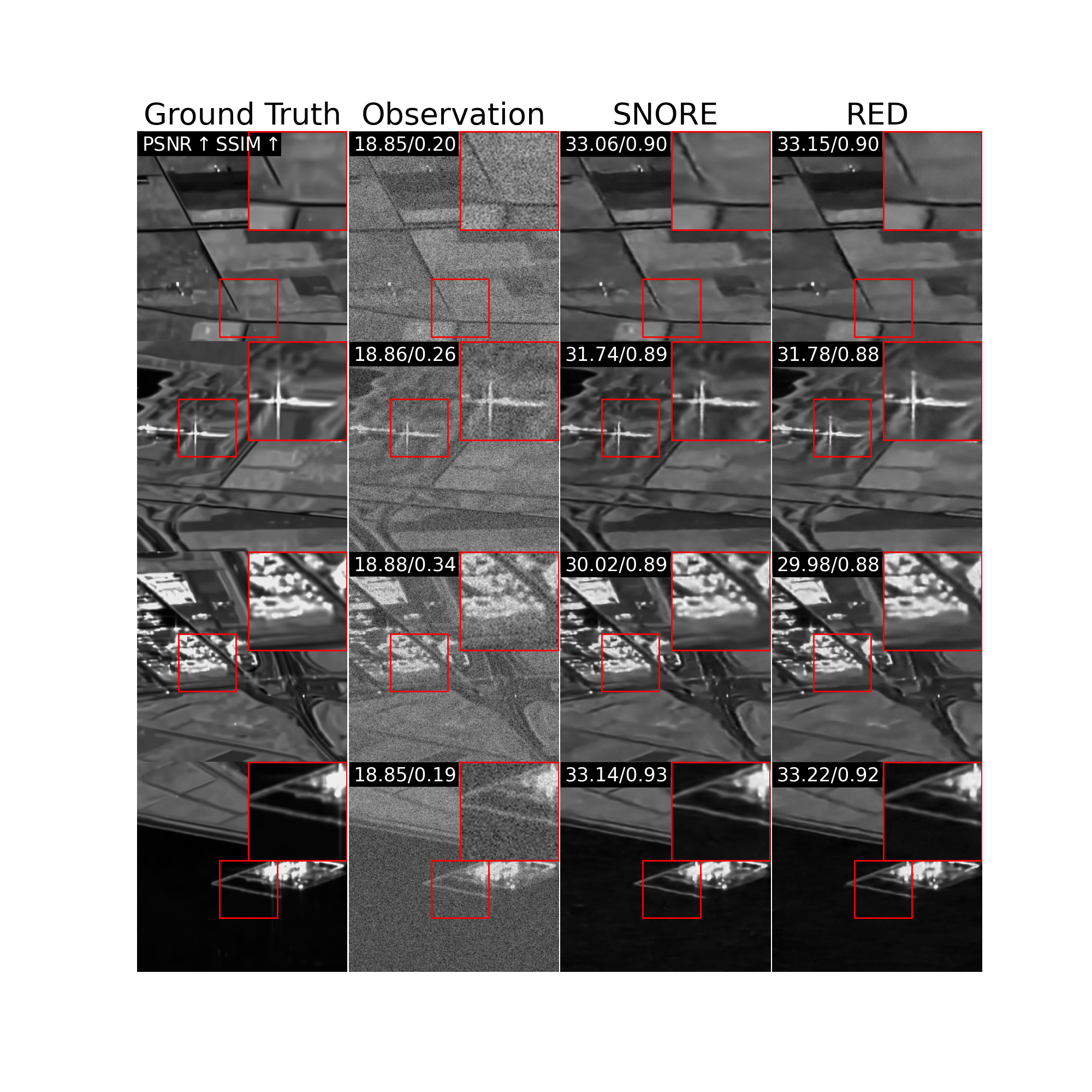}
    \caption{Image despeckling ($L = 1$) on various SAR images with RED and SNORE algorithm with a GS-denoiser trained on SAR images.
    }
    \label{fig:despeckle}
\end{figure}

\subsection{Super-resolution}
For image super-resolution, the observation $\vy \in \R^q$ is the low-resolution version of $\vx \in \R^d$, obtained by $\vy = \mS \mH \vx + \vn$, where $\mH$ is an anti-aliasing blur kernel and $\vn \sim \nN(0, \sigma_{\vy})$. $\mS$ is the standard down-sampling matrix with the super-resolution factor $s_f$. The data-fidelity is given by $\fF(\vx, \vy) = \frac{1}{2 \sigma_{\vy}}\|\mS \mH \vx - \vy \|^2$ and its proximal operator~\citep{zhao2016fast} by $\text{Prox}_{\delta \fF}(\vz) = \hat{\vz} - \frac{1}{s_f^2} \mF^\star \bar{\mLambda}^\star \left( \mI_q + \frac{\delta}{s_f^2} \bar{\mLambda} \bar{\mLambda}^\star \right)^{-1} \bar{\mLambda} \mF \hat{\vz}$, where $\hat{\vz} = \delta \mH^T \mS^T \vy + \vz$ and $\bar{\mLambda} = \left( \mLambda_1, \dots \mLambda_{s_f^2} \right) \in \R^{q\times d}$, with $\mLambda = \text{diag}\left(\mLambda_1, \dots \mLambda_{s_f^2}\right)$ a block-diagonal decomposition according to a $s_f \times s_f$ paving of the Fourier domain.

Thanks to the previous expression of the Proximal operator, SNORE Prox and RED Prox can be computed for super-resolution. On Figure~\ref{fig:SR_results}, we give qualitative results for image super-resolutions on a kernel of blur and various images. Note that SNORE Prox produce better perceptual results (LPIPS, BRISQUE) by creating local texture. However, SNORE Prox is worst in distortion metrics (PSNR, SSIM) as for image deblurring. These are preliminaries experiments and we leave for futur work to adapt each methods for image super-resolution and make a quantitative evaluation.

\begin{figure}[!ht]
    \centering
    \includegraphics[width=\textwidth]{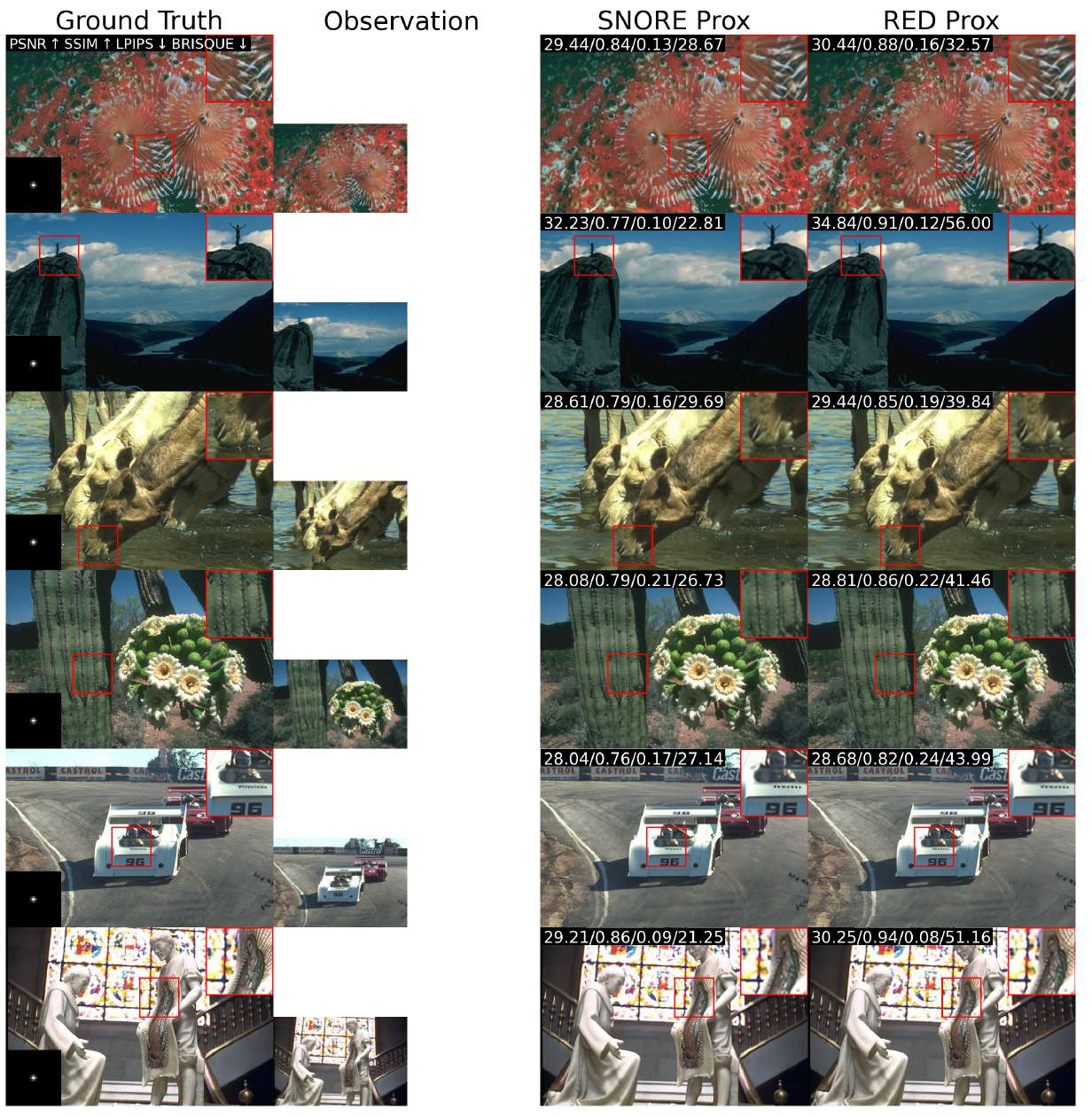}
    \caption{Image super-resolution obtained with SNORE Prox and RED Prox, with a super-resolution factor $s_f = 2$ and an input noise $\sigma_{\vy} = 5/255$, on various images from the dataset CBSD68.}
    \label{fig:SR_results}
\end{figure}

\paragraph{Parameters setting}
On Table~\ref{table:parameters_SR}, we details the practical choice of parameters we made. Note that the number of iterations of the algorithm are the same for image super-resolution.

\begin{table}
\centering
\resizebox{0.45\linewidth}{!}{%
\begin{tabular}{ |c || c| c|c|c| }
\hline
Parameters & RED Prox & SNORE Prox \\
\hline
$\alpha$ & 0.065 & \\
\hline
$\alpha_{0}$ & & 0.02 \\
\hline
$\alpha_{m-1}$ & & 0.3\\
\hline
$\sigma/\sigma_{\vy}$ & 2 & \\
\hline
$\sigma_{0}/\sigma_{\vy}$ & & 4 \\
\hline
$\sigma_{m-1}/\sigma_{\vy}$ & & 2\\
\hline
$\text{max}_{\text{itr}} $ & 400 & 400 \\
\hline
$\delta$ (initialization) & $1/\alpha$ & 1\\
\hline
\end{tabular}
}
\caption{Parameters setting for image super-resolution for the different implemented methods.}
\label{table:parameters_SR}
\end{table}

\subsection{SAR despeckling}
For Synthetic Aperture Radar (SAR) despeckling, we consider the Goodman's model~\citep{Goodman1976SomeFP} in which we aim at recovering  the underlying reflexivity $\vR \in \R_+^d$ from the observed intensity $\vI \in \R_+^d$, that is a noisy version of $\vR$ with a multiplicative gamma noise, $\vI = \vN \vR$, where $\vN \sim \Gamma(1, L)$ and $L>0$ is called the number of looks~\citep{Goodman1976SomeFP}. By tacking the log of the previous model and denoting $\vy = \log{\vI}$ and $\vx = \log{\vR}$, we turn the multiplicative noise into an additive noise, $\vy = \vx + \vn$, where $\vn$ is following the Fisher-Tippett distribution.

To solve this problem, we can solve the following variation problem~\citep{Deledalle_2017}
\begin{align*}
    \argmin_{\vx \in \R^d}{-\log{p(\vy|\vx)} + \alpha \fR(\vx)},
\end{align*}
where $-\log{p(\vy|\vx)} = L \sum_{k=1}^d x_k + e^{y_k - x_k} + \text{Cst.}$ ($x_k$ is the value of the $k$-pixel of $\vx$) and $\fR$ the regularization.

Working in the log-domain allow us to solve an unconstrained optimization with a convex data-fidelity term. However the gradient of this data-fidelity term is not Lipschitz. By solving this problem with the SNORE regularization, our theoretical analysis still holds (the data-fidelity is $C^{\infty}$). However, the boundedness hypothesis might be harder to verify because the objective function is not necessarily coercive (see Appendix G for more details).

To solve experimentally this problem, we have trained a GS-denoiser~\citep{hurault2022gradient} on the SAR-speckle free dataset developed by~\citet{Dalsasso_2020}. This denoiser is trained to remove additive gaussian noise to SAR images for $\sigma \in [0,50]$ with the same traning parameters than proposed by~\citet{hurault2022gradient}.

We use RED and SNORE algorithm to despeckle images with the parameters setting details in Table~\ref{table:parameters_despeckling}. On Figure~\ref{fig:despeckle}, we show SNORE and RED algorithm qualitative performance on various SAR images. We only compute the PSNR and SSIM metrics because the LPIPS and the BRISQUE are designed for color images. We can notice that SNORE and RED succeed to restore good quality images for the non-standard data-fidelity term of image despeckling.

\begin{table}
\centering
\resizebox{0.4\linewidth}{!}{%
\begin{tabular}{ |c || c| c|c|c| }
\hline
Parameters & RED & SNORE \\
\hline
$\alpha$ & 80 &  \\
\hline
$\alpha_{0} = \alpha_{m-1}$ &  & 80\\
\hline
$\sigma \times 255$ & 10 & \\
\hline
$\sigma_{0} \times 255$ &  & 30 \\
\hline
$\sigma_{m-1} \times 255$ &  & 10\\
\hline
$\text{max}_{\text{itr}} $ & 100 & 100\\
\hline
$\delta$ & 0.01 &  0.01\\
\hline
\end{tabular}
}
\caption{Parameters setting for image despeckling for the different implemented methods.}
\label{table:parameters_despeckling}
\end{table}

\section{SNORE applies the denoiser on its training domain}\label{sec:estimation_noise}
One motivation to use SNORE is to force the denoiser to be applied on its training domain. To do so, at each iteration, the denoiser $D_{\sigma}$ is not applied to the previous iteration $\vx_k$ (as in RED) but to a noisy version of $\vx_k$, $\tx_k = \vx_k + \sigma \epsilon$, where $\epsilon \sim \nN(0, \Id)$. In this noisy version $\tx_k$, the input noise level is exactly the noise level of the denoiser.

However, in pratice, there might be residual noise in the iteration $\vx_k$, so the noise level of $\tx_k$ might be higher than $\sigma$. In order to verify this experimentally, we use a robust wavelet-based noise estimator~\citep{Donoho1994} implemented in the library library scikit-image (aka skimage) as the function estimate\_sigma(). We define $\hat{\sigma}(\vx)$ the noise estimation of the image $\vx$.

\begin{figure}[!ht]
    \centering
    \includegraphics[width=0.5\textwidth]{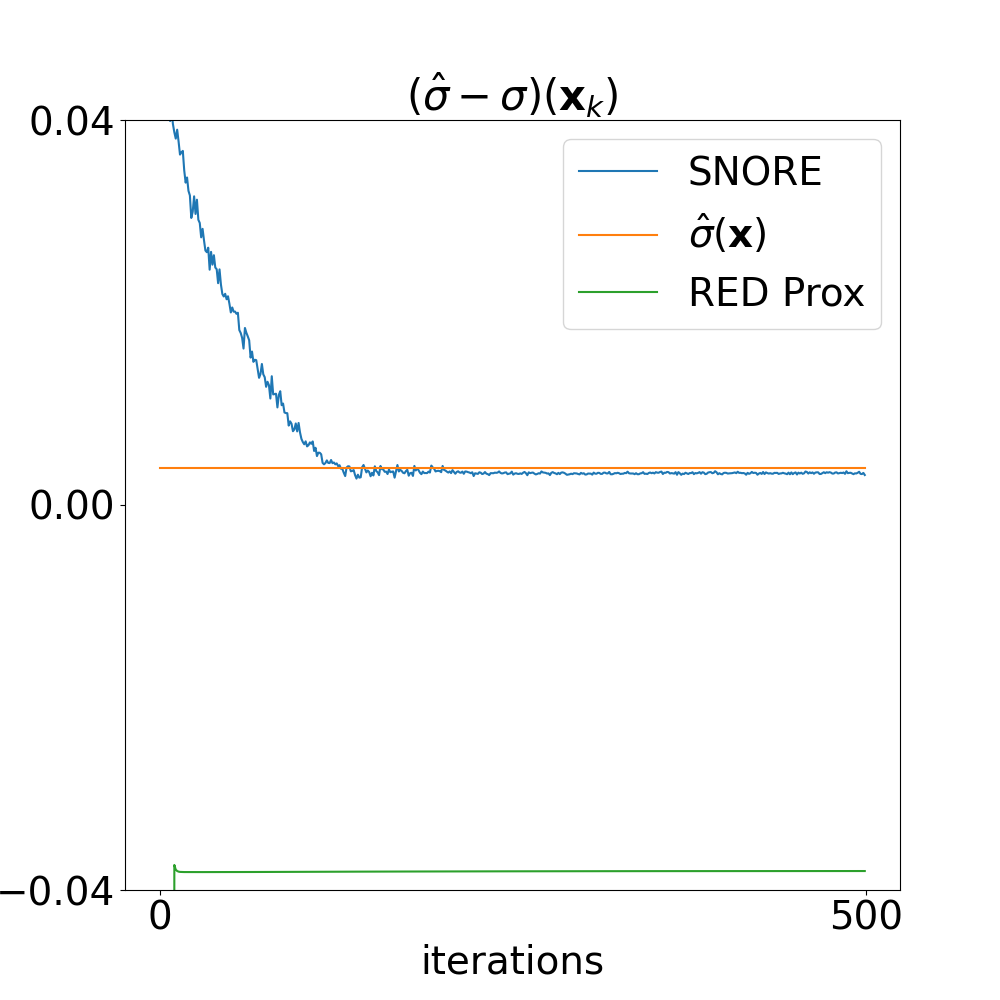}
    \caption{Difference between the estimated noise $\hat{\sigma}$ and the denoiser parameter $\sigma$ for iterations of SNORE and RED algorithms for inpainting (with a proportion $p = 0.5$ of masked pixels) on one image of the dataset CBSD68.
    In orange, the estimated noise on the clean image.
    }
    \label{fig:noise_estimation}
\end{figure}

On Figure~\ref{fig:noise_estimation}, we see that the annealing level are not visible for SNORE, which suggests that the algorithm is well adapt to the noise level. Moreover, the residual noise is decreasing and converging to the natural noise shift for SNORE. After a first phase of removing the noise of the initialization, SNORE algorithm effectively applies the denoiser to an image with the right noise level. For RED, the shift between the input noise level and the noise level of the denoiser is still large at convergence. This suggests that the denoiser is applied on an image out-of-distribution (without the right level of noise) at convergence.

\section{Discussion on the boundedness of $(\vx_k)_{k \in \N}$}\label{sec:boundness_sequence}

In Proposition~\ref{prop:convergence_unbiased}-\ref{prop:biais_convergence}, the convergence of Algorithm~\ref{alg:Average_PnP} is studied almost surely on $\Lambda_{\cK}$, the set of realizations where $(\vx_k)_{k \in \N}$ is bounded in a compact $\cK$. In what follows, we name that the \textit{boundedness assumption}. That kind of assumption is standard in stochastic gradient descent analysis with non-convex objective functions~\citep{Benaim1999,tadic2017asymptotic}.
However, one can remark that in similar non-stochastic Plug-and-Play methods~\citep[Appendix D]{hurault2022gradient} or in posterior sampling algorithms~\citep{Laumont_2022_pnpula, renaud2023plugandplay}, a projection or a penalty term can be added to guarantee a bounded sequence. 
Unfortunately, to our knowledge, a simple projected stochastic gradient descent step is not simple to analyze.

\citet{davis2018stochastic} prove convergence of a projected stochastic gradient descent algorithm, but the convergence analysis relies on a random choice of the ending step. 
\citet{ghadimi2013stochastic} develop a similar approach. The random choice of the ending step is not satisfying in our setting as we want to fix the number of iterations for a fair comparison with deterministic methods.

\begin{algorithm}
\caption{Randomly Projected SNORE}\label{alg:Projected_Average_PnP}
\begin{algorithmic}[1]
\STATE \textbf{input:} $\vx_0 \in \R^d$, $m \in \N$, $\delta > 0$, $\sigma > 0$, $\alpha > 0$, $N \in \N$, $\beta_0 > 0$, $\lambda_0 = 0$.
\FOR{$k = 0, 1, \dots, N-1$}
        \STATE $\delta_k \gets \frac{\delta}{k+1}$
        \STATE $\veps \gets \mathcal{N}(0, \mI_d)$
        \STATE $\tx_k \gets \vx_k + \sigma \veps$
        \STATE $\vz_{k+1} \gets \vx_k - \delta_k \nabla \fF(\vx_k, \vy) - \frac{\alpha \delta_k}{\sigma^2} \left(\vx_k - D_{\sigma}(\tx_k) \right)$
        \STATE $\vx_{k+1} \gets \vz_{k+1} \mathbf{1}_{\|\vz_{k+1}\| \le \beta_{\lambda_k}} + \vx_0 \mathbf{1}_{\|\vz_{k+1}\| > \beta_{\lambda_k}}$
        \STATE $\lambda_{k+1} \gets \lambda_k + \mathbf{1}_{\|\vz_{k+1}\| > \beta_{\lambda_k}}$
\ENDFOR
\end{algorithmic}
\end{algorithm}

Another way to ensure convergence is to use the random projected stochastic gradient descent algorithm proposed by~\cite{Nurminski1973}. As detailed in Algorithm~\ref{alg:Projected_Average_PnP}, at each iteration, this algorithm realizes a projection onto a ball parameterized by an increasing sequence of positive real numbers  $(\beta_n)_{n \in \mathbf{N}}$. This algorithm is proved to converge without the boundedness assumption. \citet[Theorem A1.1.]{tadic2017asymptotic} explored this perspective and demonstrated that the iterates of this algorithm are bounded. However, in our context, \citep[Assumption 1.2.]{tadic2017asymptotic} is difficult to verify. In fact, our demonstrations (Proof of Proposition~\ref{prop:convergence_unbiased} in Appendix~\ref{sec:proof_unbiased}) rely on an upper bound  $\eE\left(\|\xi_k\|^2 \right)$ which is obtained thanks to the boundedness assumption. In our context, Algorithm~\ref{alg:Projected_Average_PnP} has not been proved to converge.

Finally, we leave for future work the exploration of a strategy to demonstrate the convergence of Algorithm~\ref{alg:Average_PnP} without any boundedness assumption.

\end{document}